\documentclass[A4,10pt]{article}
\usepackage[margin=1in,dvips]{geometry}
\usepackage{graphicx,cancel}
\usepackage{graphics}
\usepackage{color, soul}
\usepackage{amsmath, amsthm, slashed}
\usepackage{amssymb}
\usepackage{marvosym}
\usepackage{rotating}
\usepackage{mathrsfs}
\usepackage{epsfig}
\usepackage{comment}
\usepackage[affil-it]{authblk}
\usepackage{rotating}
\usepackage{graphicx}
\usepackage{accents}
\usepackage{harmony, slashed}
\usepackage{tikz}
\usetikzlibrary{patterns}
\usepackage{scalerel}

\usepackage{enumerate}
\usepackage{stackengine,wasysym}

\newtheorem{definition}{Definition}[section]
\newtheorem{theorem}{Theorem}[section]

\newtheorem*{conjecture*}{Conjecture}
\newtheorem{question}{Question}
\newtheorem{corollary}{Corollary}[section]
\newtheorem*{theorem*}{Theorem}
\newtheorem*{corollary*}{Corollary}
\newtheorem{proposition}{Proposition}[section]
\newtheorem{lemma}{Lemma}[section]
\newtheorem{remark}{Remark}[section]

\DeclareMathAlphabet\mathbfcal{OMS}{cmsy}{b}{n}

\title{Twisted Self-Similarity and the Einstein Vacuum Equations}

\date\today
\author[1,2]{Yakov Shlapentokh-Rothman}

\affil[1]{\small University of Toronto, Department of Mathematics, 40 St.~George~Street, Toronto, ON, Canada\vskip.2pc \ }
\begin{document}
\affil[2]{\small University of Toronto Mississauga, Department of Mathematical and Computational Sciences, 3359 Mississauga Road, Mississauga, ON, Canada\vskip.2pc \ }

\maketitle
\begin{abstract}In the previous works~\cite{nakedone,nakedinterior} we have introduced a new type of self-similarity for the Einstein vacuum equations characterized by the fact that the homothetic vector field may be spacelike on the past light cone of the singularity. In this work we give a systematic treatment of this new self-similarity. In particular, we provide geometric characterizations of spacetimes admitting the new symmetry and show the existence and uniqueness of formal expansions around the past null cone of the singularity which may be considered analogues of the well-known Fefferman--Graham expansions. In combination with results from~\cite{nakedone} our analysis will show that the twisted self-similar solutions are sufficiently general to describe all possible asymptotic behaviors for spacetimes in the small data regime which are self-similar and whose homothetic vector field is everywhere spacelike on an initial spacelike hypersurface. We present an application of this later fact to the understanding of the global structure of Fefferman--Graham spacetimes and the naked singularities of~\cite{nakedone,nakedinterior}. Lastly, we observe that by an amalgamation of the techniques from~\cite{scaleinvariant,nakedone}, one may associate true solutions to the Einstein vacuum equations to each of our formal expansions in a suitable region of spacetime.
\end{abstract}
\tableofcontents
\section{Introduction}
We say that a $3+1$ dimensional Lorentzian manifold $\left(\mathcal{M},g\right)$ is a \emph{self-similar} solution to the Einstein vacuum equations if 
\begin{equation}\label{eve}
{\rm Ric}\left(g\right) = 0,
\end{equation}
and if it possesses a \emph{homothetic} vector field $K$, that is,
\begin{equation}\label{3okpo982}
\mathcal{L}_Kg = 2g.
\end{equation}
It is moreover straightforward to adapt this definition to suitable Einstein matter systems.

Originally motivated by the longtime use of self-similar variables in fluid mechanics, a substantial literature developed, initiated by Cahill and Taub~\cite{cahilltaub}, concerning self-similar solutions to Einstein matter systems with various additional symmetries imposed, for example, spherical symmetry. Two of the most important themes that emerged are the construction of naked singularities: see, for example,~\cite{dustnaked,nakedfluidoripiran,nakedfluidjoshidwivedi,nakedeulereinstin} for Einstein fluid systems and~\cite{ChristNaked} for the Einstein-scalar field system, and the connection of self-similar solutions with \emph{critical phenomena},\footnote{It should be noted, however, that in many cases it is \underline{discrete} self-similarity which is truly relevant for critical phenomena.} see~\cite{chop,critrotate,critsurv}. (So-called scale invariant spacetimes also play an important role in Christodoulou's theory of BV-solutions~\cite{ChristBV}.)

In this paper, however, we will be interested in the study of self-similar solutions to the Einstein vacuum equations with no additional symmetry imposed. Though some general aspects of these spacetimes had been previously studied (see, for example,~\cite{eardleyss}) the first work which provided an avenue for a systematic study of a class of self-similar solutions (with no additional symmetry) was the Ambient Metric construction of Fefferman--Graham~\cite{FG1,FG2}. An important geometric property of these Ambient Metrics is that they possess a dilation invariant null hypersurface whose normal vector coincides with the homothetic vector field $K$. For the purposes of this paper, the key result from the work of Fefferman--Graham is the existence and uniqueness of suitable formal power series expansions corresponding to these metrics. (Later, in the work~\cite{scaleinvariant}, we showed that, in appropriate regions of the spacetime, these power series expansions correspond to true solutions of the Einstein vacuum equations.)

While the Fefferman--Graham theory has been very successful, there is no reason to expect the vector field $K$ of a self-similar solution to satisfy the property that $\left\{|K| = 0\right\}$ is a null hypersurface.\footnote{For example, one may check that any vacuum Kasner spacetime  does not satisfy this property.} A  less rigid assumption is to instead assume the existence of a null hypersurface where $K$ is only required to be \emph{tangent}, but may be either null or spacelike depending on the precise location.\footnote{In this paper we will not systematically explore necessary and sufficient conditions for a given self-similar spacetime to admit such a hypersurface. However, the reader may check that one sufficient condition is for there to exist a point $P$ in a suitably regular conformal compactification of (a portion of) the spacetime such that there is a non-trivial past null cone of $P$ and so that the homothetic vector field $K$ extends to $P$  in a sufficiently regular fashion where it then vanishes.} We will use the term \emph{twisted self-similarity} to refer to this new type of self-similarity, where the ``twisting'' refers to the fact that the flow of the vector field $K$ is allowed to wind around the null generators of the null hypersurface.

Particular instances of this new type of self-similarity have played an important role in the works~\cite{nakedone,nakedinterior} where we gave the first construction of naked singularities for the Einstein vacuum equations. The goal of the current paper is to give a systematic treatment of this new class of self-similar solutions in the \underline{small data regime}. In particular we will carry out the following:
\begin{itemize}
	\item Provide explicit geometric characterizations of spacetimes admitting the new twisted self-similarity.
	\item Establish the existence and uniqueness of suitable formal expansions for twisted self-similar solutions.
\end{itemize}
Having carried out these two steps we will furthermore observe that, in appropriate regions of spacetime, the arguments from~\cite{scaleinvariant,nakedone} \emph{mutatis mutandis} show that there always exist true solutions to the Einstein vacuum equations corresponding to our formal expansions and that our formal expansions describe the possible asymptotic behaviors in the entire small data regime for solutions arising from spacelike self-similar data.

 As a particular application of this final point above, we may use these twisted self-similar expansions to describe certain aspects of the global structure of Fefferman--Graham spacetimes (in the small data regime and in $3+1$ dimensions) and also the naked singularity exteriors constructed in~\cite{nakedone}.

Lastly, we note that for simplicity of the presentation and because these cover the primary situations of interest, we have restricted our discussions in this paper to spacetimes in $3+1$ dimensions and which are diffeomorphic to $\mathbb{R}^2 \times \mathbb{S}^2$. However, neither of these dimension/topological restrictions are in fact necessary, and, in view of the potential applications to understanding the global structure of Fefferman--Graham spacetimes, it may be of interest to develop the analogues of our twisted self-similar theory in higher dimensions for spacetimes with topology $\mathbb{R}^2\times \mathcal{S}$, for $\mathcal{S}$ a $d$-dimensional compact manifold ($d \geq 2$). 

In the remaining sections of the introduction we will give an overview of the main results established in the paper. Throughout these remaining sections, when we discuss various definitions and results in an informal fashion we will also indicate where later in the paper to find a precise statement.
\subsection{Canonical Gauges}
In this introductory section we will skip the precise definition that we shall use for a twisted self-similar spacetime (see Definition~\ref{ijo93u92}); instead we will just state the two most important facts.
\begin{enumerate}
	\item There exists a null hypersurface $\mathcal{H}$ where the homothetic vector field $K$ is tangent and has complete non-compact orbits. We moreover have that $\mathcal{H}$ may be covered by coordinates  $\left(s,\theta^A\right) \in (-\infty,0) \times \mathbb{S}^2$ where $K|_{\mathcal{H}} = \partial_s$, and the submanifolds of constant $s$ are spacelike.
	\item For our later applications, we must allow for the possibility that a twisted self-similar spacetime has limited regularity on $\mathcal{H}$ in the null direction which is transversal to $\mathcal{H}$, where it may be only H\"{o}lder continuous.\footnote{The possibility of losing regularity in the transversal direction to $\mathcal{H}$ is already familiar from the solutions constructed in~\cite{nakedone,nakedinterior}; however, the solutions from~\cite{nakedone,nakedinterior}  are $C^{1,\alpha}$ across $\mathcal{H}$ and thus are better than the worst case scenario for our expansions.}
\end{enumerate}
Before one can undertake any sort of detailed analysis of twisted self-similar solutions, it will be necessary to fix a gauge. It will in fact turn out to be convenient to define two different gauges. We note that in this discussion of gauges we do not need to assume that the metric $g$ is a solution to the Einstein vacuum equations.

The first and simplest gauge we define is what we call the \emph{homothetic gauge} (see Definition~\ref{3moi09989}). In this gauge the metric $g$ takes the form
\begin{equation*}\label{formhomotheticform}
g = P dt^2 + 2t\left(dt\otimes d\rho + d\rho \otimes dt\right) - h_A\left(dt\otimes d\theta^A + d\theta^A\otimes dt\right) + \slashed{g}_{AB}d\theta^A\otimes d\theta^B, 
\end{equation*}
and the homothetic vector field $K$ takes the form
\[K = t\partial_t.\]
Here $\left\{\theta^A\right\}$ denote local coordinates on $\mathbb{S}^2$ and $\left(\rho,t\right) \in (-c,c) \times (-\infty,0)$ for some constant $c > 0$.\footnote{One may also consider ``one-sided'' versions of the gauge where $\rho \in [0,c)$ or $\rho \in (-c,0]$, but in this introductory section we will just consider the ``two-sided'' case.} The fact that $\mathcal{L}_Kg = 2g$ forces the following constraints on $P$, $h$, and $\slashed{g}$:
\begin{equation}\label{oieiu123123123}
P\left(\rho,t,\theta^A\right) = \tilde{P}\left(\rho,\theta^A\right),\qquad h_A\left(\rho,t,\theta^B\right) = t\tilde{h}_A\left(\rho,\theta^B\right),\qquad \slashed{g}_{AB}\left(\rho,t,\theta^C\right) = t^2\tilde{\slashed{g}}_{AB}\left(\rho,\theta^C\right),
\end{equation}
for suitable $\tilde{P}$, $\tilde{h}_A$, and $\tilde{\slashed{g}}_{AB}$. The null hypersurface $\mathcal{H}$ will correspond to $\{\rho = 0\}$, and the requirement that $\mathcal{H}$ is null and that $\slashed{g}_{AB}$ is positive definite forces
\[\left(P - \left|h\right|^2_{\slashed{g}}\right)|_{\rho = 0} = 0.\]

 We will show that any twisted self-similar spacetime may be put into the homothetic gauge by considering the foliation of the spacetime by spheres which lie at the intersection of a family of dilation invariant hypersurfaces  and a family of null hypersurfaces which intersect $\mathcal{H}$ transversally (see Proposition~\ref{homotheticgaugeexist}). These two families of hypersurfaces then correspond to the constant $\rho$ and constant $t$ hypersurfaces respectively. We note that this gauge is the natural generalization of the \emph{normal form} gauge defined by Fefferman--Graham in Definition 2.7 of~\cite{FG2}. If we wish to consider the case when $K$ is null  along $\mathcal{H}$ (which corresponds to the Fefferman--Graham geometries) then we must have that both $P$ and $h$ vanish when $\rho = 0$.

Despite the apparent relative simplicity of the homothetic gauge, most of our study of the twisted self-similar spacetimes will take place in more analytically flexible double-null gauges (see Section~\ref{nullstreqn}) where the metric takes either the form
\begin{equation}\label{doubledoubleone}
g = -2\Omega^2\left(du\otimes dv + dv\otimes du\right) + \slashed{g}_{AB}\left(d\theta^A - b^Adu\right)\otimes \left(d\theta^B - b^Bdu\right)
\end{equation}
or
\begin{equation}\label{doubledoubletwo}
g = -2\Omega^2\left(du\otimes dv + dv\otimes du\right) + \slashed{g}_{AB}\left(d\theta^A - b^Adv\right)\otimes \left(d\theta^B - b^Bdv\right),
\end{equation}
and the homothetic vector field $K$ takes the form
\begin{equation}\label{3oj1oi491}
K = u\partial_u+v\partial_v.
\end{equation}
Here $\left\{\theta^A\right\}$ again denote local coordinates along $\mathbb{S}^2$, and $(u,v)$ lie in a suitable subset of $\mathbb{R}^2$. We refer to~\eqref{doubledoubleone} as a double-null coordinate system with the ``shift in the $u$-direction'' and to~\eqref{doubledoubletwo} as a double-null coordinate system with the ``shift in the $v$-direction.'' If a spacetime has been equipped with one of these coordinate systems and the homothetic vector field takes the form~\eqref{3oj1oi491}, then we say that the spacetime has been put into the ``self-similar double-null gauge'' (see Definition~\ref{thisisselfdoublenull}). The fact that $\mathcal{L}_Kg = 2g$ forces the following constraints on $\Omega$, $b$, and $\slashed{g}$:
\begin{equation}\label{constrtobeself}
\Omega\left(u,v,\theta^A\right) = \tilde{\Omega}\left(\frac{v}{u},\theta^A\right),\qquad b^A\left(u,v,\theta^B\right) = u^{-1}\tilde{b}^A\left(\frac{v}{u},\theta^B\right),\qquad \slashed{g}_{AB}\left(u,v,\theta^C\right) = u^2\tilde{\slashed{g}}_{AB}\left(\frac{v}{u},\theta^C\right),
\end{equation}
for suitable $\tilde{\Omega}$, $\tilde{b}^A$, and $\tilde{\slashed{g}}_{AB}$. 

It turns out that one cannot in general expect to put a twisted self-similar spacetime into a self-similar double-null gauge in a region which includes the important null hypersurface $\mathcal{H}$. However, by constructing a suitable solution to a self-similar eikonal equation, we will show that given a twisted self-similar solution in the homothetic gauge, one may find a dilation invariant neighborhood $\mathcal{U}$ of $\mathcal{H}$ such that $\mathcal{U}\setminus \mathcal{H}$ may be put in the self-similar double-null gauge with the shift in the $u$-direction and with the $(u,v)$ coordinates ranging over 
\begin{equation}\label{3ojoj43221211}
\left\{(u,v) \in \mathbb{R}^2 : u \in (-\infty,0) : \frac{v}{-u} \in (-\tilde{c},\tilde{c})\setminus \{0\}\right\},
\end{equation}
for some $\tilde{c} > 0$ (see Proposition~\ref{makeitadouble}). Furthermore, we will identify precise assumptions on the behavior of the metric and Ricci coefficients as $v\to 0$ which allow us to reverse this process and go from a spacetime in the self-similar double-null gauge with coordinates covering the range~\eqref{3ojoj43221211} to a self-similar spacetime in the homothetic gauge (see Definition~\ref{thisdefissomething} and Proposition~\ref{nonononononoo}). In the following diagram, the gray area (but not the bolded line) corresponds to the region of validity for the self-similar double-null gauge:
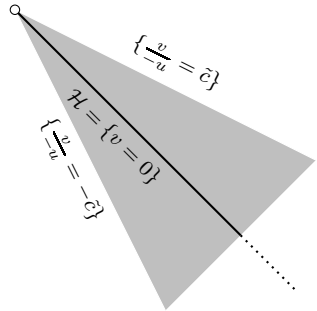
\begin{center}
\begin{tikzpicture}
\fill[lightgray] (0,0)--(3,-3) -- (4,-2) -- (0,0);
\draw [white](0,0) -- (4,-2) node[black, sloped,above,midway]{\footnotesize $\{\frac{v}{-u} = \tilde{c}\}$};
\fill[lightgray] (0,0)--(3,-3) -- (2,-4) -- (0,0);
\draw [white](0,0) -- (2,-4) node[black, sloped,below,midway]{\footnotesize $\{\frac{v}{-u} = -\tilde{c}\}$};
\draw[thick] (0,0) -- (3,-3) node[sloped,below,midway]{\footnotesize $\mathcal{H} = \{ v= 0\}$}; 
\path [draw=black,fill=white] (0,0) circle (1/16); 
\draw[dotted,thick] (3,-3)--(3.75,-3.75);
\end{tikzpicture}
\end{center}
\begin{center}
\footnotesize{figure 1}
\end{center}
Finally, we note that in the case of a Fefferman--Graham type solution to the Einstein vacuum equations where $K$ is null along all of $\mathcal{H}$, then one can always construct a self-similar double-null gauge which includes the hypersurface $\mathcal{H}$ (see Appendix~\ref{feffgrahamreggauge}). It is also possible to identify explicit conditions on the behavior of the metric in the self-similar double-null gauge which guarantee that the solution is one of the Fefferman--Graham type geometries. 
\subsection{Expansions}
Having constructed our canonical gauges, we then turn to the computation of the formal expansions for twisted self-similar spacetimes which solve the Einstein vacuum equations (to high order as $\frac{v}{-u} \to 0$). These expansions will be computed in the self-similar double-null gauge where we have access to the standard null structure equations for the Ricci coefficients and metric components (see Section~\ref{nullstreqn}). However, by using the proof of the equivalence between the homothetic and double-null gauge, having computed the expansions in the double-null gauge, one may deduce corresponding expansions in the homothetic gauge in an algorithmic fashion.

We now provide a heuristic overview of these formal expansions. We first explain in Section~\ref{regreg} how one may determine the expected regularity of the metric in the self-similar double-null gauge as $v\to 0$. In Section~\ref{whatisthepoint} we explain in more detail what the desired properties of our expansions are. Lastly in Section~\ref{computecompute} we will see how one then actually carries out the computation of the formal expansions and verifies their desired properties.
\subsubsection{Regularity as $v\to 0$ in the Double-Null Gauge}\label{regreg}

Even though our formal expansions will be done with respect to the distance from the null hypersurface $\mathcal{H}$, the quickest way to understand where these expansions come from is to switch perspectives and consider the problem of evolution towards $\mathcal{H}$ solving to the past from dilation invariant spacelike data along a dilation invariant hypersurface:
\begin{center}
\begin{tikzpicture}
\fill[lightgray] (0,0)--(3,-3) -- (4,-2) -- (0,0);
\draw [thick](0,0) -- (4,-2) node[black, sloped,above,midway]{\footnotesize \text{spacelike initial data}};
\draw[dashed] (0,0) -- (3,-3) node[sloped,below,midway]{\footnotesize $\mathcal{H} = \{ v= 0\}$}; 
\path [draw=black,fill=white] (0,0) circle (1/16); 
\draw[dotted,thick] (3.5,-2.5)--(4.5,-3.5);
\end{tikzpicture}
\end{center}
\begin{center}
\footnotesize{figure 2}
\end{center}
The basic question we wish to understand is, assuming such a self-similar spacetime exists in the double-null gauge, what kind of regularity can we expect from the metric components and Ricci coefficients as we approach $\mathcal{H}$ (which corresponds to the limit $v\to 0$). We start by observing that every Ricci coefficient $\psi$ or metric coefficient $\phi$ except for $\hat{\chi}$, $\underline{\eta}$, and $\omega$ satisfies an equation which relates $\mathcal{L}_{\partial_v}\psi$ or $\mathcal{L}_{\partial_v}\phi$ to contractions of angular derivatives of other metric components and Ricci coefficients. Thus, as long as the right hand sides of these equations are integrable in $v$, we could expect to simply integrate them and derive the leading order asymptotics of $\psi$ or $\phi$ as $v\to 0$. 

While $\omega$ does not satisfy a $\mathcal{L}_{\partial_v}$ equation, we may use self-similarity (see~\eqref{constrtobeself}) to relate $\omega$ with $\underline{\omega}$ and $\log\Omega$. For $\underline{\omega}$, we may compute (see already~\eqref{4uomega} and~\eqref{slashr}) that $\mathcal{L}_{\partial_v}\left(\Omega\underline{\omega}\right)$ satisfies an equation whose right hand side involves contractions of angular derivatives of Ricci coefficients and metric components, none of which are $\omega$. Thus, if we optimistically assume that this right hand is integrable in $v$, we may then conclude that it is reasonable to expect that $\Omega\underline{\omega}$ will continuously extend to $\{v = 0\}$ and satisfies the following bound $\left|\Omega\underline{\omega}\right| \lesssim |u|^{-1}$ (which is consistent with~\eqref{constrtobeself}). We may then relate $\log\Omega$ to $\Omega\underline{\omega}$ using self-similarity (see~\eqref{constrtobeself}) and obtain
\begin{equation}\label{lapfrom}
\left(-\frac{v}{u}\mathcal{L}_{\partial_v}+\mathcal{L}_b\right)\log\Omega = 2\Omega\underline{\omega}.
\end{equation}
On any constant $u$-hypersurface this is a transport equation which degenerates as $v\to 0$. In particular, it is clear from~\eqref{lapfrom} that $\log\Omega$ may blow-up as $v\to 0$ no matter how regular $\Omega\underline{\omega}$ is as $v\to 0$. However, if we assume that $b$ is suitably small, $b^A \sim  u^{-1}O\left(\epsilon\right)$, then the blow-up of the lapse (and angular derivatives thereof) will be ``mild'' in that, for example, $\left|\slashed{\partial}\log\Omega\right| \lesssim \left(\frac{v}{-u}\right)^{-O\left(\epsilon\right)}$ for any given coordinate derivative $\slashed{\partial}$ along $\mathbb{S}^2$. However, the blow-up of $\omega$ as $v\to 0$ is ``violent'' in that we do not expect it to be integrable as $v\to 0$. Thus, when integrating a $\mathcal{L}_{\partial_v}$ equation for any Ricci coefficient $\psi$ or metric component $\phi$, it is crucial that one conjugates the equation with a correct power of the lapse $\Omega$ so as to remove any occurrences of $\omega$. The upshot of this observation is that if we assume that our initial data is suitable small, then (suppressing the need to understand $\hat{\chi}$ for the moment) we expect that $\eta$, $\Omega^{-1}{\rm tr}\chi$, $\Omega{\rm tr}\underline{\chi}$, $\Omega\hat{\underline{\chi}}$, $\Omega\underline{\omega}$, $b$, and $\slashed{g}$ all extend in a H\"{o}lder continuous fashion to $\{v = 0\}$. Furthermore, for any large integer $N$, if we assume that $\epsilon$ is sufficiently small depending on $N$, we expect that $N$ coordinate angular derivatives of these quantities will also extend in a H\"{o}lder continuous fashion to $\{v = 0\}$. 

Since $\underline{\eta}$ and $\zeta$ may be recovered from angular derivatives of $\log\Omega$ and $\eta$, it only remains to discuss the regularity of $\hat{\chi}$ as $v\to 0$. For $\hat{\chi}$ we may use self-similarity and the $\nabla_3$ equation for $\hat{\chi}$ (see~\eqref{3hatchi}) to eventually derive an equation of the form
\begin{align}\label{123456789}
& -\frac{v}{u}\nabla_v\left(\Omega^{-1}\hat{\chi}\right)_{AB} +\mathcal{L}_b\left(\Omega^{-1}\hat{\chi}\right)_{AB}- \left(\slashed{\nabla}\hat{\otimes}b\right)^C_{\ \ (A}\left(\Omega^{-1}\hat{\chi}\right)_{B)C} -\frac{1}{2}\slashed{\rm div}b\left(\Omega^{-1}\hat{\chi}\right)_{AB}
\\ \nonumber&\qquad -\frac{v}{u}\Omega^{-1}{\rm tr}\chi\left(\Omega\hat{\chi}\right)_{AB} - 4\left(\Omega\underline{\omega}\right)\Omega^{-1}\hat{\chi}_{AB}= H_{AB},
 \end{align}
 where, by the previous heuristics, we expect $H_{AB}$ to extend in a H\"{o}lder continuous fashion to $\{v = 0\}$. As with the equation~\eqref{lapfrom}, we see that if $b$ and $\Omega\underline{\omega}$ are sufficiently small, then while $\hat{\chi}$  may in principle blow-up as $v\to 0$, the rate of blow-up will be ``mild.''
 
 We now quickly observe how the presence of a Fefferman--Graham geometry may be detected by this analysis. Namely,  it is possible that we have $\lim_{v\to 0}b = 0$. In this case, it is then possible to show that we must also have that $\lim_{v\to 0}\left(\Omega\underline{\omega}\right) = 0$ and, in view of~\eqref{lapfrom}, one is able to see that $\Omega$ extends continuously to $\{v = 0\}$. Similarly, $b|_{v=0} = 0$ ends up implying that the tensor $H$ on the right hand side of~\eqref{123456789} vanishes in the limit as $v\to 0$. We then eventually conclude that $\hat{\chi}$ also extends regularly to $\{v = 0\}$. By repeating the analysis \emph{mutatis mutandis} one eventually concludes that, as long as the initial data is sufficiently differentiable, any sequence of derivatives of a metric component or Ricci component must extend continuously to $\{v = 0\}$. (See Appendix~\ref{feffgrahamreggauge}.)
 
 We close the section by observing that even though it may appear at first glance as if our above analysis could only be expected to hold to the future of the null hypersurface $\mathcal{H}$, it is in fact equally valid in a region to the past of $\mathcal{H}$. The point is that the above arguments just involve integration of the null structure equations as ODE's with all angular derivatives treated as errors on the right hand side (as opposed to energy estimates). It thus makes no difference whether this argument is done to the past or future of $\mathcal{H}$. 

\subsubsection{What is a Formal Expansion?}\label{whatisthepoint}
Having determined the expected regularity of our twisted self-similar solutions as $v\to 0$, we now discuss the formal expansions. While we will leave the precise definitions to Section~\ref{sectionwithmaintheorems} (see Theorem~\ref{theoremexpand}), it may be useful for the reader if we informally(!)  describe in a bit more detail what we mean by ``formal expansion.'' Let $N$ be any sufficiently large positive integer and $\epsilon > 0$ be a small constant, possibly depending on $N$, which will govern all of the smallness necessary for our constructions. We then want a procedure which produces metrics $g$ so that
\begin{enumerate}
	\item Each metric $g$ is a twisted self-similar metric in the self-similar double-null gauge.
	\item Along any constant $u$ hypersurface, we have that 
	\begin{equation}\label{riccivanishvanish}
	{\rm Ric}\left(g\right) = \epsilon O\left(|v|^{N-1-O\left(\epsilon\right)}\right).
	\end{equation}
	\item We  may write $g = \sum_{j=0}^Ng^{(j)}$ so that 
	\begin{enumerate}
		\item Along any constant $u$ hypersurface, $g^{(j)} = \epsilon O\left(|v|^{j-O\left(\epsilon\right)}\right)$.
		\item Each $g^{(j)}$ is determined by inductively solving (possibly degenerate) explicit linear transport equations.
		\item  The base case of the induction is determined by an explicit choice of ``seed data'' which must satisfy a suitable smallness assumption.
	\end{enumerate}
	\item\label{describeeventually} If $\hat{g}$ is any twisted self-similar metric in the self-similar double-null gauge which satisfies the Einstein vacuum equations~\eqref{eve}, a suitable smallness condition (depending on $N$), and suitable regularity conditions as $v\to 0$, then there exists a metric $g$ produced by our procedure so that along any constant $u$-hypersurface $g - \hat{g} =  \epsilon O\left(|v|^{N-O\left(\epsilon\right)}\right)$.
\end{enumerate}
(If one does not restrict to a constant $u$-hypersurface, then for any given metric component or Ricci coefficient, one may add the $u$-dependence in these various estimate by use of~\eqref{constrtobeself}.) 

If we are successful in this constrution, then our procedure may be considered to exhaustively describe all possible asymptotic behaviors towards the null hypersurface $\mathcal{H}$ for twisted self-similar metrics which satisfy a suitable smallness and regularity assumption as $v\to 0$. We note, however, a significant difference with the formal expansions for Fefferman--Graham spacetimes from~\cite{FG1,FG2}; namely, our smallness assumption on $\epsilon$ depends on how far out in the expansion we wish to go. In particular, for any given choice of seed data, we do not produce expansions to infinite order; instead, for any given $N$, we can introduce a suitable smallness assumption such for all seed data satisfying the smallness assumption we may produce an expansion up to order $N$. Relatedly, our smallness assumption will also turn out to depend on the total number of derivatives of the metrics $g$ that we wish to estimate.

\subsubsection{Seed Data, Computing the Expansion, and Propagation of Constraints}\label{computecompute}
As in Section~\ref{regreg} we split the Ricci and metric coefficients into two groups, $\log\Omega$ and $\Omega^{-1}\hat{\chi}$ which only satisfy $\nabla_3$ equations, and all of the other metric and Ricci coefficients which satisfy $\nabla_4$ equations. (We drop $\underline{\eta}$, $\zeta$, and $\omega$ from this dichotomy since their behavior may be deduced from those of the other quantities.) Any of these $\nabla_4$ equations will imply an equation of the following form for a metric component $\phi$ or Ricci coefficient $\psi$:
\begin{equation}\label{fourofoureqn}
\mathcal{L}_{\partial_v}\left(\Omega^s\left(\phi,\psi\right)\right) = F,
\end{equation}
where $F$ consists of contractions of metric and Ricci coefficients and their angular derivatives, and $\Omega^s$ represents the power of the lapse which, after conjugating through by, leads to an equation which does not involve $\omega$. If we had an equation like~\eqref{fourofoureqn} for all metric components and Ricci coefficients, the computation of the expansions would be simple; we would choose suitable values for $\Omega^s\left(\phi,\psi\right)|_{v=0}$ and then inductively compute a power series in $v$ for $\Omega^s\left(\phi,\psi\right)$ via the equation~\eqref{fourofoureqn}. It would then be natural to take as seed data the values $\Omega^s\left(\phi,\psi\right)|_{v=0}$ (in turn, by self-similarity, these are in fact determined by their values on the sphere at $(u,v) = (-1,0)$). However, due to the presence of the null constraint equations, we do not expect to be able to freely pose the values of all of these along $\{v = 0\}$. Eventually it turns out to be reasonable to take as (part of our) seed data a choice of triple $\left(b,\slashed{g},\Omega\underline{\omega}\right)|_{(u,v) = (-1,0)}$ which are required to satisfy a suitable constraint equation. (See Definition~\ref{thisistheseed}.) It then turns out that $\Omega^{-1}{\rm tr}\chi|_{v=0}$ and $\eta|_{v=0}$ may be computed from these using the Einstein equations along $\{v = 0\}$.

We now discuss $\log\Omega$ and $\Omega^{-1}\hat{\chi}$. For these we re-write~\eqref{lapfrom} and~\eqref{123456789} as  
\begin{equation}\label{lapfromsecond}
\left(-\frac{v}{u}\mathcal{L}_{\partial_v}+\mathcal{L}_{b(0)}\right)\log\Omega = G_1,
\end{equation} 
\begin{align}\label{123456789second}
& -\frac{v}{u}\nabla_v\left(\Omega^{-1}\hat{\chi}\right)_{AB} +\mathscr{L}|_{v=0}\left(\Omega^{-1}\hat{\chi}\right)_{AB}- 4\left(\Omega\underline{\omega}\right)|_{v=0}\left(\Omega^{-1}\hat{\chi}\right)_{AB}= \left(G_2\right)_{AB},
 \end{align}
 where $\mathscr{L}$ is the operator
 \[ \mathscr{L}f_{AB} \doteq \mathcal{L}_bf_{AB}- \left(\slashed{\nabla}\hat{\otimes}b\right)^C_{\ \ (A}f_{B)C} -\frac{1}{2}\slashed{\rm div}bf_{AB},\]
 and $G_1$ and $G_2$ involve either suitable contractions of angular derivatives of metric components and Ricci coefficients which satisfy equations like~\eqref{fourofoureqn} or involve  $\Omega^{-1}\hat{\chi}$ or $\log\Omega$ contracted with a quantity which we expect to vanish at $\{v = 0\}$. If we consider $G_1$ and $G_2$ as given inhomogeneities, then the solutions $\log\Omega$ and $\Omega^{-1}\hat{\chi}$ to~\eqref{lapfromsecond} and~\eqref{123456789second} would be uniquely determined up to the addition of an element in the kernel of the operators on the left hand sides. It is then natural to consider a choice of such kernel elements also as (a part of our) seed data for the expansion (see Definition~\ref{thisistheseed}). 
 
In conclusion, it is natural to take as the seed data for our expansion a choice of $\left(b,\slashed{g},\Omega\underline{\omega}\right)|_{(u,v) = (-1,0)}$ which satisfies a suitable constraint equation and then two solutions $\log\Omega$ and $\Omega^{-1}\hat{\chi}$ which lie in the kernel of the operators on the left hand side of~\eqref{lapfromsecond} and~\eqref{123456789second}. Given this, an expansion may be inductively computed by plugging in the previously computed metric into the right hand sides of~\eqref{fourofoureqn},~\eqref{lapfromsecond}, and~\eqref{123456789second} and then solving the corresponding equations to produce a new metric. 

There is one important complication we still have not mentioned. Namely, since the null-structure equations are overdetermined, the procedure we have outlined above does not manifestly lead to metrics $g$ which satisfy~\eqref{riccivanishvanish}. Instead, only the components of the Ricci tensor which correspond to the subset of the null structure equations we solved will have the desired vanishing property as $v\to 0$. To obtain the full statement~\eqref{riccivanishvanish} we will need to run an additional ``preservation of constraints'' argument which uses the fact that the Einstein tensor is divergence free, to upgrade the initial partial vanishing of the Ricci tensor as $v\to 0$ to the full desired statement~\eqref{riccivanishvanish}. (See Section~\ref{woahconstraintsarecool}.) A similar difficulty occurs in the work~\cite{nakedinterior} though the resolution here will be dramatically simpler.

\subsection{Applications}
In this final introductory section, we will discuss various applications of the existence of the formal expansions. 

We first observe that by an amalgamation of the techniques from~\cite{scaleinvariant,nakedone}, it is possible to show that one may associate true solutions to the Einstein vacuum equations to each of our formal expansions in a region to the future of $\mathcal{H}$ (it is only in this region that the homothetic vector field is everywhere spacelike):
\begin{center}
\begin{tikzpicture}
\fill[lightgray] (0,0)--(3,-3) -- (4,-2) -- (0,0);
\draw [dashed](0,0) -- (4,-2);
\draw [<-] (1.75,-1.25)--(3,0);
\draw (3,0) node[above]{\footnotesize ${\rm true\ solution}$};
\draw[thick] (0,0) -- (3,-3) node[sloped,below,midway]{\footnotesize $\mathcal{H} = \{ v= 0\}$}; 
\path [draw=black,fill=white] (0,0) circle (1/16); 
\draw[dotted,thick] (3.5,-2.5)--(4.5,-3.5);
\end{tikzpicture}
\end{center}
\begin{center}
\footnotesize{figure 3}
\end{center}
However, in order to not increase the length of this paper significantly we will not carry out the proof of this result. (See Theorem~\ref{notprovedhere}.)

We will next explain how our expansions can be used to understand the possible asymptotic behaviors in the entire small data regime for solutions arising from spacelike self-similar data. The first point is that by swapping the roles of the $u$ and $v$ variable, we may consider our expansions as being defined to the past of a null hypersurface at $\{u = 0\}$ and where the metric is now in a double-null gauge with the shift in the $v$-direction:
\begin{center}
\begin{tikzpicture}
\fill[lightgray] (0,0)--(3,3) -- (4,2) -- (0,0);
\draw [dashed](0,0) -- (4,2);
\draw[thick] (0,0) -- (3,3) node[sloped,above,midway]{\footnotesize $\{ u= 0\}$}; 
\path [draw=black,fill=white] (0,0) circle (1/16); 
\draw[dotted,thick] (3.5,2.5)--(4.5,3.5);
\end{tikzpicture}
\end{center}
\begin{center}
\footnotesize{figure 4}
\end{center}
In particular, from item~\ref{describeeventually} of Section~\ref{whatisthepoint} we will have that any twisted self-similar metric which is in the self-similar double-null gauge (with the shift in the $v$-direction), exists in a dilation invariant neighborhood to the past of $\{u=0\}$, satisfies a suitable smallness assumption, and satisfies suitable regularity conditions as $u\to 0$ will be well-approximated by one of our formal expansions. A priori, metrics which satisfy the smallness assumption and regularity requirements as $u\to 0$ could comprise a very special class. However, it is essentially an immediate consequence of the a priori estimates from Sections 7, 8, and 9 of~\cite{nakedone} and the technique for establishing the local existence of self-similar spacetimes from~\cite{scaleinvariant} that any spacetime arising from suitably small spacelike self-similar data may be put into a self-similar double-null gauge with the shift in the $v$-direction, will exist up to the hypersurface $\{u = 0\}$, and will always satisfy the necessary smallness and regularity conditions at $\{u=0\}$. Thus, these solutions will then be well-approximated by the formal expansions in the sense of item~\ref{describeeventually} of Section~\ref{whatisthepoint} with the role of $u$ and $v$ swapped (see Theorem~\ref{spacespaceself}):
\begin{center}
\begin{tikzpicture}
\fill[lightgray] (0,0)--(3,3)--(4,-1);
\draw [dashed](0,0) -- (3.4,1.7);
\draw[thick] (0,0) -- (3,3) node[sloped,above,midway]{\footnotesize $\{ u= 0\}$}; 
\draw[thick] (0,0) -- (4,-1) node[sloped,below,midway]{\footnotesize {\rm small\ spacelike\ self-similar\ data}};
\path [draw=black,fill=white] (0,0) circle (1/16); 
\draw[dotted,thick] (3.5,1)--(4.1,1.2);
\draw[<-] (1.5,1.1) -- (2.5,3);
\draw (2.5,3) node[above]{\footnotesize {\rm approximated\ by\ expansions}};
\end{tikzpicture}
\end{center}
\begin{center}
\footnotesize{figure 5}
\end{center}

This general description of the future asymptotics of solutions arising from small spacelike self-similar data then has the following two additional applications:
\begin{enumerate}
	\item By the results of~\cite{scaleinvariant} any Fefferman--Graham spacetime in $3+1$ dimensions associated to sufficiently small Dirichlet and Neumann data will exist in a region as depicted by figure 3 and will induce spacelike self-similar data which is sufficiently small so that the discussion above is valid. Hence the global structure up to $\{u = 0\}$ of these Fefferman--Graham is described by one of our formal expansions in the sense of item~\ref{describeeventually} of Section~\ref{whatisthepoint} with the role of $u$ and $v$ swapped (see Corollary~\ref{afirstcororor}).
	\item While the exterior region of the naked singularity produced in~\cite{nakedone} is not self-similar, one may apply the method of ``self-similar extraction'' as described by Theorem 1.2 of~\cite{scaleinvariant} to produce a corresponding self-similar metric; more specifically, let $g_{\rm naked}$ denote the metric produced by Theorem 1 of~\cite{nakedone} and then, for $\lambda > 0$,  denote the diffeomorphism $\left(u,v,\theta^A\right) \mapsto \left(\lambda u,\lambda v,\theta^A\right)$ by $\Phi_{\lambda}$. The proof of Theorem 1.2 \emph{mutatis mutandis} shows that $\lambda^{-2}\Phi_{\lambda}^*g_{\rm naked}$ will converge as $\lambda \to 0$ to a self-similar metric $\tilde{g}$ which moreover induces spacelike self-similar data which is sufficiently small so that the discussion above is valid, and thus $\tilde{g}$ is well-approximated by our formal expansions near $\{u = 0\}$ in the sense of item~\ref{describeeventually} of Section~\ref{whatisthepoint} with the role of $u$ and $v$ swapped (see Corollary~\ref{nakedcorcor}). Thus, in this sense, our formal expansions describe the asymptotic behavior behavior of the naked singularities from~\cite{nakedone,nakedinterior} as the metric approaches the future light cone of the singularity.\footnote{Though this is outside the scope of the current paper to discuss in detail, one expects that the self-similar metric $\tilde{g}$ captures the main singularities as $u\to 0$ in that one expects $g_{\rm naked} - \tilde{g}$ to be more regular as $u\to 0$ than $g_{\rm naked}$.}
\end{enumerate}

We close the introduction by stating one natural question which is left unaddressed by the above analysis.
\begin{question}\label{stableglob} By the above discussion, any Fefferman--Graham spacetime associated to sufficiently small Dirichlet and Neumann data will exist up to the outgoing cone $\{u = 0\}$ and will be well-approximated there by one of the formal expansions. Is it the case that the cone $\{u = 0\}$ is also described by a Fefferman--Graham geometry? 
\end{question}
One motivation behind this question is as follows:  An analogy may be drawn between the role of Fefferman--Graham self-similar solutions  within the class of general twisted self-similar solutions and the role of the ``scale-invariant'' solutions written down by Christodoulou in~\cite{ChristBV} within the more general class of self-similar solutions for the spherically symmetric Einstein-scalar field system studied by Christoudoulou in~\cite{ChristNaked}. The fact that small scale-invariant solutions exist up to and are regular along the analogue of $\{u = 0\}$ may be considered a harbinger of the well-posedness of the spherically symmetric Einstein-scalar field system in the class of ``solutions of bounded variation''~\cite{ChristBV}.
\subsection{Relation with Naked Singularities from~\cite{nakedone,nakedinterior}}
In this final introductory subsection, we briefly indicate how the general twisted self-similar solutions of this paper are related to the self-similarity which is relevant for the solutions considered in~\cite{nakedone,nakedinterior}. The key specialization is that in the works~\cite{nakedone,nakedinterior} we require that the choice of the seed data $\left(b,\slashed{g},\Omega\underline{\omega}\right)|_{(u,v) = (-1,0)}$ (see Section~\ref{computecompute}) take a certain specific form, the most important aspect of which is a very particular leading order ansatz for the vector field $b$. As a consequence of the seed data taking this special form, the solutions $\log\Omega$ and $\Omega^{-1}\hat{\chi}$ to the equations~\eqref{lapfromsecond} and~\eqref{123456789second} turn out to have additional regularity as $v\to 0$ than in the general case.  This is then ultimately responsible for the spacetimes produced in~\cite{nakedone,nakedinterior} being $C^{1,\alpha}$ across $\{v = 0\}$ as opposed to just $C^{0,\alpha}$.

\subsection{Acknowledgements}
The author thanks Igor Rodnianski for useful conversations. The author also acknowledges support from NSF grant DMS-1900288, from an Alfred P. Sloan Fellowship in Mathematics, and from NSERC discovery grants RGPIN-2021-02562 and DGECR-2021-00093. 
\subsection{Data Sharing and Conflict of Interest Statement} 
Data sharing is not applicable to this article as no datasets were generated or analyzed during the current study. The author has no competing interests to declare that are relevant to the content of this article.

\section{Null Structure Equations and Double-Null Formalism}\label{nullstreqn}
In this section we provide a quick review of the double-null gauge formalism. We will use all of the same normalizations and notation as in Section 3.2 of~\cite{nakedinterior}. Here we will just briefly recall the relevant notation and equations. (For a true introduction to the double-null gauge, see~\cite{KN} or~\cite{Chr}.)

We let $\mathcal{U}$ denote a suitable open set in $\mathbb{R}^2$ and consider spacetimes covered by coordinates $(u,v,\theta^A) \in \mathcal{U} \times \mathbb{S}^2$ where the metric $g$ takes the form
\begin{equation}\label{iooioi284}
g = -2\Omega^2\left(du\otimes dv + dv\otimes du\right) + \slashed{g}_{AB}\left(d\theta^A - b^Adu\right)\otimes \left(d\theta^B - b^Bdu\right).
\end{equation}
We do not assume that the Einstein vacuum equations hold. 

The null vectors $e_4$ and $e_3$ are defined by 
\[e_4 \doteq \Omega^{-1}\partial_v,\qquad e_3 \doteq \Omega^{-1}\left(\partial_u + b\cdot\slashed{\nabla}\right).\]
These satisfy $g\left(e_3,e_4\right) = -2$.  The $\mathbb{S}^2_{u,v}$-projected $e_4$ and $e_3$ derivatives are denoted by $\nabla_4$ and $\nabla_3$ respectively. We will use $\slashed{\nabla}$ to denote the induced covariant derivative along $\mathbb{S}^2$. As in~\cite{nakedinterior} slashed quantities will always denote an object defined with respect to $\slashed{g}$. We use $\mathcal{L}$ to denote Lie-derivatives.

The Ricci coefficients $\omega$, $\underline{\omega}$, $\zeta_A$, $\eta_A$, $\underline{\eta}_A$, $\chi_{AB}$, and $\underline{\chi}_{AB}$ are defined in terms of the metric via the following formulas:
\begin{equation}\label{blahblah1}
\omega = -\frac{1}{2}\nabla_4\log\Omega,\qquad \underline{\omega} = -\frac{1}{2}\nabla_3\log\Omega,\qquad \mathcal{L}_{\partial_v}b^A = -4\Omega^2\zeta^A,
\end{equation}
\begin{equation}\label{blahblah2}
\eta_A = \zeta_A + \slashed{\nabla}_A\log\Omega,\qquad \underline{\eta}_A = -\zeta_A + \slashed{\nabla}_A\log\Omega,
\end{equation}
\begin{equation}\label{blahblah3}
\mathcal{L}_{e_4}\slashed{g}_{AB} = 2\chi_{AB},\qquad \mathcal{L}_{e_3}\slashed{g}_{AB} = 2\underline{\chi}_{AB}.
\end{equation}
We also split $\chi$ and $\underline{\chi}$ into their trace and trace-free parts:
\[\chi_{AB} = \hat{\chi}_{AB} + \frac{1}{2}{\rm tr}\chi\slashed{g}_{AB},\qquad \underline{\chi}_{AB} = \hat{\underline{\chi}}_{AB} + \frac{1}{2}{\rm tr}\underline{\chi}\slashed{g}_{AB}.\]
In general, $\widehat{\Theta}_{AB}$ will denote the trace-free part of a symmetric $\mathbb{S}^2_{u,v}$ $(0,2)$-tensor. The norms and raising and lowering of indices of $\mathbb{S}^2_{u,v}$ tensors are computed with respect to $\slashed{g}$.

Now we present the null structure equations which relate the derivatives of Ricci coefficients to curvature components. The following is Proposition 3.1 from~\cite{nakedinterior}.
\begin{proposition}\label{thenullstructeqns}
\begin{align}
\label{4trchi}\nabla_4{\rm tr}\chi + \frac{1}{2}\left({\rm tr}\chi\right)^2 &=-{\rm Ric}_{44} -\left|\hat{\chi}\right|^2 - 2\omega{\rm tr}\chi,
\\ \label{4hatchi} \nabla_4\hat{\chi}_{AB}+{\rm tr}\chi \hat{\chi}_{AB} &= -\widehat{R_{A4B4}} -2\omega\hat{\chi}_{AB},
\\ \label{3truchi} \nabla_3{\rm tr}\underline{\chi} + \frac{1}{2}\left({\rm tr}\underline{\chi}\right)^2 &=-{\rm Ric}_{33} -\left|\hat{\underline{\chi}}\right|^2 - 2\underline{\omega}{\rm tr}\underline{\chi},
\\ \label{3hatuchi} \nabla_3\underline{\hat{\chi}}_{AB}+{\rm tr}\underline{\chi} \underline{\hat{\chi}}_{AB} &= -\widehat{R_{A3B3}} -2\underline{\omega}\underline{\hat{\chi}}_{AB},
\\ \label{3hatchi} \nabla_3\hat{\chi}_{AB} +\frac{1}{2}{\rm tr}\underline{\chi}\hat{\chi}_{AB}&= \widehat{{\rm Ric}}_{AB} + 2\underline{\omega}\hat{\chi}_{AB} + \left(\slashed{\nabla}\hat\otimes \eta\right)_{AB} + \left(\eta\hat\otimes \eta\right)_{AB} - \frac{1}{2}{\rm tr}\chi \hat{\underline{\chi}}_{AB},
\end{align}
\begin{align}
 \label{3trchi}\nabla_3{\rm tr}\chi + \frac{1}{2}{\rm tr}\chi{\rm tr}\underline{\chi} &= \slashed{g}^{AB}R_{3AB4}+ 2\underline{\omega}{\rm tr}\chi + 2\slashed{\rm div}\eta + 2\left|\eta\right|^2 - \hat{\chi}\cdot\hat{\underline{\chi}},
\\ \label{4hatuchi} \nabla_4\hat{\underline{\chi}}_{AB} + \frac{1}{2}{\rm tr}\chi \hat{\underline{\chi}}_{AB}&= \widehat{{\rm Ric}}_{AB} + 2\omega\hat{\underline{\chi}}_{AB} + \left(\slashed{\nabla}\hat\otimes \underline\eta\right)_{AB} + \left(\underline\eta\hat\otimes \underline\eta\right)_{AB} - \frac{1}{2}{\rm tr}\underline{\chi} \hat{\chi}_{AB},
\\ \label{4truchi} \nabla_4{\rm tr}\underline{\chi} + \frac{1}{2}{\rm tr}\chi{\rm tr}\underline{\chi} &= \slashed{g}^{AB}R_{3AB4}+ 2\omega{\rm tr}\underline{\chi} + 2\slashed{\rm div}\underline{\eta} + 2\left|\underline\eta\right|^2 - \hat{\chi}\cdot\hat{\underline{\chi}},
\\ \label{4eta} \nabla_4\eta_A &= -\chi_{AB}\cdot\left(\eta^B-\underline{\eta}^B\right) -\frac{1}{2}R_{A434},
\\ \label{3ueta} \nabla_3\underline{\eta}_A &= -\underline{\chi}_{AB}\cdot\left(\underline\eta^B-\eta^B\right) -\frac{1}{2}R_{A343},
\\ \label{curleta} \slashed{\nabla}_A\eta_B-\slashed{\nabla}_B\eta_A &= R_{4[AB]3} + \hat{\underline{\chi}}^C_{\ \ [A}\hat{\chi}_{B]C},
\\ \label{curlueta} \slashed{\nabla}_A\underline{\eta}_B-\slashed{\nabla}_B\underline{\eta}_A &= -R_{4[AB]3} - \hat{\underline{\chi}}^C_{\ \ [A}\hat{\chi}_{B]C},
\\ \label{4uomega} \nabla_4\underline{\omega} &= \frac{1}{4}\left({\rm Ric}_{34} + \slashed{g}^{AB}R_{3AB4}\right) + 2\underline{\omega}\omega + \frac{1}{2}\left|\eta\right|^2 - \eta\cdot\underline{\eta},
\\ \label{3omega} \nabla_3\omega &= \frac{1}{4}\left({\rm Ric}_{34} + \slashed{g}^{AB}R_{3AB4}\right) + 2\underline{\omega}\omega + \frac{1}{2}\left|\underline\eta\right|^2 - \eta\cdot\underline{\eta},
\end{align}
and

\begin{align}\label{genGauss}
\slashed{R}_{ABCD} &= R_{ABCD} + \frac{1}{2}\left(\underline{\chi}_{BC}\chi_{AD} + \chi_{BC}\underline{\chi}_{AD} - \underline{\chi}_{AC}\chi_{BD} - \chi_{AC}\underline{\chi}_{BD}\right),
\\ \label{useful} \widehat{\slashed{g}^{CD}R_{ACBD}} &= 0,
\\ \label{ricequality} \widehat{{\rm Ric}}_{AB} &= \widehat{R_{3(AB)4}},
\\ \label{slashr} K &= \frac{1}{2}\slashed{g}^{AB}R_{3A4B} + \frac{1}{2}R +\frac{1}{2}{\rm Ric}_{34}+\frac{1}{2}\hat{\chi}\cdot\hat{\underline{\chi}} -\frac{1}{4}{\rm tr}\chi{\rm tr}\underline{\chi},
\\ \label{cod1} \slashed{\nabla}_A\chi_{BC} - \slashed{\nabla}_B\chi_{AC} &= R_{ABC4} +\chi_{AC}\zeta_B - \chi_{BC}\zeta_A,
\\ \label{cod2} \slashed{\nabla}_A\underline{\chi}_{BC} - \slashed{\nabla}_B\underline{\chi}_{AC} &= R_{ABC3} -\underline{\chi}_{AC}\zeta_B + \underline{\chi}_{BC}\zeta_A,
\\ \label{tcod1} \slashed{\nabla}^B\hat{\chi}_{AB}-\frac{1}{2}\slashed{\nabla}_A{\rm tr}\chi &=  -\frac{1}{2}R_{A434}+ {\rm Ric}_{4A}+ \frac{1}{2}{\rm tr}\chi \zeta_A - \zeta^B\hat{\chi}_{AB},
\\ \label{tcod2} \slashed{\nabla}^B\hat{\underline\chi}_{AB}-\frac{1}{2}\slashed{\nabla}_A{\rm tr}\underline\chi  &= -\frac{1}{2}R_{A343} + {\rm Ric}_{3A}- \frac{1}{2}{\rm tr}\underline\chi \zeta_A + \zeta^B\underline{\hat{\chi}}_{AB}.
\end{align}
\end{proposition}

\begin{remark}\label{shiftshift}(Shifting the Shift) As in our work~\cite{nakedone}, it will sometimes be useful to consider an alternative double-null coordinate system where~\eqref{iooioi284} is replaced by  
\begin{equation}\label{2l3kj2lijo42}
g = -2\Omega^2\left(du\otimes dv + dv\otimes du\right) + \slashed{g}_{AB}\left(d\theta^A - b^Adv\right)\otimes \left(d\theta^B - b^Bdv\right).
\end{equation}
All of the corresponding equations remain the same with following exceptions: we now have
\[e_4 \doteq \Omega^{-1}\left(\partial_v+b\cdot\slashed{\nabla}\right),\qquad e_3 \doteq \Omega^{-1}\partial_u,\]
and the relation between $b$ and the torsion $1$-form $\zeta$ becomes 
\[\mathcal{L}_{\partial_u}b^A = 4\Omega^2\zeta^A.\]

By reversing the roles of $u$ and $v$ we can easily go back and forth between double-null coordinate systems with the shift in the $u$ or $v$-direction. The following particular version of this will be important to us later: Suppose we have a metric $g$ covered by a double-null coordinate system with the shift in $u$-direction and $(u,v,\theta^A)$ lying in the set $\mathcal{P}  \doteq \left\{ \left(u,v,\theta^A\right) \in \mathbb{R}^2\times \mathbb{S}^2: u \in (-\infty,0)\text{ and }0 < \frac{v}{-u} < c\right\}$ for some $c > 0$. Let $\tilde{\mathcal{P}} \doteq \left\{ \left(u,v,\theta^A\right) \in \mathbb{R}^2\times \mathbb{S}^2: v \in (0,\infty)\text{ and }c^{-1} < \frac{v}{-u} < \infty \right\}$ denote the image of $\mathcal{P}$ under the map $\Phi$ which sends $\left(u,v,\theta^A\right) \mapsto \left(-v,-u,\theta^A\right)$. Then the pushforward of $g$ by $\Phi$ will be in a double-null coordinate system with the shift in the $v$-direction, and we will have the following correspondences of metric components and Ricci coefficients:
\[\left(\Omega,b,\slashed{g}\right) \mapsto \left(\Omega,-b,\slashed{g}\right),\]
\[\left(\omega,\underline{\omega},\eta,\underline{\eta},\chi,\underline{\chi}\right) \mapsto \left(-\underline{\omega},-\omega,\underline{\eta},\eta,-\underline{\chi},-\chi\right).\]
\end{remark}

\section{Degenerate Transport Equations}
 In this section we recall some previously proven results concerning degenerate transport equations along $\mathbb{S}^2$ which were proven in~\cite{nakedone}, and then also establish a few other results concerning linear degenerate transport equations along $(v,\theta^A) \in (0,c)\times\mathbb{S}^2$ for $c > 0$. We will then use this to study a particular nonlinear degenerate transport equation which will come up when we study the relation between the homothetic and self-similar double-null gauges. 

\subsection{Linear Degenerate Transport Equations}
As in the works~\cite{nakedone,nakedinterior} an important role will be played by various transport equations. We will now recall/establish the necessary linear estimates in this section.

The following proposition may be deduced from the proof of Proposition 4.5 in~\cite{nakedone}.
\begin{proposition}\label{somestuimdie}Let $H$ be a $(0,k)$ tensor on $\mathbb{S}^2$ and $q$ be a vector field on $\mathbb{S}^2$. We are interested in finding $(0,k)$-tensors $\psi$ on $\mathbb{S}^2$ so that
\begin{equation}\label{3o3oko03}
\mathcal{L}_q\psi + A\psi + W\cdot\psi = H,
\end{equation}
where $W$ is a tensor so that the contraction $W\cdot \psi$ produces a $(0,k)$-tensor, and $A$ is a constant satisfying $A \gtrsim 1$. 

Suppose that for some $j \geq 2$, $q$ lies in $\mathring{H}^{j+1}(\mathbb{S}^2)$  with $\left\vert\left\vert q\right\vert\right\vert_{C^1} \ll_j A$, $W$ lies in $\mathring{H}^j\left(\mathbb{S}^2\right)$ with $\left\vert\left\vert W\right\vert\right\vert_{L^{\infty}} \ll_j A$, and that $H$ lies in $\mathring{H}^j(\mathbb{S}^2)$, where $\mathring{H}^j\left(\mathbb{S}^2\right)$ is the Sobolev space with respect to the round metric.  Then,  there exists a unique $\mathring{H}^j(\mathbb{S}^2)$  solution $\psi$ to~\eqref{3o3oko03}. Furthermore, for any integer $s \in [0,j]$ we have
\begin{equation}\label{ij2938h482h4}
\left\vert\left\vert \left(\psi,\mathcal{L}_q\psi\right)\right\vert\right\vert_{\mathring{H}^s} \leq C\left(k,j,\left\vert\left\vert q\right\vert\right\vert_{\mathring{H}^{s+1}},\left\vert\left\vert W\right\vert\right\vert_{\mathring{H}^s}\right) \left\vert\left\vert H\right\vert\right\vert_{\mathring{H}^s}
\end{equation}
\end{proposition}

The following lemma, which also concerns transport equations, will also be useful.
\begin{lemma}\label{laptrans}Let $k \in \mathbb{Z}_{\geq 0}$, $j \in \mathbb{Z}_{\geq 1}$ be sufficiently large, $c > 0$, and $q(v)$ be a $1$-parameter family of vector fields along $\mathbb{S}^2$ satisfying 
\[\sup_v\sum_{j_1+j_2 \leq j+1}\left\vert\left\vert \left(v\mathcal{L}_{\partial_v}\right)^{j_1}q\right\vert\right\vert_{C^{j_2}\left(\mathbb{S}^2\right)} \lesssim \epsilon,\]
 and $p$ denotes any positive number so that $\epsilon p^{-1} \ll_j 1$. We will then consider certain transport equations for an unknown $(0,k)$-tensor $\psi_{A_1\cdots A_k}\left(v,\theta^A\right) : (0,c) \times \mathbb{S}^2 \to \mathbb{R}$ and known $(0,k)$-tensor $H_{A_1\cdots A_k}: (0,c) \times \mathbb{S}^2 \to \mathbb{R}$:
\begin{equation}\label{3oj2omo402}
v\mathcal{L}_{\partial_v}\psi + \mathcal{L}_q\psi - A\psi + W\cdot \psi = H,
\end{equation}
where $A \geq 0$, and $W$ is a tensor so that the contraction $W\cdot \psi$ is still a $(0,k)$ tensor and satisfies
\[\sup_v\sum_{j_1+j_2 \leq j}\left\vert\left\vert \left(v\mathcal{L}_{\partial_v}\right)^{j_1}W(v)\right\vert\right\vert_{C^{j_2}\left(\mathbb{S}^2\right)} \lesssim \epsilon,\]
where we use a fixed choice of a round metric $\mathring{\slashed{g}}$ to compute the norm above (and also below).

Then we have the following: 
\begin{enumerate}
\item If $k = A = 0$, $W = 0$, and the right hand side of~\eqref{2om3om2o4} is finite,  then for every choice of $\tilde{\psi} \doteq \psi|_{v=c}$, there exists a unique solution $\psi$ which will satisfy the following bounds for any $i \in [0,j]$:
\begin{align}\label{2om3om2o4}
&\sum_{i_1+i_2 \leq i}\left\vert\left\vert \left(v\mathcal{L}_{\partial_v}\right)^{i_1}\psi(v)\right\vert\right\vert_{C^{i_2}(\mathbb{S}^2)} \lesssim_i 
\\ \nonumber &\qquad v^{-D(i)\epsilon} \left(\left(1+\left|\log(v)\right|\right)\sup_v\sum_{i_1+i_2\leq i}\left\vert\left\vert\left(v\mathcal{L}_{\partial_v}\right)^{i_1} H\right\vert\right\vert_{C^{i_2}\left(\mathbb{S}^2\right)}+ \left\vert\left\vert \tilde{\psi}\right\vert\right\vert_{C^i\left(\mathbb{S}^2\right)}\right),
\end{align}
for a sequence of non-negative constants $D(i)$ so that $D(0) = 0$. 

Furthermore, if the right hand side of~\eqref{2om3omo2} is finite, then there also exists a solution $\psi$ such that 
\begin{equation}\label{2om3omo2}
\sum_{i_1+i_2\leq i}\left\vert\left\vert \left(v\mathcal{L}_{\partial_v}\right)^{i_1}\psi(v)\right\vert\right\vert_{C^{i_2}\left(\mathbb{S}^2\right)} \lesssim_{p^{-1}} v^p\sup_v\left[ v^{-p}\sum_{i_1+i_2 \leq i}\left\vert\left\vert \left(v\mathcal{L}_{\partial_v}\right)^{i_1}H\right\vert\right\vert_{C^{i_2}\left(\mathbb{S}^2\right)}\right].
\end{equation}
\item If $k = 0$, $A = 1$, $W = 0$, and the right hand side of~\eqref{oiio19281} is finite,  then for every choice of $\tilde{\psi} \doteq \psi|_{v=c}$, there exists a unique solution $\psi$ which will satisfy the following bound:
\begin{align}\label{oiio19281}
v^{-1}\left\vert\left\vert \psi\left(v\right) \right\vert\right\vert_{L^{\infty}\left(\mathbb{S}^2\right)} \lesssim \int_v^c \tau^{-2}\left\vert\left\vert H(\tau)\right\vert\right\vert_{L^{\infty}\left(\mathbb{S}^2\right)}\, d\tau + c^{-1}\left\vert\left\vert \tilde{\psi}\right\vert\right\vert_{L^{\infty}\left(\mathbb{S}^2\right)}.
\end{align}
\item If $A = 0$, and the right hand side of~\eqref{198ijnbhuknbghj} is finite,  then for every choice of $\tilde{\psi} \doteq \psi|_{v=c}$, there exists a unique solution $\psi$ which will satisfy the following bounds for any $i \in [0,j]$:
\begin{align}\label{198ijnbhuknbghj}
&\sum_{i_1+i_2 \leq i}\left\vert\left\vert \left(v\mathcal{L}_{\partial_v}\right)^{i_1}\psi(v)\right\vert\right\vert_{C^{i_2}(\mathbb{S}^2)} \lesssim_i 
\\ \nonumber &\qquad v^{-D(i)\epsilon} \left(\left(1+\left|\log(v)\right|\right)\sup_v\sum_{i_1+i_2\leq i}\left\vert\left\vert\left(v\mathcal{L}_{\partial_v}\right)^{i_1} H\right\vert\right\vert_{C^{i_2}\left(\mathbb{S}^2\right)}+ \left\vert\left\vert \tilde{\psi}\right\vert\right\vert_{C^i\left(\mathbb{S}^2\right)}\right),
\end{align}
for a sequence of non-negative constants $D(i)$ (where it may be the case that $D(0) > 0$).

Furthermore, if the right hand side of~\eqref{21ljoi1jo21} is finite, then there will also exist a solution $\psi$ such that
\begin{equation}\label{21ljoi1jo21}
\sum_{i_1+i_2\leq i}\left\vert\left\vert \left(v\mathcal{L}_{\partial_v}\right)^{i_1}\psi(v)\right\vert\right\vert_{C^{i_2}\left(\mathbb{S}^2\right)} \lesssim_{p^{-1}} v^p\sup_v\left[ v^{-p}\sum_{i_1+i_2 \leq i}\left\vert\left\vert \left(v\mathcal{L}_{\partial_v}\right)^{i_1}H\right\vert\right\vert_{C^{i_2}\left(\mathbb{S}^2\right)}\right].
\end{equation}

\end{enumerate}

\end{lemma}
\begin{proof}We start with the proof of~\eqref{2om3om2o4}. Changing variables to $s = -\log(v)$, the equation becomes
\begin{equation}\label{2m4o2}
\left(\mathcal{L}_{\partial_s}-\mathcal{L}_q\right)\psi = -H.
\end{equation}
Letting, $\alpha$ denote a suitable multi-index, and $Z^{(\alpha)}$ then a product of angular momentum operators on $\mathbb{S}^2$, we may also commute~\eqref{2m4o2} with $\mathcal{L}_{Z^{(\alpha)}}$ to obtain
\begin{equation}\label{23412}
\left(\mathcal{L}_{\partial_s}-\mathcal{L}_q\right)\left(\mathcal{L}_{Z^{(\alpha)}}\psi\right) - \left[\mathcal{L}_{Z^{(\alpha)}},\mathcal{L}_q\right]\psi= -\mathcal{L}_{Z^{(\alpha)}}H.
\end{equation}
In either case we may simply integrate towards $s = \infty$ along the integral curves of the vector field $\partial_s - q$ starting from $s = -\log(c)$,  in the case of~\eqref{23412}, sum over $\alpha$ and apply Gr\"{o}nwall's inequality, and change back to the $v$ variable to obtain, for any $i \in [0,j]$:
\begin{equation}\label{3ioj1oiji4}
\left\vert\left\vert \psi(v) \right\vert\right\vert_{L^{\infty}\left(\mathbb{S}^2\right)} \lesssim \left(1+|\log(v)|\right) \sup_v\left\vert\left\vert H(v)\right\vert\right\vert_{L^{\infty}\left(\mathbb{S}^2\right)} + \left\vert\left\vert \tilde{\psi}\right\vert\right\vert_{L^{\infty}\left(\mathbb{S}^2\right)},
\end{equation}
\begin{equation}\label{oo1ipio214}
\left\vert\left\vert \psi(v) \right\vert\right\vert_{C^i\left(\mathbb{S}^2\right)} \lesssim v^{-\epsilon D(i)}\left(\left(1+\left|\log(v)\right|\right)\sup_v\left\vert\left\vert H(v)\right\vert\right\vert_{C^i\left(\mathbb{S}^2\right)} + \left\vert\left\vert \tilde{\psi}\right\vert\right\vert_{C^i\left(\mathbb{S}^2\right)}\right).
\end{equation}
Combining these estimates with the use of the equation (and $v\partial_v$ differentiated versions thereof) directly to estimate $v\mathcal{L}_{\partial_v}$ derivatives of the solution then leads to~\eqref{2om3om2o4}.

For the second bound~\eqref{2om3omo2}, we let $s_{\infty} \gg 1$ be an arbitrarily large constant and integrate~\eqref{2m4o2} along the integral curves of $q$ backwards from $s_{\infty}$ to define a solution $\psi_{s_{\infty}}\left(s,\theta^A\right) : (-\log(c),s_{\infty})\times\mathbb{S}^2 \to\mathbb{R}$ which satisfies $\psi_{s_{\infty}}|_{s = s_{\infty}} = 0$. Commuting with angular derivatives and applying the fundamental theorem of calculus leads to, for each $s \in (-\log(c),s_{\infty})$,
\[\left\vert\left\vert \psi_{s_{\infty}}(s)\right\vert\right\vert_{C^i\left(\mathbb{S}^2\right)} \lesssim \epsilon\int_s^{s_{\infty}}\left\vert\left\vert \psi_{s_{\infty}}(\tau)\right\vert\right\vert_{C^i\left(\mathbb{S}^2\right)}\, d\tau + e^{-ps}\sup_v\left[ (-v)^{-p}\left\vert\left\vert H\right\vert\right\vert_{C^i}\right].\]
Applying the ``reverse''-Gr\"{o}nwall's inequality\footnote{Namely, if $u(t) \leq \alpha(t) + \int_t^B\beta(s)u(s)\, ds$ for $t \in [A,B]$, $u \geq 0$, $\alpha \geq 0$, and $\beta \geq 0$, then
\[u(t) \leq \alpha(t) + \exp\left(\int_t^B\beta(\tau)\, d\tau\right)\int_t^B\left[\alpha(\tau)\beta(\tau)\exp\left(-\int_{\tau}^B\beta(s)\, ds\right)\right]\, d\tau,\qquad \forall t \in [A,B].\]
} and using that $\epsilon p^{-1} \ll 1$, we obtain the bound 
\[\left\vert\left\vert \psi_{s_{\infty}}(s)\right\vert\right\vert_{C^i\left(\mathbb{S}^2\right)} \lesssim  e^{-ps}\sup_v\left[ (-v)^{-p}\left\vert\left\vert H\right\vert\right\vert_{C^i}\right].\]
It is then straightforward to use these bounds to obtain a limiting solution $\psi \doteq \lim_{s_{\infty}\to \infty}\psi_{s_{\infty}}$ which will satisfy the desired estimates. As above, the equation may be used to directly estimate $v\mathcal{L}_{\partial_v}$ derivatives in terms of angular derivatives and $H$. 

In order to establish~\eqref{oiio19281} we first multiply the equation through by $v^{-1}$ to derive an equation for $v^{-1}\psi$ of the type  discussed above. Then~\eqref{oiio19281} follows immediately by integrating along the integral curves of $\partial_s+q$. Finally,~\eqref{198ijnbhuknbghj} and~\eqref{21ljoi1jo21} are established in the same fashion \emph{mutatis mutandis} as we established~\eqref{2om3om2o4} and~\eqref{2om3omo2}. 
\end{proof}

\subsection{A Nonlinear Degenerate Transport Equation}\label{nonlineardegentran}
In this section we will study a nonlinear degenerate transport equation that will arise in our discussion of canonical gauges in Section~\ref{3joi99192} where we will need to construct suitable self-similar solutions to the eikonal equation.
\begin{proposition}\label{solvedegennonlinear}Let $\hat{c} > 0$, $N$ be a sufficiently large integer, and $H\left(\hat{v},\theta^A\right) : [0,\hat{c}) \times \mathbb{S}^2 \to \mathbb{R}$  be a given function with $H\left(0,\theta^A\right) = 0$, $h^A(\hat{v})$ be an $\mathbb{S}^2$ vector field for each $\hat{v} \in [0,\hat{c})$, and $\slashed{g}_{AB}(\hat{v})$ be an $\mathbb{S}^2$ Riemannian metric for each $\hat{v} \in [0,\hat{c})$ such that 
\begin{align}\label{oiiuiu328932o}
&\sum_{i=0}^1\sum_{j=0}^{N-1}\sup_v\left\vert\left\vert |\hat{v}|^{di+j}\mathcal{L}_{\partial^{i+j}_{\hat{v}}}\left(h,\slashed{g}-(\hat{v}+1)^2\mathring{\slashed{g}}\right)\right\vert\right\vert_{C^{N-i-j}\left(\mathbb{S}^2\right)} 
\\ \nonumber &\qquad \qquad + \sum_{i=0}^1\sum_{k=0}^1\sum_{j=0}^{N-2}\sup_v\left\vert\left\vert |\hat{v}|^{dk+j}\mathcal{L}_{\partial^{i+j+k}_{\hat{v}}}H\right\vert\right\vert_{C^{N-i-j-k}\left(\mathbb{S}^2\right)}  \doteq \mathscr{A} \ll {\rm min}\left(d,N^{-1}\right),
\end{align}
where we use the metric $\mathring{\slashed{g}}_{AB}$ to compute the norm of tensors along $\mathbb{S}^2$ and $d > 0$ satisfies $d \ll 1$. 

Let $a \in \mathbb{R}$ satisfy $|a| \lesssim 1$. Then there exists a function $v\left(\hat{v},\theta^A\right) : \mathcal{R} \to (0,\infty)$, for some open $\mathcal{R} \subset (0,\hat{c})\times \mathbb{S}^2$ with $(0,\tilde{c}) \times \mathbb{S}^2 \subset \mathcal{R}$ for some $\tilde{c} > 0$ such that the map $\left(\hat{v},\theta^A\right) \mapsto \left(v,\theta^A\right)$ is a $C^1$-diffeomorphism onto its range which is equal to $\left(0,c\right)\times \mathbb{S}^2$ for some constant $c > 0$, $\lim_{\hat{v}\to 0}v = 0$, and so that, if we consider $\hat{v}$ as a function of $v$, then
\begin{equation}\label{3o29u9o}
(-v)\mathcal{L}_{\partial_v}\hat{v} + \hat{v} - h^A\slashed{\nabla}_A\hat{v} +a\left|\slashed{\nabla}\hat{v}\right|_{\slashed{g}}^2= H,
\end{equation}
where $h$, $\slashed{g}$, and $H$ are all evaluated at $\left(\hat{v}\left(v,\theta^A\right),\theta^A\right)$. 

Let $q > 0$ satisfy $\mathscr{A} \ll q \ll {\rm min}\left(d,N^{-1}\right)$. Then, considered as a function of $v$, $\hat{v}$ will satisfy the following estimates:
\begin{equation}\label{jwiojoqr}
\sup_v \left\vert\left\vert v^{1-q}\left(v-\hat{v}\right)\right\vert\right\vert_{C^{N-1}\left(\mathbb{S}^2\right)} \lesssim_q \mathscr{A},
\end{equation}
\begin{equation}\label{3oij2oij4o}
\sum_{j=0}^{\tilde{N}-1} \sup_v\left\vert\left\vert v^{\tilde{d}+j}\mathcal{L}_{\partial_v}^j\left(\log\mathcal{L}_{\partial_v}\hat{v}\right)\right\vert\right\vert_{C^{N-2-j}\left(\mathbb{S}^2\right)} + \sup_{(v,\theta^A)}\left|\frac{\log\mathcal{L}_{\partial_v}\hat{v}}{\log\left(v\right)}\right|\lesssim \mathscr{A}
\end{equation}
for a constant $\tilde{d}$ which satisfies $|\tilde{d}| \lesssim d+\mathscr{A}$.
\end{proposition}

\begin{proof}Let $c > 0$ be a fixed constant with $\mathscr{A} \ll c$. An application of the method of characteristics  shows that there exists a solution $\hat{v}\left(v,\theta^A\right) : (0,c) \times \mathbb{S}^2 \to (0,\infty)$ to~\eqref{3o29u9o} which satisfies 
\begin{equation}\label{ooi1031}
\lim_{v \to c^-}\hat{v} = c,\qquad {\rm sup}_{\left(v,\theta^A\right) \in (0,c)\times \mathbb{S}^2}\frac{v^{1+q}}{\hat{v}} \lesssim 1,\qquad \sup_{v \in (0,c)} v^{-1+q}\left\vert\left\vert \hat{v}-v\right\vert\right\vert_{\mathring{H}^{N-1}\left(\mathbb{S}^2_v\right)} \lesssim_q \mathscr{A},
\end{equation}
if we can establish suitable a priori estimates for solutions to~\eqref{3o29u9o}. More specifically it suffices to establish the following for a suitable constant $C_{\rm boot} > 0$: Let $a \in (0,c)$ and suppose that $\hat{v} : (a,c) \times \mathbb{S}^2 \to (0,\infty)$ solves~\eqref{3o29u9o} and satisfies
 \begin{equation}\label{nicebootstrapassumption}
 \lim_{v \to c^-}\hat{v} = c, \qquad \sup_{v \in (0,c)} v^{-1+q}\left\vert\left\vert \hat{v}-v\right\vert\right\vert_{C^{N-1}\left(\mathbb{S}^2_v\right)} \leq  2C_{\rm boot}\mathscr{A}.
 \end{equation}
 Then we in fact have 
 \begin{equation}\label{tobeshownboot}
 {\rm sup}_{\left(v,\theta^A\right) \in (a,c)\times \mathbb{S}^2}\frac{v^{1+q}}{\hat{v}} \lesssim 1,\qquad \sup_{v \in (a,c)} v^{-1+q}\left\vert\left\vert \hat{v}-v\right\vert\right\vert_{C^{N-1}\left(\mathbb{S}^2_v\right)} \leq  C_{\rm boot}\mathscr{A}.
 \end{equation}

We now will show that this is the case. We start with the first estimate in~\eqref{tobeshownboot}. Note that as a consequence of~\eqref{oiiuiu328932o}, we have
\begin{equation}\label{9901}
\left|H\left(\hat{v},\theta^A\right)\right| \lesssim \mathscr{A}\hat{v}.
\end{equation} 
Dividing~\eqref{3o29u9o} by $v^{-1-q}$ leads to the following equation for for $p \doteq \hat{v} v^{-1-q}$:
\begin{equation}\label{3pokpokpo2}
(-v)\mathcal{L}_{\partial_v}p -qp - h^A\slashed{\nabla}_Ap +a\slashed{\nabla}\hat{v}\cdot \slashed{\nabla}p = v^{-1-q}H.
\end{equation}
Now  we claim that $p$ must attain its minimum on $c$. To see this, suppose for the sake of contradiction that $p$ attains its minimum at a point $\left(v_{\rm min},\theta^A_{\rm min}\right)$ with $v_{\rm min } < c$. By the first derivative test, we have that $\slashed{\nabla}p|_{\left(v_{\rm min},\theta^A_{\rm min}\right)} = 0$ and that $(-v)\mathcal{L}_{\partial_v}p|_{\left(v_{\rm min},\theta^A_{\rm min}\right)} \leq 0$. This implies that at the point $\left(v_{\rm min},\theta^A_{\rm min}\right)$
\[q p + v^{-1-q}H \leq 0.\]
In view of~\eqref{9901} and the assumption that $\hat{v} > 0$, we reach a contradiction.

We now turn to establishing the second estimates of~\eqref{tobeshownboot}. The function $\hat{v}\left(v,\theta^A\right) = v$ lies in the kernel of the operator on the left hand side of~\eqref{3o29u9o}. We will linearize around this solution by setting $\tilde{v} \doteq \hat{v} - v$ and studying the equation 
\begin{equation}\label{jiu3r2ui3rui}
(-v)\mathcal{L}_{\partial_v}\tilde{v}+ \tilde{v} - h^A\slashed{\nabla}_A\tilde{v} +a\left|\slashed{\nabla}\tilde{v}\right|_{\slashed{g}}^2= H.
\end{equation}

For a multi-index $\alpha$, we let $Z^{(\alpha)}$ denote a product of angular momentum operators on $\mathbb{S}^2$. We may differentiate the equation~\eqref{jiu3r2ui3rui} with $\mathcal{L}_{Z^{(\alpha)}}$ to produce an equation of the following form:
\begin{equation}\label{3oijoirio4}
(-v)\mathcal{L}_{\partial_v} \tilde{v}^{(\alpha)} + \tilde{v}^{(\alpha)} -h^A\slashed{\nabla}_A\tilde{v}^{(\alpha)} +2a\slashed{g}^{AB}\slashed{\nabla}_A\tilde{v}\slashed{\nabla}_B\tilde{v}^{(\alpha)}= H^{(\alpha)} + \mathscr{E}\left(\alpha\right),
\end{equation} 
where $\tilde{v}^{(\alpha)} \doteq \mathcal{L}_{Z^{(\alpha)}}\tilde{v}$ and $H^{(\alpha)} \doteq \mathcal{L}_{Z^{(\alpha)}}H$. In view of the assumed bound~\eqref{oiiuiu328932o}, the first estimate of~\eqref{tobeshownboot}, and the bootstrap assumption~\eqref{nicebootstrapassumption}, we have the following estimates for $\mathscr{E}\left(\alpha\right)$ and $H^{(\alpha)}$ for $\left|\alpha\right| \leq N-1$:
\begin{align}\label{3o24u89h284h892}
\left\vert\left\vert \mathscr{E}(\alpha) \right\vert\right\vert_{L^{\infty}\left(\mathbb{S}^2_v\right)} &\lesssim \mathscr{A}\left\vert\left\vert \tilde{v}\right\vert\right\vert_{C^{|\alpha|}},\qquad \left\vert\left\vert H^{(\alpha)}\right\vert\right\vert_{L^{\infty}\left(\mathbb{S}^2_v\right)} \lesssim \mathscr{A}\left\vert\left\vert \hat{v} \right\vert\right\vert_{C^{|\alpha|}},
\end{align}   
where we recall that $H$ is evaluated at $\left(\hat{v}\left(v,\theta^A\right),\theta^A\right)$. We then apply the estimate~\eqref{oiio19281} and Gr\"{o}nwall's inequality to obtain that 
\begin{equation}\label{3oij2oij4oi2}
\sup_{v \in (a,c)}v^{-2+2q}\left\vert\left\vert \tilde{v}\right\vert\right\vert_{C^{N-1}\left(\mathbb{S}^2_v\right)}^2 \lesssim \mathscr{A}^2.
\end{equation}
where the implied constant does not depend on $C_{\rm boot}$. This thus finishes establishing~\eqref{tobeshownboot}, and we obtain the existence of a solution $\hat{v}\left(v,\theta^A\right) : (0,c) \times \mathbb{S}^2 \to (0,\infty)$ to~\eqref{3o29u9o} which satisfies~\eqref{ooi1031}. We then note that using the equation directly to estimate the $\mathcal{L}_{\partial_v}$ derivatives of $\tilde{v}$ from angular derivatives of $\tilde{v}$ then yields the bound
\begin{equation}\label{3ij2oij4oi42}
\sup_{v \in (0,c)} v^{-2+2q}\sum_{j=0}^N v^{2j}\left\vert\left\vert \mathcal{L}^j_{\partial_v}\tilde{v} \right\vert\right\vert_{C^{N-1-j}\left(\mathbb{S}^2_v\right)}^2 \lesssim \mathscr{A}^2.
\end{equation}

We now turn to improving our estimate for $\mathcal{L}_{\partial_v}\hat{v}$. Whenever $\mathcal{L}_{\partial_v}\hat{v} \neq 0$, differentiating~\eqref{jiu3r2ui3rui} yields the following equation for $\mathcal{L}_{\partial_v}\hat{v}$:
\begin{align}\label{ekj2oh3oho2}
&(-v)\mathcal{L}_{\partial_v}\left(\mathcal{L}_{\partial_v}\hat{v}\right)-  h^A\slashed{\nabla}_A\mathcal{L}_{\partial_v}\hat{v} +2a\slashed{g}^{AB}\slashed{\nabla}_A\hat{v}\slashed{\nabla}_B\mathcal{L}_{\partial_v}\hat{v}
\\ \nonumber &\qquad = \mathcal{L}_{\partial_v}\hat{v}\left(\mathcal{L}_{\partial_{\hat{v}}}H-2a\mathcal{L}_{\partial_{\hat{v}}}\left(\slashed{g}^{AB}\right)\slashed{\nabla}_A\hat{v}\slashed{\nabla}_B\hat{v}-\mathcal{L}_{\partial_{\hat{v}}}h^A\slashed{\nabla}_A\hat{v} \right).
\end{align}
In view of the bound~\eqref{3ij2oij4oi42}, we in particular have that $\mathcal{L}_{\partial_v}\hat{v} > 0$ holds for $v$ sufficiently near $c$. Whenever $\mathcal{L}_{\partial_v}\hat{v} > 0$ we may then derive that 
\begin{align}\label{ekj2oh3oho2lijij}
&(-v)\mathcal{L}_{\partial_v}\left(\log\mathcal{L}_{\partial_v}\hat{v}\right)-  h^A\slashed{\nabla}_A\log\mathcal{L}_{\partial_v}\hat{v} +2a\slashed{g}^{AB}\slashed{\nabla}_A\hat{v}\slashed{\nabla}_B\log\mathcal{L}_{\partial_v}\hat{v}
\\ \nonumber &\qquad = \left(\mathcal{L}_{\partial_{\hat{v}}}H-2a\mathcal{L}_{\partial_{\hat{v}}}\left(\slashed{g}^{AB}\right)\slashed{\nabla}_A\hat{v}\slashed{\nabla}_B\hat{v}-\mathcal{L}_{\partial_{\hat{v}}}h^A\slashed{\nabla}_A\hat{v} \right).
\end{align}
This is a linear transport equation for $\log\mathcal{L}_{\partial_v}\hat{v}$, and, in particular, we conclude that $\log\mathcal{L}_{\partial_v}\hat{v}$ must in fact be finite for all $v \in (0,c]$. This then implies that $\mathcal{L}_{\partial_v}\hat{v} > 0$ holds for all $v \in (0,c]$. Applying Lemma~\ref{laptrans} we then obtain that
\[\sup_{v \in (0,c]}\sum_{j=0}^N v^{2\tilde{d}+2j}\left\vert\left\vert \mathcal{L}^j_{\partial_v}
\mathcal{L}_{\partial_v}\log\tilde{v}\right\vert\right\vert_{C^{N-1-j}\left(\mathbb{S}^2_v\right)}^2 \lesssim \mathscr{A}^2,\]
\[\sup_{v \in (0,c]}\left|\frac{\log\left(\mathcal{L}_{\partial_v}\hat{v}\right)}{\log\left(v\right)}\right| \lesssim \mathscr{A}.\]
The lemma then follows from an application of the implicit function theorem.

\end{proof}

\section{Twisted Self-Similar Spacetimes: Definitions and Canonical Gauges}\label{3joi99192}
Throughout this section we will use $\epsilon$ to denote a small constant. Our convention is that, unless said otherwise, $\epsilon$ may always be assumed sufficiently small depending on any implied constants or other fixed constants.

\subsection{Precise Definition of Twisted Self-Similarity and our Canonical Gauges}
The following defines the class of spacetimes which will ultimately be covered by our expansions.
\begin{definition}\label{ijo93u92}Let  $\epsilon > 0$ be a positive constant which is sufficiently small. Then we say that a $3+1$ dimensional Lorentzian manifold $\left(\mathcal{M},g\right)$ is an  $\epsilon$-twisted self-similar spacetime if 
\begin{enumerate}
	\item\label{firtimte} There exists some $c > 0$ so that $\mathcal{M}$ is diffeomorphic to $(-c,c) \times (-\infty,0) \times \mathbb{S}^2$. We will use coordinates $\left(q,s,\theta^A\right) \in (-c,c) \times (-\infty,0) \times \mathbb{S}^2$ to refer to this decomposition and denote  the copy of $(-\infty,0) \times \mathbb{S}^2$ at a value of $q$ by $\mathcal{H}_q$.  
	\item There exists a smooth vector field $K$ which satisfies $\mathcal{L}_Kg = 2g$ and so that the orbits associated to the flow of $K$ are complete and are given by a curves of constant $q$ and $\theta^A$. 
	\item\label{3ioijoi23} The metric $g$ lies in $C^0_qC^2_{s,\theta^A}$ and the restriction of $g$ to $\mathcal{M}\setminus \{q = 0\}$ is $C^2$. Moreover,  we have that whenever $s$ is restricted to lie in a compact set of $(-\infty,0)$,
	\[|q|^{\tilde{C}}\mathcal{L}_{\partial_q}g\]
	extends as a $C^0_qC^1_{s,\theta^A}$ tensor to $\{q = 0\}$ where it vanishes. Here $\tilde{C} > 0$ is a suitably small constant independent of $\epsilon$.	\item\label{thesphere} There exists a choice of one of the copies of $\mathbb{S}^2 \subset \mathcal{H}_0$ such that if  we call this hypersurface $\mathcal{S}$, then the induced metric $\slashed{g}$ on $\mathcal{S}$ satisfies $\left|\slashed{g}-\mathring{\slashed{g}}\right|_{C^N\left(\mathcal{S}\right)} \lesssim \epsilon$ and $\left\vert\left\vert \Pi_{T\mathcal{S}}K\right\vert\right\vert_{C^N\left(\mathcal{S}\right)} \lesssim \epsilon$ for some positive integer $N \gg 1$. Here $\mathring{\slashed{g}}$ refers to a choice of a round metric on $\mathbb{S}^2$, $\Pi_{T\mathcal{S}}: T\mathcal{H}_0 \to T\mathcal{S}$ denotes the orthogonal projection to the tangent bundle of $\mathcal{S}$, $C^N$ is defined with respect to the round metric $\mathring{\slashed{g}}$. 
\end{enumerate}

Sometimes it will be useful to replace the interval $(-c,c)$ in item~\ref{firtimte} with half-open intervals $(-c,0]$ or $[0,c)$. In this case we will refer to $\left(\mathcal{M},g\right)$ as a $1$-sided $\epsilon$-twisted self-similar spacetime.
\end{definition}

In order to undertake a detailed study of $\epsilon$-twisted self-similar spacetimes, we will need to fix a gauge. It will in fact be convenient to work with two separate gauges. Below we define the \emph{homothetic gauge}, which, \emph{mutatis mutandis} amounts to the requirement that $\left(\mathcal{M},g\right)$ is in \emph{normal form} in the sense of Definition 2.7 from~\cite{FG2}. The primary advantage of this gauge is that it will turn out to be regular across the null hypersurface $\mathcal{H}_0$, and it is thus natural to state our global regularity assumptions in the homothetic gauge. (In the case of $1$-sided $\epsilon$-twisted self-similarity, it will be regular up to and including the null hypersurface $\mathcal{H}_0$.)
\begin{definition}\label{3moi09989}We say that an $\epsilon$-twisted self-similar spacetime $\left(\mathcal{M},g\right)$ is in the homothetic gauge if $\mathcal{M}$ is covered by a coordinate system $\left(\rho,t,\theta^A\right) \in (-\tilde{c},\tilde{c}) \times (-\infty,0) \times \mathbb{S}^2$ where $\tilde{c} > 0$, the self-similar vector field $K$ takes the form
\[K = t\partial_t,\]
the metric $g$ takes the form
\begin{equation}\label{homotheticformyay}
g = P dt^2 + 2t\left(dt\otimes d\rho + d\rho \otimes dt\right) - h_A\left(dt\otimes d\theta^A + d\theta^A\otimes dt\right) + \slashed{g}_{AB}d\theta^A\otimes d\theta^B,
\end{equation}
where $P$ is a suitable function, $h_A$ is a suitable $1$-form which takes values in the cotangent bundle of the $\mathbb{S}^2$ at a given value of $\left(\rho,t\right)$, $\slashed{g}_{AB}$ is a suitable Riemannian metric on the $\mathbb{S}^2$ at any given value of $\left(\rho,t\right)$, the null hypersurface $\mathcal{H}_0$ (see Definition~\ref{ijo93u92}) corresponds to the hypersurface $\{\rho = 0\}$, and we require that $P|_{\{\rho = 0\}} = \left|h\right|_{\slashed{g}}^2|_{\{\rho = 0\}}$.

We say that an $\epsilon$-twisted self-similar spacetime is an $\left(\epsilon,N\right)$-regular twisted self-similar spacetime if it may be put into a homothetic gauge where the metric components have the following regularity:
\[P, \slashed{g}_{AB} \in C^0_{\rho}C^N_{t,\theta^A},\qquad P - \left|h\right|^2_{\slashed{g}}, h^{\sharp} \in \cup_{j=0}^1 C^j_{\rho}C^{N-j}_{t,\theta^A},\]
where $\left(h^{\sharp}\right)^A \doteq \slashed{g}^{AB}h_B$.

We further assume that there exists a small constant $d > 0$ (which is sufficiently small depending on $N$) such that
\[\sum_{j=0}^{N-1}|\rho|^{d+j}\mathcal{L}_{\partial^{1+j}_{\rho}}\left(\slashed{g},P,h\right) \in C^0_{\rho}C^{N-1-j}_{t,\theta^A}\]
and takes the value  $0$ at $\{\rho = 0\}$. Finally, we assume the following global smallness conditions:
\begin{align}\label{3ij2ioji}
&\sup_{t=-1}\sum_{i=0}^1\sum_{j=0}^{N-i}\sum_{\left|\alpha\right| \leq N-i-j}\left||\rho|^{id+j}\mathcal{L}^{i+j}_{\partial_{\rho}}\mathcal{L}_{t,\theta^A}^{(\alpha)}\left(\slashed{g} - t^2\left(1-\rho\right)^2\mathring{\slashed{g}},P-4\rho\right)\right|_{\mathring{\slashed{g}}}\lesssim \epsilon,
\end{align}
\begin{align}\label{3ij2ioji123}
&\sup_{t=-1}\sum_{i=0}^1\sum_{k=0}^1\sum_{j=0}^{N-i-k}\sum_{\left|\alpha\right| \leq N-i-j-k}\left| |\rho|^{j+dk}\mathcal{L}^{i+j+k}_{\partial_{\rho}}\mathcal{L}_{t,\theta^A}^{(\alpha)}\left(P-4\rho-\left|h\right|^2_{\slashed{g}},h^{\sharp}\right)\right|_{\mathring{\slashed{g}}} \lesssim \epsilon.
\end{align}
In the above we use some specification of a round metric $\mathring{\slashed{g}}_{AB}$ on each sphere along $\{t = -1\}$ so that $\mathcal{L}_{\rho}\mathring{\slashed{g}}_{AB} = 0$, and $\alpha$ denotes a suitable multi-index.

We make the obvious modifications for $1$-sided $\epsilon$-regular twisted self-similarity.
\end{definition}
\begin{remark}Even though the smallness conditions~\eqref{3ij2ioji} and~\eqref{3ij2ioji123} are only stated along $\{t = -1\}$, in view of the self-similarity of the spacetime (see~\eqref{oieiu123123123}) we may obtain suitable estimates at any point in the spacetime. 
\end{remark}
\begin{remark}The assumption~\eqref{3ij2ioji123} shows that the quantities $P-\left|h\right|^2_{\slashed{g}}$ and $h^{\sharp}$ (with the index upstairs!) have improved regularity in $\rho$ relative to the quantities $P$ or $\slashed{g}_{AB}$.
\end{remark}
\begin{remark}The smallness assumptions imply that $K = t\partial_t$ is spacelike for $\rho > 0$; however, no analogous conclusion may be drawn for $\rho < 0$. 
\end{remark}
\begin{remark}\label{normalizationsFG}When comparing explicitly our Definition~\ref{3moi09989} with Definition 2.7 from~\cite{FG2}, the reader should keep in mind that there are a few different normalizations. First of all, our $t$-coordinate runs over the negative real numbers, while in~\cite{FG2} the $t$-coordinate runs over the positive reals. Furthermore, the normalizations for $\partial_{\rho}$ along $\mathcal{H}$ are slightly different. We have that $g\left(\partial_{\rho},\partial_t\right) = 2t$ while in~\cite{FG2} they have that $g\left(\partial_{\rho},\partial_t\right) = t$. We have chosen these alternative normalizations so that our later formulas in the double-null gauge take a standard form for Minkowski space. 
\end{remark}

Most of our analysis will take place in a suitable double-null gauge where the self-similar vector field $K$ takes the simple form $u\partial_u+v\partial_v$. The key advantages of this gauge are its flexibility from the analytic point of view and also the relative simplicity/familiarity (compared to other gauges) of the algebraic relations induced by self-similarity. The key disadvantage of this gauge is that it will not extend in a regular fashion to the null hypersurface $\mathcal{H}_0$; this loss of regularity is why we have also introduced the homothethic gauge in Definition~\ref{3moi09989}. 
\begin{definition}\label{thisisselfdoublenull}We say that an $\left(\epsilon,N\right)$-regular twisted self-similar spacetime $\left(\mathcal{M},g\right)$ is in the self-similar double-null gauge with the shift in the $u$-direction if $\mathcal{M}\setminus \mathcal{H}_0$ is covered by a coordinate chart $\left(u,v,\theta^A\right) \in \mathcal{U}\times \mathbb{S}^2$ where, for some $\tilde{c} > 0$, 
\begin{equation}\label{whatisuu}
\mathcal{U} = \left\{(u,v) \in (-\infty,0) \times\mathbb{R} :  \frac{v}{-u} \in (-\tilde{c},\tilde{c})\setminus \{0\}\right\},
\end{equation}
the self-similar vector field $K$ takes the form
\[K = u\partial_u + v\partial_v,\]
and the metric $g$ takes the form
\begin{equation}\label{3k2oijo42}
g = -2\Omega^2\left(du\otimes dv + dv\otimes du\right) + \slashed{g}_{AB}\left(d\theta^A - b^Adu\right)\otimes\left(d\theta^B - b^Bdu\right),
\end{equation}
where $\Omega$ is a suitable function, $b^A$ is a suitable vector field which takes values in tangent bundle of the $\mathbb{S}^2$ at a given value of $\left(u,v\right)$, and $\slashed{g}_{AB}$ is a suitable Riemannian metric on the $\mathbb{S}^2$ at any given value of $\left(u,v\right)$. We make the obvious modifications for $1$-sided $\left(\epsilon,N\right)$-regular twisted self-similarity.

We will also refer to a metric $\left(\mathcal{U}\times \mathbb{S}^2,g\right)$ where $\mathcal{U}$ is given by~\eqref{whatisuu}, $g$ takes the form~\eqref{3k2oijo42}, and $K \doteq u\partial_u + v\partial_v$ satisfies $\mathcal{L}_Kg = 2g$, as a self-similar double-null metric. 

We may also consider self-similar double-null gauges with the shift in the $v$-direction if we replace~\eqref{3k2oijo42} with~\eqref{2l3kj2lijo42}. If we just write ``self-similar double-null gauge'' then it should be assumed that we mean for the shift to be in the $u$-direction.
\end{definition}

\subsection{Equipping a Twisted Self-Similar Spacetime with the Homothethic Gauge}
Our first goal will be to prove that any $\epsilon$-twisted self-similar spacetime, after restriction to a suitable open set containing $\mathcal{H}_0$, may be put into the homothetic gauge. The following lemma will be useful for this. 
\begin{lemma}\label{3ok001}Let $\left(\mathcal{M},g\right)$ be an $\epsilon$-twisted self-similar spacetime. Let $V$ be any smooth vector field along $\mathcal{H}_0$ such $V^q = \alpha$ never vanishes. Then there exists a unique extension of vector field $V$ defined in an open set containing $\mathcal{H}_0$ such that
\[D_VV = 0\]
holds weakly and the integral curves of $V$ lie in an appropriate function space (defined in the proof), where $D$ denotes the Levi-Civita connection of $\left(\mathcal{M},g\right)$. 

\end{lemma}
\begin{proof}Let us fix a point $p_0 \in \mathcal{H}_0$. We then desire to find a curve $x\left(\tau\right) = \left(q\left(\tau\right),s\left(\tau\right),\theta^A\left(\tau\right)\right) : [0,c) \to \mathcal{M}$ for some $c > 0$ such that $x\left(\tau\right)$ is the unique solution to the weak geodesic equation:
\begin{align}\label{3po2poj4ijp2}
x^{\lambda}\left(\tau\right) &= \int_0^{\tau} y^{\lambda}\left(w\right)\, dw + p^{\lambda}_0,
\\ \label{3ij2oij} y^{\lambda}\left(\tau\right) &= -\int_0^{\tau}\left[\Gamma^{\lambda}_{\mu\nu}\left(x\left(w\right)\right)y^{\mu}\left(w\right)y^{\nu}\left(w\right) \right]\, dw + V^{\lambda}|_{p_0},
\end{align}
where the indices range over $q$, $s$, and $\theta^A$, and $\Gamma^{\lambda}_{\mu\nu}$ denotes the Christoffel symbols of $g$.  The main potential problem in solving these equations is that the Christoffel symbols $\Gamma^{\lambda}_{\mu\nu}$ may blow up as $q \to 0$ due to the limited regularity that we assume for $\mathcal{L}_{\partial_q}g$. The resolution to this is that the allowed blow-up rate is quite slow (see item~\ref{3ioijoi23} of Definition~\ref{ijo93u92}) and thus will turn out not to pose a problem in the integral on the right hand side of~\eqref{3ij2oij}.

 We then define new variables $\left(\tilde{x}^{\lambda},\tilde{y}^{\lambda}\right)$ by $\tilde{y}^q = y^q - \alpha$, $\tilde{y}^s = y^s$, $\tilde{y}^{\theta^A} = y^{\theta^A}$, $\tilde{x}^q = x^q - \tau \alpha$, $\tilde{x}^s = x^s$, $\tilde{x}^{\theta^A} = x^{\theta^A}$. We then may equivalently think of~\eqref{3po2poj4ijp2} and~\eqref{3ij2oij} as equations for the unknowns $\tilde{x}^{\lambda}$ and $\tilde{y}^{\lambda}$. Now we may define a Banach space $B$ consisting of curves $z^{\lambda}(\tau),d^{\lambda}(\tau) : [0,c) \to \mathcal{M}$ by defining the norm
\[\left\vert\left\vert \left(z^{\lambda},d^{\lambda}\right)\right\vert\right\vert_B \doteq \sup_{\tau \in [0,c)} \left[\left|z^s(\tau)\right| + \left|z^{\theta^A}(\tau)\right|+ \tau^{-3/2}\left|z^q(\tau)\right| + \left|d^s(\tau)\right| + \left|d^{\theta^A}(\tau)\right| + \tau^{-1/2}\left|d^q\right|\right].\]
Now define $\hat{x}^{\lambda}(\tau)$ so that $\hat{x}^q = \tau \alpha$ and all other components of $\hat{x}$ vanish, and define $\hat{y}^{\lambda}(\tau)$ so that $\hat{y}^q = \alpha$ and so that all other components of $\hat{y}$ vanish. 

Then, in view of  item~\ref{3ioijoi23} of Definition~\ref{ijo93u92} and the fact that $p_0^q = 0$, it is straightforward to see that the following defines a Lipschitz continuous map $\mathscr{K} : B \to B$:
\[\left(z^{\lambda},d^{\lambda}\right) \mapsto \left(\int_0^{\tau} d^{\lambda}\left(w\right)\, dw + p^{\lambda}_0,-\int_0^{\tau}\left[\Gamma^{\lambda}_{\mu\nu}\left(\left(z^{\lambda}+\hat{x}^{\lambda}\right)\left(w\right)\right)\left(d+\hat{y}\right)^{\mu}\left(w\right)\left(d+\hat{y}\right)^{\nu}\left(w\right) \right]\, dw + V^{\lambda}|_{p_0}-\hat{y}^{\lambda}\right).\]
Moreover, if $c$ is sufficiently small, then the Lipschitz constant will be less than $1$. Since we can re-write~\eqref{3po2poj4ijp2} and~\eqref{3ij2oij} as
\[\left(\tilde{x}^{\lambda},\tilde{y}^{\lambda}\right)  = \mathscr{K}\left(\tilde{x}^{\lambda},\tilde{y}^{\lambda}\right),\]
the existence and uniqueness of the curve $x^{\lambda}(\tau)$ then follows from the contraction mapping principle. 

Since $V^q|_{\mathcal{H}_0}$ never vanishes, it is straightforward to see that the curves $x^{\lambda}(\tau)$ foliate an open set of $\mathcal{H}_0$. Setting $V = \frac{dx}{d\tau}$ thus defines (uniquely) the desired vector field in an open set around $\mathcal{H}_0$. 
\end{proof}

In this next proposition we show that any $\epsilon$-twisted self-similar spacetime $\left(\mathcal{M},g\right)$ possesses an open set around $\mathcal{H}_0$ which may be put into the homothetic gauge. We, moreover, provide an explicit procedure for the construction of these coordinates. 
\begin{proposition}\label{homotheticgaugeexist}Let $\left(\mathcal{M},g\right)$ be an $\epsilon$-twisted self-similar spacetime. Then there exists an open set $\tilde{\mathcal{M}} \subset \mathcal{M}$ which contains $\mathcal{H}_0$ and which is invariant under the flow of $K$ such that $\left(\tilde{\mathcal{M}},g\right)$ may be put in the homothetic gauge via the following procedure:
\begin{enumerate}
	\item\label{9090901902190} Let $\mathcal{S}$ be the sphere from item~\ref{thesphere} of Definition~\ref{ijo93u92}. Then we define coordinates $\left(t,\theta^A\right) \in (-\infty,0)\times \mathbb{S}^2$ covering $\mathcal{H}_0$ as follows. We define the function $t: \mathcal{H}_0 \to (-\infty,0)$ by setting $t|_{\mathcal{S}} = -1$ and then setting $K\log\left(-t\right) = 1$. For any choice of local coordinates $\{\theta^A\}$ on $\mathcal{S}$, we thus obtain a coordinate system $\left(t,\theta^A\right)$ on $\mathcal{H}_0$ by requiring that $K\left(\theta^A\right) = 0$. Let $\mathcal{S}_t$ denote the copy of $\mathbb{S}^2$ at $t$, and let $\slashed{g}$ denote the induced metric on these $\mathcal{S}_t$'s. Then the induced metric on $\mathcal{H}_0$ takes the form
	\begin{equation}\label{3opk2pok2}
	-b_A\left(dt\otimes d\theta^A+d\theta^A\otimes dt\right) + \left|b\right|_{\slashed{g}}^2dt\otimes dt + \slashed{g}_{AB}d\theta^A\otimes d\theta^B,
	\end{equation}
	where, for each $t$, $b_A$ denotes a suitable $1$-form along $\mathcal{S}_t$. 
	\item We then define a vector field $V$ along $\mathcal{H}_0$ by requiring that 
	\[g\left(V,V\right) = 0,\qquad g\left(V,K\right)|_{t= -1} = 2,\qquad \mathcal{L}_KV = 0,\]
	and then extend $V$ off of $\mathcal{H}_0$ by using Lemma~\ref{3ok001} and requiring that 
	\[D_VV = 0\]
	holds weakly where $D$ denotes the Levi--Civita connection of $g$. This extension is well defined in some open set $\tilde{\mathcal{M}} \subset \mathcal{M}$ which contains $\mathcal{H}_0$ and which is invariant under the flow of the self-similar vector field $K$. 
	\item We then define a function $\rho: \tilde{\mathcal{M}} \to \mathbb{R}$ by setting $\rho|_{\mathcal{H}_0} = 0$ and $V\rho = 1$. Finally, we extend $t$ and the $\mathcal{S}_t$-local coordinates $\{\theta^A\}$ to $\tilde{\mathcal{M}}$ by requiring that $Vt = V\theta^A = 0$. Then $\left(\tilde{\mathcal{M}},g\right)$ will be in the homothetic gauge in the coordinates $\left(\rho,t,\theta^A\right)$. 
\end{enumerate}

The form of the metric in this gauge is unique once the form of the induced metric~\eqref{3opk2pok2} along $\mathcal{H}_0$ is fixed. (The form of the induced metric along $\mathcal{H}_0$ however depends on the particular choice of the hypersurface $\mathcal{S}$.)

Finally, we note that it is straightforward to modify the construction to establish an analogous result in the case of $1$-sided self-similarity. 
\end{proposition}
\begin{proof}We note that the proof of this proposition will be a minor adaption of the proof of Proposition 2.8 from~\cite{FG2}. 

We start by verifying the form~\eqref{3opk2pok2} on the induced metric along $\mathcal{H}_0$. It is an immediate consequences of the definition of an $\epsilon$-twisted self-similar vacuum spacetime that $\left(t,\theta^A\right)$ forms a set of coordinates along $\mathcal{H}_0$. Moreover, since $\mathcal{H}_0$ is a null hypersurface by assumption and $\mathcal{S}_t$ is spacelike, there must exist a vector field $b^A$, which is everywhere tangent to $\mathcal{S}_t$, such that
\begin{equation}\label{3ojoj492}
g\left(\partial_t+b,\partial_t+b\right) = 0,\qquad g\left(\partial_t+b,\partial_{\theta^A}\right) = 0.
\end{equation}
The equations~\eqref{3ojoj492} then leads to~\eqref{3opk2pok2}.

We now turn to the definition of the vector field $V$. Along $\mathcal{H}_0$, by elementary linear algebra, the vector field $V$ is uniquely determined by the specification at the point $\{t = -1\}$ by requiring that 
\[g\left(V,V\right) = g\left(V,\partial_{\theta^A}\right) = 0,\qquad g\left(V,\partial_t+b\right) = -2.\]
We then extend $V$ to all of $\mathcal{H}_0$ by requiring that $\mathcal{L}_KV = 0$. We claim that we will then have  
\[g\left(V,V\right) = g\left(V,\partial_{\theta^A}\right) = 0,\qquad g\left(V,\partial_t+b\right) = 2t.\]
To see this, it suffices to observe that, in view of $\mathcal{L}_Kg  = 2g$,  the above quantities all satisfy the following equations:
\[Kg\left(V,V\right) = 2g\left(V,V\right),\qquad Kg\left(V,\partial_{\theta^A}\right) = 2g\left(V,\partial_{\theta^A}\right),\qquad Kg\left(V,K\right) = 2g\left(V,K\right).\]
Moreover, it is immediate that $V^q$ is nowhere vanishing on $\mathcal{H}_0$.

We may then apply Lemma~\ref{3ok001} to uniquely locally extend $V$ to some open set of $\mathcal{H}_0$ such that $D_VV = 0$. We claim now that $\mathcal{L}_KV = 0$ everywhere where we have extended $V$. To see this, we define a new vector field $\tilde{V}$ by setting $\tilde{V}$ to be equal to $V$ along the integral curve of $V$ starting at $\mathcal{S}_{-1}$ and then extending $\tilde{V}$ off of this curve by $\mathcal{L}_K\tilde{V} = 0$. This defines $\tilde{V}$ in a dilation invariant neighborhood of $\mathcal{H}_0$. Now we recall the formula (see Lemma 7.1.3 from~\cite{CK})
\[\left[\mathcal{L}_X,D_{\alpha}\right]Y_{\beta} = -\frac{1}{2}\left(D_{\alpha}\pi_{\beta\gamma} +D_{\beta}\pi_{\alpha\gamma}-D_{\gamma}\pi_{\alpha\beta}\right)Y^{\gamma},\]
which holds for for any vector field $X$ and $1$-form $Y$, where $\pi$ denotes the deformation tensor of $X$. Since the deformation tensor of $K$ is parallel (as it is a multiple of the metric tensor $g$) we thus have 
\begin{equation}\label{3oj2919}
\mathcal{L}_KD_{\tilde{V}}\tilde{V}^{\alpha} = \mathcal{L}_K\left(\slashed{g}^{\alpha\gamma}\tilde{V}^{\beta}D_{\beta}\tilde{V}_{\gamma}\right) = -2\slashed{g}^{\alpha\gamma}\tilde{V}^{\beta}D_{\beta}\tilde{V}_{\gamma} = -2D_{\tilde{V}}\tilde{V}^{\alpha}.
\end{equation}
Thus, $\left(D_{\tilde{V}}\tilde{V}\right)^{\alpha}$ satisfies a linear system of first order ordinary differential equations along the integral curves of $K$. Since it vanishes along the integral curve of $V$ starting at $\mathcal{S}_t$, we conclude that $D_{\tilde{V}}\tilde{V}$ vanishes everywhere. In turn, since we also have $\tilde{V} = V$ along $\mathcal{H}_0$, it is straightforward to use the uniqueness part of Lemma~\ref{3ok001} to see that we have that $\tilde{V}$ equals to $V$ wherever $V$ is defined, and we may thus assume now that $V$ is defined in an open set $\tilde{\mathcal{M}}$ which contains $\mathcal{H}_0$ and Lie commutes with $K$. 

We have that $K$ is tangent to $\mathcal{H}_0$ and, along $\mathcal{H}_0$ we have that 
\[K\log(-t)  - 1 = K\theta^A = 0 \Rightarrow K|_{\mathcal{H}_0} = t\partial_t.\]
To obtain the form of $K$ off of $\mathcal{H}_0$ we then observe that in view of the fact that $\mathcal{L}_VK = 0$, $Vt = V\theta^A = 0$, and $V\rho = 1$, we have that the following holds everywhere
\[V\left(K\log\left(-t\right) - 1\right) = V\left(K\rho\right) = V\left(K\theta^A\right) = 0 \Rightarrow K = t\partial_t.\] 
As a consequence of $Vt = V\theta^A = 0$ and $V\rho = 1$, we have that
\begin{equation}\label{3iooi4891919}
V = \partial_{\rho}.
\end{equation}

Next we claim that $g\left(K,V\right)$ is constant along the integral curves of $V$. Indeed, we have
\[V g\left(K,V\right) = g\left(D_VK,V\right) = \frac{1}{2}Kg\left(V,V\right) = 0.\]
Similarly, since $\mathcal{L}_V\mathcal{L}_{\partial_{\theta^A}} = 0$,
\[Vg\left(V,\partial_{\theta^A}\right) = 0.\]

With all of these various facts established, the rest of the lemma now follows in a straightforward fashion.

\end{proof}

\subsection{From the Homothetic Gauge to the Self-Similar Double-Null Gauge}
In the next few propositions we will establish certain equivalences between the homothetic gauge and the double-null gauge. We start with a useful definition.
\begin{definition}\label{thisdefissomething}Let $\epsilon > 0$ be sufficiently small and $N \in \mathbb{Z}_{>0}$ be sufficiently large. We say that a self-similar double-null metric is ``$\left(\epsilon,N\right)$-regular up to $\{v = 0\}$'' if the metric coefficients and Ricci coefficients satisfy the following bounds along $\{u = -1\}$:
\begin{align}\label{3oj2ij4ino192939}
&\sup_v\left\vert\left\vert \left(b,\slashed{g}-\left(v+1\right)^2\mathring{\slashed{g}},|v|^d\log\Omega\right)|_{u=-1}\right\vert\right\vert_{C^N\left(\mathbb{S}^2_{-1,v}\right)}
\\ \nonumber &\qquad + \sup_v\left\vert\left\vert |v|^d\left(\Omega-1\right)|_{u=-1}\right\vert\right\vert_{C^0\left(\mathbb{S}^2_{-1,v}\right)} \lesssim \epsilon,
\end{align}
\begin{equation}\label{3ij2oijoi223456}
\sum_{j=0}^1\sum_{i=0}^{N-1}\sup_v\left\vert\left\vert |v|^{jd+i}\mathcal{L}_{\partial_v}^{j+i}\left(\eta,\Omega\underline{\omega}\right)|_{u=-1}\right\vert\right\vert_{C^{N-1-i-j}\left(\mathbb{S}^2_{-1,v}\right)} \lesssim \epsilon,
\end{equation}
\begin{equation}\label{2lk3jl1kjl21}
\sum_{i=0}^{N-1}\sup_v\left\vert\left\vert |v|^{d+i}\mathcal{L}_{\partial_v}^i\left(\Omega\chi - \left(v+1\right)\mathring{\slashed{g}}\right)|_{u=-1}\right\vert\right\vert_{C^{N-1-i-j}\left(\mathbb{S}^2_{-1,v}\right)} \lesssim \epsilon,
\end{equation}
where $d$ is a constant satisfying $0 < d \ll 1$ where the smallness may depend on $N$, $\mathring{\slashed{g}}_{AB}$ is some fixed choice of a round metric, and we assume that the quantities $\left(\slashed{g},b,\eta,\Omega\underline{\omega}\right)$ all have uniquely defined continuous extensions to $\{v = 0\}$.

It will also be convenient to have a version of this definition with the roles of $u$ and $v$ swapped in the sense of Remark~\ref{shiftshift}. We say that a self-similar double-null metric is ``$\left(\epsilon,N\right)$-regular up to $\{u = 0\}$'' if the metric is in the self-similar double-null gauge with the shift in the $v$-direction, the spacetime exists in a region $c < \frac{v}{-u} < \infty$ for some $c > 0$, and coefficients and the estimates~\eqref{3oj2ij4ino192939},~\eqref{3ij2oijoi223456}, and~\eqref{2lk3jl1kjl21} all hold with each $u$ replaced by $v$, each $v$ replaced by $u$, and $\left(\chi,\underline{\omega},\eta\right) \mapsto \left(\underline{\chi},\omega,\underline{\eta}\right)$. 

If we do not explicitly say, then our convention is that $\left(\epsilon,N\right)$ regularity is understood to be up to $\{v = 0\}$.
\end{definition}

In the next proposition we will show that any $\left(\epsilon,N\right)$-regular twisted self-similar vaccum spacetime may be put into the self-similar double-null gauge which is furthermore $\left(\epsilon,N-2\right)$-regular up to $\{v = 0\}$. 
\begin{proposition}\label{makeitadouble}Let $\left(\mathcal{M},g\right)$ be an $\left(\epsilon,N\right)$-regular twisted self-similar spacetime which is in the homothetic gauge. Then $\left(\mathcal{M},g\right)$ may be put into the self-similar double-null gauge by the following procedure:
\begin{enumerate}
	\item  We define a new coordinate system $\left(u,\hat{v},\theta^A\right)$ by setting 
	\[u \doteq t,\qquad \hat{v} \doteq -t\rho.\]
	In these new coordinates, the metric takes the form
	\begin{equation}\label{3io2oi384}
	g = \left(P - 4\frac{\hat{v}}{-u}\right) du^2 - 2\left(du\otimes d\hat{v} + d\hat{v}\otimes du\right) - h_A\left(du\otimes d\theta^A + d\theta^A\otimes du\right) + \slashed{g}_{AB}d\theta^A\otimes d\theta^B,
	\end{equation}
	and $K$ takes the form
	\begin{equation}\label{3oj2oijoi}
	K = u\partial_u + \hat{v}\partial_{\hat{v}}.
	\end{equation}
	The coordinates $\left(u,\hat{v}\right)$ must lie in the set $\mathcal{U} \doteq \left\{(u,\hat{v}) : u \in (-\infty,0) : \frac{\hat{v}}{-u} \in (-\tilde{c},\tilde{c})\setminus \{0\}\right\}$, where $\tilde{c}$ is such that $\rho \in (-\tilde{c},\tilde{c})$ where $\rho$ is the coordinate from the homothetic gauge (see Definition~\ref{3moi09989}).
	\item We then define another coordinate system $\left(v,u,\theta^A\right)$ by defining $v = v\left(\hat{v},u,\theta^A\right) = u\tilde{v}\left(\frac{\hat{v}}{u},\theta^A\right)$ for a suitable function $\tilde{v}$ defined in some dilation invariant neighborhood of $\mathcal{H}_0$. Wherever this coordinate change is non-singular, the metric will take the form:
	\begin{align}\label{3poj2oijoi4}
	&g = -2\frac{\partial \hat{v}}{\partial v}\left(du\otimes dv+dv\otimes du\right) +\left(P-4\frac{\hat{v}}{-u}-4\frac{\partial \hat{v}}{\partial u}\right)du^2
	\\ \nonumber &\qquad  -\left(h_A+2\frac{\partial \hat{v}}{\partial \theta^A}\right)\left(du\otimes d\theta^A + d\theta^A\otimes du\right) +\slashed{g}_{AB}d\theta^A\otimes d\theta^B, 
	\end{align}
	and, moreover, we will have that
	\begin{equation}\label{3ij2oij4o2}
	K = u\partial_u + v\partial_v.
	\end{equation}
	\item In order to obtain a double-null coordinate system we pick the function $v$ so that, in the $\left(u,v,\theta^A\right)$ coordinate system, we have
	\begin{equation}\label{3ij2oij42}
	P-4\frac{\hat{v}}{-u}-4\frac{\partial \hat{v}}{\partial u} = \left|h+2\slashed{\nabla}\hat{v}\right|_{\slashed{g}}^2.
	\end{equation}
	\item The resulting double-null metric will be $\left(\epsilon,N-2\right)$-regular up to $\{v = 0\}$.
	\end{enumerate}

Finally, we note that it is straightforward to modify the construction to establish an analogous result in the case of $1$-sided self-similarity. 
\end{proposition}
\begin{proof} The equations~\eqref{3io2oi384}-\eqref{3ij2oij4o2} are all straightforward calculations. However, we have to justify that one may pick $v$ so that~\eqref{3ij2oij42} holds. We first observe that by self-similarity, wherever the coordinate change is well-defined, we have that
\[\frac{\partial \hat{v}}{\partial u} = \frac{1}{u}\hat{v} - \frac{v}{u}\frac{\partial \hat{v}}{\partial v}.\]
In particular, we can re-write~\eqref{3ij2oij42} as
\begin{equation}\label{3noinoi2jo4}
\frac{v}{u}\frac{\partial \hat{v}}{\partial v} - \frac{1}{u}\hat{v} - h^A\slashed{\nabla}_A\hat{v} -\left|\slashed{\nabla}\hat{v}\right|^2_{\slashed{g}} = \frac{1}{4}\left(\left|h\right|^2_{\slashed{g}} - \left(P - 4\frac{\hat{v}}{-u}\right)\right).
\end{equation}
Furthermore, in view of the self-similarity of both the $\hat{v}$ and $v$ coordinates, it suffices to define $v$ along $\{u = -1\}$, then extend $v$ by self-similarity off of $\{u = -1\}$, and guarantee that~\eqref{3noinoi2jo4} holds when $\{u = -1\}$. Now we simply observe that in view of the assumption~\eqref{3ij2ioji123} from Definition~\ref{3moi09989}, an application of Proposition~\ref{solvedegennonlinear} (and an additional application after the coordinate change $\hat{v} \mapsto -\hat{v}$) exactly implies that it is possible to pick such a function $v$ such that the corresponding coordinate system is valid for 
\[\left\{\left(u,v,\theta^A\right) \in (-\infty,0) \times \mathbb{R} \times \mathbb{S}^2: \frac{v}{-u} \in (-c,c)\setminus \{0\}\right\},\]
for a suitably small $c > 0$, and that this chart covers the complement of $\mathcal{H}_0$ in a dilation invariant neighborhood of $\mathcal{H}_0$. 

It remains to check that the double-null metric is $\left(\epsilon,\tilde{N}\right)$-regular up to $\{v = 0\}$ for some $\tilde{N} \gtrsim N$. The estimates~\eqref{3oj2ij4ino192939} and~\eqref{2lk3jl1kjl21} are straightforward consequences of our application of Proposition~\ref{solvedegennonlinear} and the original bounds for the metric from~\eqref{3ij2ioji} and~\eqref{3ij2ioji123}.  For~\eqref{3ij2oijoi223456} we need to establish improved estimates for $\Omega\underline{\omega}$ and $\eta$. We start with $\Omega\underline{\omega}$. We have 
\[\Omega^2 = \frac{\partial \hat{v}}{\partial v},\qquad b^A = h^A + 2\slashed{\nabla}^A\hat{v},\]
 and thus a short calculation using the underlying self-similarity yields that
\begin{equation}\label{formform}
\Omega\underline{\omega} = -\frac{1}{2}\left(\partial_u + \mathcal{L}_b\right)\log\Omega = -\frac{1}{4}\left(-\frac{v}{u}\partial_v + \left(h^A+2\slashed{\nabla}^A\hat{v}\right)\slashed{\nabla}_A\right)\log\left(\frac{\partial\hat{v}}{\partial v}\right).
\end{equation}
On the other hand, we may differentiate~\eqref{3noinoi2jo4} and obtain
\begin{align*}
&\frac{v}{u}\partial_v\left(\log\left(\frac{\partial \hat{v}}{\partial v}\right)\right) - h^A\slashed{\nabla}_A\log\left(\frac{\partial \hat{v}}{\partial v}\right) - 2\left(\slashed{\nabla}^A\hat{v}\right)\slashed{\nabla}_A\log\left(\frac{\partial \hat{v}}{\partial v}\right) = 
\\ \nonumber &\qquad \left(\mathcal{L}_{\partial_{\hat{v}}}h\right)^A\slashed{\nabla}_A\hat{v} + \left(\mathcal{L}_{\partial_{\hat{v}}}\slashed{g}^{-1}\right)^{AB}\slashed{\nabla}_A\hat{v}\slashed{\nabla}_B\hat{v} + \mathcal{L}_{\partial_{\hat{v}}}\left(\frac{1}{4}\left(\left|h\right|^2_{\slashed{g}} - \left(P - 4\frac{\hat{v}}{-u}\right)\right)\right).
\end{align*} 
The desired estimates for $\Omega\underline{\omega}$ then follow from the original metric bounds~\eqref{3ij2ioji} and~\eqref{3ij2ioji123} as well as the estimates for $\hat{v}$ which we obtained after the application of Proposition~\ref{solvedegennonlinear}. Finally, we turn to the improved estimates for $\eta$. We have 
\begin{align*}
\eta^A  &= -\frac{1}{4}\Omega^{-2}\mathcal{L}_{\partial_v}b^A + \slashed{\nabla}^A\log\Omega
\\ \nonumber &= -\frac{1}{4}\frac{\partial v}{\partial \hat{v}}\mathcal{L}_{\partial_v}\left(h^A + 2\slashed{\nabla}^A\hat{v}\right) + \frac{1}{2}\slashed{\nabla}^A\log\left(\frac{\partial \hat{v}}{\partial v}\right)
\\ \nonumber &= -\frac{1}{4}\mathcal{L}_{\partial_{\hat{v}}}h^A - \frac{1}{2}\left(\mathcal{L}_{\partial_{\hat{v}}}\slashed{g}^{-1}\right)^{AB}\slashed{\nabla}_B\hat{v}.
\end{align*}
Then, yet again, the desired estimates for $\eta$ follow now from the original metric bounds~\eqref{3ij2ioji} and~\eqref{3ij2ioji123} as well as the estimates for $\hat{v}$ which we obtained after the application of Proposition~\ref{solvedegennonlinear}.

\end{proof}
In this next proposition we show that given a self-similar double-null metric which is $\left(\epsilon,N\right)$-regular up to $\{v = 0\}$, then the metric in fact arises from a $\left(\epsilon,N\right)$-regular twisted self-similar vacuum spacetime in the homothetic gauge which has been put into the self-similar double-null gauge. As in Proposition~\ref{homotheticgaugeexist}, we moreover provide an explicit procedure for the construction of the corresponding homothetic gauge. 
\begin{proposition}\label{nonononononoo}Let $\epsilon > 0$ be sufficiently small, $N \in \mathbb{Z}_{>0}$ be sufficiently large, and let $\left(\mathcal{U}\times\mathbb{S}^2,g\right)$ be a self-similar double-null metric which is $\left(\epsilon,N\right)$-regular up to $\{v = 0\}$. 

Then we may define a new coordinate system $\left(\hat{v},u,\theta^A\right)$ by setting
\[\hat{v}\left(v,u,\theta^A\right) \doteq \int_0^v\Omega^2\left(\tilde{v},u,\theta^A\right)\, d\tilde{v}.\]
In the new $\left(\hat{v},u,\theta^A\right)$ coordinates, the metric takes the form
\begin{align}\label{3ij290901}
&g = -2\left(du\otimes d\hat{v} + d\hat{v}\otimes du\right) - \left(b_A-4\int_0^v\left(\Omega^2\mathcal{L}_{\partial_{\theta^A}}\log\Omega\right)\, d\tilde{v}\right)\left(du\otimes d\theta^A + d\theta^A\otimes du\right) 
\\ \nonumber &\qquad + \left(|b|^2 +8\int_0^v\left(\Omega^2\mathcal{L}_{\partial_u}\log\Omega\right)\, d\tilde{v}\right)\, du^2 + \slashed{g}_{AB}d\theta^A\otimes d\theta^B,
\end{align}
and the self-similar vector field $K$ takes the form
\begin{equation}\label{3ij2j3o42}
K = u\partial_u + \hat{v}\partial_{\hat{v}}.
\end{equation}

Finally, we define $\left(t,\rho,\theta^A\right)$ coordinates by setting
\[t \doteq u,\qquad \rho \doteq \frac{\hat{v}}{-u}.\]
In these coordinates, the metric takes the form
\begin{align}\label{32oi4810}
&g = 2t\left(dt\otimes d\rho + d\rho \otimes dt\right) +  \left(|b|^2 +8\int_0^v\left(\Omega^2\mathcal{L}_{\partial_u}\log\Omega\right)\, d\tilde{v}+4\rho\right)dt^2 
\\ \nonumber &\qquad - \left(b_A-4\int_0^v\left(\Omega^2\mathcal{L}_{\partial_{\theta^A}}\log\Omega\right)\, d\tilde{v}\right)\left(dt\otimes d\theta^A + d\theta^A\otimes dt\right) +\slashed{g}_{AB}d\theta^A\otimes d\theta^B,
\end{align}
and the self-similar vector field $K$ now takes the form
\begin{equation}\label{32oi4oo2}
K = t\partial_t.
\end{equation}

Our assumptions on the metric and Ricci coefficients show that the metric defined by~\eqref{3ij290901} in fact extends continuously to $\{\hat{v} = 0\}$. In turn, after possibly shrinking the original neighborhood where the $\left(v,u,\theta^A\right)$ are defined, we find that $g$ expressed in the $\left(t,\rho,\theta^A\right)$ will be an $\left(\epsilon,N-2\right)$-regular twisted self-similar spacetime which is in the homothetic gauge.

Finally, we note that it is straightforward to modify the construction to establish an analogous result in the case of $1$-sided self-similarity. 
\begin{proof}The formulas~\eqref{3ij290901}-\eqref{32oi4oo2} are straightforward calculations. 

It remains to check the regularity statements about $g$. The estimate~\eqref{3ij2ioji} is an immediate consequence of the fact that our original metric was $\left(\epsilon,N\right)$-regular up to $\{v = 0\}$. In order to check~\eqref{3ij2ioji123} we need to show that $P-4\rho - \left|h\right|^2_{\slashed{g}}$ and $h^A$ have improved regularity as $\rho \to 0$. We start with $P-4\rho -\left|h\right|^2_{\slashed{g}}$:
\begin{align*}
&\mathcal{L}_{\partial_{\rho}}\left(P-4\rho -\left|h\right|^2_{\slashed{g}}\right) =
\\ \nonumber &\qquad  (-u)\mathcal{L}_{\partial_{\hat{v}}}\left(8\int_0^v\left(\Omega^2\mathcal{L}_{\partial_u}\log\Omega\right)\, d\tilde{v}+8b^A\int_0^v\left(\Omega^2\mathcal{L}_{\partial_{\theta^A}}\log\Omega\right)\, d\tilde{v}-16\left|\int_0^v\left(\Omega^2\slashed{\nabla}\log\Omega\right)\, d\tilde{v}\right|^2\right) =
\\ \nonumber &\qquad -4(-u)\left(\Omega\underline{\omega}\right) + 8(-u)\mathcal{L}_{\partial_{\hat{v}}}b^A\int_0^v\left(\Omega^2\mathcal{L}_{\partial_{\theta^A}}\log\Omega\right)\, d\tilde{v} -32(-u)\slashed{\nabla}^A\log\Omega\left(\int_0^v\left(\Omega^2\slashed{\nabla}_A\log\Omega\right)\, d\tilde{v}\right)
\\ \nonumber &\qquad - 16 (-u)\mathcal{L}_{\partial_{\hat{v}}}\left(\slashed{g}^{-1}\right)^{AB}\left(\int_0^v\left(\Omega^2\slashed{\nabla}_A\log\Omega\right)\, d\tilde{v}\right)\left(\int_0^v\left(\Omega^2\slashed{\nabla}_B\log\Omega\right)\, d\tilde{v}\right).
\end{align*}
From this identity, the necessary improved regularity for $P-4\rho - \left|h\right|^2$ follows immediately from the $\left(\epsilon,N\right)$-regularity assumptions. For $h^A$ we have
\begin{align*}
\mathcal{L}_{\partial_{\rho}}h^A = (-u)\mathcal{L}_{\partial_{\hat{v}}}\left(b^A -4\slashed{g}^{AB}\int_0^v\left(\Omega^2\mathcal{L}_{\partial_{\theta^B}}\log\Omega\right)\, d\tilde{v}\right) = u\eta^A +4u\mathcal{L}_{\hat{v}}\left(\slashed{g}^{-1}\right)^{AB}\int_0^v\left(\Omega^2\mathcal{L}_{\partial_{\theta^B}}\log\Omega\right)\, d\tilde{v}.
\end{align*}
From this identity, the necessary improved regularity for $h^A$ follows immediately from the $\left(\epsilon,N\right)$-regularity assumptions. 

\end{proof}

\end{proposition}

We now turn to self-similar spacetimes whose homothetic vector filed is everywhere spacelike.
\begin{definition}\label{3901o1111}Let $N \in \mathbb{Z}_{>0}$ be sufficiently large. Then we say that a $3+1$ dimensional Lorentzian manifold $\left(\mathcal{M},g\right)$ is in the spacelike $N$-regular self-similar regime if 
\begin{enumerate}
\item  There exists some $0 < a_0 < a_1 < \infty$ so that $\mathcal{M}$ is diffeomorphic to $(a_1-a_0,a_1+a_0) \times (-\infty,0) \times \mathbb{S}^2$. We will use coordinates $\left(q,s,\theta^A\right) \in (a_1-a_0,a_1+a_0) \times (-\infty,0) \times \mathbb{S}^2$ to refer to this decomposition and denote  the copy of $(-\infty,0) \times \mathbb{S}^2$ at a value of $q$ by $\mathcal{H}_q$.  
	\item There exists a vector field $K$ which satisfies $\mathcal{L}_Kg = 2g$, so that the orbits associated to the flow of $K$ are complete and given by a curves of constant $q$ and $\theta^A$, and $K$ is everywhere spacelike.
	\item The metric $g \in C^N$. 
	\end{enumerate}

\end{definition}

We now introduce the analogue of the $\left(\epsilon,N\right)$-regularity concept for spacelike self-similar metrics.
\begin{definition}\label{defdefdefdefdeflko}We say that a spacelike self-similar metric is in the spacelike homothethic gauge if the metric $g$ takes the form~\eqref{homotheticformyay} for coordinates $\left(\rho,t,\theta^A\right) \in (\rho_0-c,\rho_0+c) \times (-\infty,0) \times \mathbb{S}^2$ and suitable $0 < c < \rho_0$ and if the self-similar vector field $K$ takes the form $t\partial_t$.

We then say that a metric $g$ is in the $\left(\epsilon,N\right)$-small spacelike self-similar regime if it is a spacelike self-similar metric and, after being put into the homothetic gauge, it satisfies the following smallness assumption relative to Minkowski space in the homothethic gauge along $\{t = -1\}:$
	\begin{align}
	&\left\vert\left\vert \left(g_{tt} - 4\rho,g_{t\rho}-2t,g_{tA},\slashed{g}_{AB} - t^2\left(\rho-1\right)^2\mathring{\slashed{g}}_{AB}\right)|_{t=-1}\right\vert\right\vert_{C^N} \lesssim \epsilon,
	\end{align}
	where $\mathring{\slashed{g}}_{AB}$ denotes a fixed choice of a round metric and we use the metric $d\rho^2 + \mathring{\slashed{g}}_{AB}$ to define the $C^N$ norm along $\{t = -1\}$. We require that $N$ be a sufficiently large positive integer.
\end{definition} 
\begin{remark}From a mild adaption of the proof of Proposition~\ref{3ok001}, it is clear that any self-similar spacelike spacetime may be put into the spacelike homothetic gauge in some dilation invariant neighborhood of any given point in the spacetime.
\end{remark}
\begin{remark}Even though the smallness assumption is only stated along $\{t = -1\}$, in view of the underlying self-similarity, this should be considered a global smallness assumption.
\end{remark}

\subsection{The Spacelike Self-Similar Regime}
In the next lemma, we show that any spacetime which is in the $\left(\epsilon,N\right)$-small spacelike self-similar regime may be put into a double-null gauge.
\begin{lemma}\label{oioioioi12}Let $\left(\mathcal{M},g\right)$ be in the $\left(\epsilon,N\right)$-small spacelike self-similar regime. Then there exist coordinates 
\begin{equation}\label{3ij2oijoi2}
\left\{\left(u,v,\theta^A\right) \in (-\infty,0) \times \mathbb{R} \times \mathbb{R}^2 : \frac{v}{-u} \in (\lambda-\tilde{c},\lambda+\tilde{c}) \right\},
\end{equation}
for some $0 < \tilde{c} < \lambda < \infty$,  covering a dilation invariant neighborhood in $\mathcal{M}$, such that $g$ takes the double-null form~\eqref{3k2oijo42}, the vector field $K$ takes the form $K = u\partial_u + v\partial_v$, and the metric satisfies the following smallness condition along $\{u = -1\}:$
\begin{equation}\label{32oijoi991}
\left\vert\left\vert \left(\log\Omega,b,\slashed{g} - \left(v-u\right)^2\mathring{\slashed{g}}\right)|_{u=-1}\right\vert\right\vert_{C^{\tilde{N}}} \lesssim \epsilon,
\end{equation}
for some $\tilde{N} \gtrsim N$ and where we use the metric $dv^2 + \mathring{\slashed{g}}$ to define the $C^N$ norm.

It will also be useful to note that the same proof allows us to equip the spacetime with the coordinate system~\eqref{2l3kj2lijo42} instead of~\eqref{3k2oijo42}, and we will still have the estimate~\eqref{32oijoi991}. 
\end{lemma}
\begin{proof}We may assume our metric is in the homothetic gauge. Then, along the sphere at $\{t = -1\} \cap \{\rho = \rho_0\}$, we set $L' = \partial_{\rho}$ and define a null vector field $\underline{L}'$ such that 
\[g\left(L',\underline{L}'\right) = -2,\qquad \left\vert\left\vert \underline{L}' - \left(\partial_t - t^{-1}\rho\partial_{\rho} \right)\right\vert\right\vert_{C^N} \lesssim \epsilon. \]
(We would be able to take $\underline{L}' - t^{-1}\rho\partial_{\rho} = 0$ on exact Minkowski space in the homothetic gauge.) Then we extend $L'$ and $\underline{L}'$ to all of $\{\rho = \rho_0\}$ by requiring that 
\begin{equation}\label{32oij482}
\mathcal{L}_KL' = -L',\qquad \mathcal{L}_K\underline{L}' = -\underline{L}'.
\end{equation}
We then extend $L'$ and $\underline{L}'$ to a neighborhood of $\{\rho = \rho_0\}$ by solving $D_{L'}L' = D_{\underline{L}'}\underline{L}' = 0$. Alternatively we define vector field $\tilde{L}'$ and $\underline{\tilde{L}}'$ in a neighborhood of $\{\rho = \rho_0\}$ by seeting $\tilde{L}' = L$ and $\underline{\tilde{L}}' = \underline{L}'$ along the integral curves of $L'$ and $\underline{L}'$ which start at $\{t = -1\} \cap \{\rho = \rho_0\}$ and then extend $\underline{\tilde{L}}'$ and $\tilde{L}'$ by requiring that~\eqref{32oij482} holds everywhere with $L'$ and $\underline{L}'$ replaced by $\tilde{L}'$ and $\tilde{\underline{L}}'$. We then claim that $\tilde{L}' = L'$ and $\underline{\tilde{L}}' = \underline{L}'$ on their common domain. To see this, it suffices to check that $D_{\tilde{L}'}\tilde{L}' = D_{\underline{\tilde{L}}'}\underline{\tilde{L}}' = 0$. We first consider $\tilde{L}'$. Arguing as in the derivation of~\eqref{3oj2919}, we have that
\begin{equation}\label{3joij991}
\mathcal{L}_K\left(D_{\tilde{L}'}\tilde{L}'\right) = -3\left(D_{\tilde{L}'}\tilde{L}'\right).
\end{equation}
Thus $D_{\tilde{L}'}\tilde{L}'$ satisfies a first order ordinary differential equation along the integral curves of $K$ and vanishes along the integral curve of $L'$ which starts at $\{t = -1\} \cap \{\rho = \rho_0\}$. Thus $D_{\tilde{L}'}\tilde{L}'$ must vanish everywhere. The same argument works for $\underline{L}'$. Thus, without loss of generality we may assume that $L' = \tilde{L}'$ and $\underline{L}' = \tilde{\underline{L}}'$ and that $L'$ and $\underline{L}'$ are defined in a dilation invariant neighborhood of $\{\rho = \rho_0\}$. We then define the $u$ and $v$ coordinates by
\[L'\left(u\right) = \underline{L}'\left(v\right) = 0,\qquad u|_{\rho=\rho_0} = t,\qquad v|_{\rho=\rho_0} = (-t)\rho_0.\]
The functions $u$ and $v$ are eikonal functions, and, in a suitable dilation invariant neighborhood of $\{\rho = \rho_0\}$, we have that their level sets $\underline{\mathcal{H}}_u$ and $\mathcal{H}_v$ form regular hypersurfaces which intersect traversally in topological spheres $\mathbb{S}^2_{u,v}$. Finally we define the angular coordinates $\{\theta^A\}$ in the usual fashion: We start with local coordinates $\{\theta^A\}$ defined along a suitable coordinate chart on each $\mathbb{S}^2$ on $\{\rho = \rho_0\}$ so that $\mathcal{L}_{\partial_t}\theta^A = 0$. Then we define $\{\theta^A\}$ on a suitable coordinate chart at the sphere $\mathbb{S}^2_{u,v}$ by requiring that $L\left(\theta^A\right) = 0$. It is then straightforward to see that $\left(u,v,\theta^A\right)$ yields a double-null coordinate system.

Next we need to check that $K = u\partial_u + v\partial_v$. To see this, it suffices to observe that
\[K\left(u\right)|_{\rho = \rho_0} = u,\qquad K\left(v\right)|_{\rho = \rho_0} = v,\qquad K\left(\theta^A\right)|_{\rho = \rho_0} = 0,\]
\[L'\left(K\left(u\right)-u\right) = 0,\qquad \underline{L}'\left(K\left(v\right)-v\right) = 0,\qquad L'\left(K\left(\theta^A\right)\right) = 0. \]

Finally, the estimate~\eqref{32oijoi991} is a straightforward consequence of the initial closeness to Minkowski space and ODE estimates. We omit the details. 
\end{proof}

In this final lemma of the section, we state a converse of Lemma~\ref{oioioioi12}.
\begin{lemma}\label{0990239023}Suppose that $\left(\mathcal{M},g\right)$ is covered by a double-null coordinate system in the region~\eqref{3ij2oijoi2}, has a self-similar vector field $K = u\partial_u + v\partial_v$, and moreover satisfies the smallness condition~\eqref{32oijoi991} for some suitably large positive integer $\tilde{N}$. Then $\left(\mathcal{M},g\right)$ is in the $\left(\epsilon,N\right)$-small spacelike self-similar regime for a positive integer $N \gtrsim \tilde{N}$.
\end{lemma}
\begin{proof}This follows by repeating the proof of Proposition~\ref{homotheticgaugeexist} \emph{mutatis mutandis} and using straightforward ODE estimates. We omit the details.
\end{proof}

\section{Algebraic Consequences of Self-Similarity and Seed Data}
In this section we first discuss various algebraic consequences of self-similarity. (Many of these identities have already been derived in our work~\cite{nakedinterior}.) Then, motivated by some of these formulas, we will define a notion of seed data. Later we will see that a choice of seed data will parametrize our formal expansions.

\subsection{Self-Similar Identities}
The following lemma translates the self-similar assumption into a direct statement about the metric coefficients.
\begin{lemma}\label{32ijioijoij2oi}
The following hold for any self-similar double-null spacetime in a coordinate frame:
\begin{equation}\label{scaleinvrelations2}
\Omega\left(u,v,\theta^A\right) = \tilde{\Omega}\left(\frac{v}{u},\theta\right),\qquad b_A\left(u,v,\theta^A\right) = u\tilde{b}_A\left(\frac{v}{u},\theta^A\right),\qquad \slashed{g}_{AB}\left(u,v,\theta^A\right) = u^2\tilde{\slashed{g}}_{AB}\left(\frac{v}{u},\theta\right),
\end{equation}
for some functions $\tilde\Omega$, $\tilde{b}_A$, and $\tilde{\slashed{g}}_{AB}$.
\end{lemma}
\begin{proof}This is an immediate consequence of writing out $\mathcal{L}_Kg = 2g$ in the doubl-null coordinates and using that $K = u\partial_u+v\partial_v$. 
\end{proof}
\begin{remark}By differentiation of the relations in~\eqref{scaleinvrelations2} one may easily derive the self-similar form for the various Ricci coefficients. In particular, it is straightforward to see that for any Ricci coefficient $\psi$ in an \underline{orthonormal frame} must satisfy $\psi = u^{-1}H\left(\frac{v}{u},\theta^C\right)$ for a suitable function $H$. Using this fact and the expression for $\slashed{g}$ in~\eqref{scaleinvrelations2}, one can then easily derive the corresponding expression in a coordinate frame by modifying the power of $u$ depending on how many indices the Ricci coefficient $\psi$ has. For example, we must have that in the coordinate frame 
\[\eta_A = H_A\left(\frac{v}{u},\theta^B\right),\]
for a suitable function $H_A$. 
\end{remark}

In the following proposition, we collect various consequences of self-similarity which have been established in~\cite{scaleinvariant,nakedinterior}.
\begin{proposition}\label{somanyformulassolittletime}On any self-similar double-null spacetime, the following hold:
\begin{equation}\label{20asdww}
\Omega{\rm tr}\underline{\chi} + \Omega \frac{v}{u}{\rm tr}\chi = \frac{2}{u} + \slashed{\rm div}b,\qquad \Omega\hat{\underline{\chi}} + \Omega \frac{v}{u}\hat{\chi} = \frac{1}{2}\slashed{\nabla}\hat{\otimes}b,
\end{equation}
\begin{equation}\label{okodwok22}
\Omega\underline{\omega} + \frac{v}{u}\Omega\omega +\frac{1}{2}\mathcal{L}_b\log\Omega = 0,
\end{equation}
\begin{align}\label{eqnyaydivb}
&-\frac{v}{u}\mathcal{L}_{\partial_v}\slashed{\rm div}b+ \mathcal{L}_b\slashed{\rm div}b  + \frac{1}{u}\slashed{\rm div}b  +\frac{1}{2}\left(\slashed{\rm div}b\right)^2 + 8\left(\Omega\underline{\omega}\right)u^{-1} + 4\left(\Omega\underline{\omega}\right)\slashed{\rm div}b +\left|\Omega\hat{\underline{\chi}}\right|^2=
\\ \nonumber &\qquad  \frac{2v}{u^2}\Omega{\rm tr}\chi  +\left(\mathcal{L}_{\partial_u} + b\cdot\slashed{\nabla}\right)\left( \frac{v}{u}\Omega{\rm tr}\chi\right)+ 2\frac{v}{u}\slashed{\rm div}b\Omega{\rm tr}\chi -\frac{v^2}{2u^2}\left(\Omega{\rm tr}\chi\right)^2+4\left(\Omega\underline{\omega}\right) \frac{v}{u}\Omega{\rm tr}\chi -\Omega^2{\rm Ric}_{33}, 
\end{align}
\begin{align}\label{3pk2o294}
&\frac{v}{u}\mathcal{L}_{\partial_v}\eta_A - \mathcal{L}_b\eta_A -\eta_A\left(\Omega{\rm tr}\underline{\chi}\right)- 4\slashed{\nabla}_A\left(\Omega\underline{\omega}\right) =  \slashed{\nabla}^B\left(\Omega\hat{\underline{\chi}}\right)_{AB} - \frac{1}{2}\slashed{\nabla}_A\left(\Omega{\rm tr}\underline{\chi}\right) -\Omega{\rm Ric}_{3A},
\end{align}
\begin{align}\label{2o4oijoiouoiu2}
&-u^{-1}\left(\Omega^{-1}{\rm tr}\chi\right) -\frac{v}{u}\mathcal{L}_{\partial_v}\left(\Omega^{-1}{\rm tr}\chi\right) + \mathcal{L}_b\left(\Omega^{-1}{\rm tr}\chi\right)+ \left(\Omega^{-1}{\rm tr}\chi\right)\left(\Omega{\rm tr}\underline{\chi}\right) =
\\ \nonumber &\qquad \qquad  \qquad \qquad -2K+ 2\slashed{\rm div}\eta + 2\left|\eta\right|^2 + \left(R+{\rm Ric}_{34}\right),
\end{align}
\begin{align}\label{kwdkodwok23dg}
& -\frac{v}{u}\Omega\nabla_4\left(\Omega^{-1}\hat{\chi}\right)_{AB} +\mathscr{L}\left(\Omega^{-1}\hat{\chi}\right)_{AB} -\frac{v}{u}\Omega^{-1}{\rm tr}\chi\left(\Omega\hat{\chi}\right)_{AB} - 4\left(\Omega\underline{\omega}\right)\Omega^{-1}\hat{\chi}_{AB}= 
\\ \nonumber &\qquad \left(\left(\slashed{\nabla}\hat\otimes \eta\right)_{AB} + \left(\eta\hat\otimes \eta\right)_{AB}\right) - \frac{1}{4}\left(\Omega^{-1}{\rm tr}\chi\right)\left(\slashed{\nabla}\hat{\otimes}b\right) + \widehat{{\rm Ric}}_{AB},
 \end{align}
 \begin{equation}\label{3o1984982}
 \mathscr{L}f_{AB} \doteq \mathcal{L}_bf_{AB}- \left(\slashed{\nabla}\hat{\otimes}b\right)^C_{\ \ (A}f_{B)C} -\frac{1}{2}\slashed{\rm div}bf_{AB},\end{equation}
\end{proposition}
\begin{proof}These follow from Lemma B.1 of~\cite{scaleinvariant}, Lemmas 8.3 and  8.6 of~\cite{nakedinterior}, the proof of Lemma~8.11 from~\cite{nakedinterior}, and Lemma 8.2 from~\cite{nakedinterior}.
\end{proof}

This next equation will be used to compute the expansions for $\Omega\underline{\omega}$.
\begin{lemma}\label{2kn3kn2i39}On any self-similar double-null spacetime, the following hold:
\begin{align}\label{kfejifei}  \Omega\nabla_4\left(\Omega\underline{\omega}\right) &=  -\Omega^2\frac{v}{u}\mathcal{L}_{\partial_v}\left(\Omega^{-1}{\rm tr}\chi\right) + \Omega^2\mathcal{L}_b\left(\Omega^{-1}{\rm tr}\chi\right) -u^{-1}\Omega{\rm tr}\chi + \frac{1}{2}\left(\Omega{\rm tr}\chi\right)\left(\Omega{\rm tr}\underline{\chi}\right)+ \frac{1}{2}\Omega^2\left|\eta\right|^2 
\\ \nonumber &\qquad \qquad - \Omega^2\eta\cdot\underline{\eta}-4\left(\Omega\underline{\omega}\right)\left(\Omega{\rm tr}\chi\right) - 2\Omega^2\slashed{\rm div}\eta - 2\Omega^2\left|\eta\right|^2 + \left(\Omega\hat{\chi}\right)\cdot\left(\Omega\hat{\underline{\chi}}\right)+\frac{1}{4}\Omega^2{\rm Ric}_{34}.
\end{align}
\end{lemma}
\begin{proof}From~\eqref{4uomega} and~\eqref{3trchi}, we obtain, after multiplying through by $\Omega^2$: 

\begin{align}\label{2om3om2o}  \Omega\nabla_4\left(\Omega\underline{\omega}\right) &=  \Omega^3\nabla_3\left(\Omega^{-1}{\rm tr}\chi\right) + \Omega^2\frac{1}{2}\left(\Omega^{-1}{\rm tr}\chi\right)\left(\Omega{\rm tr}\underline{\chi}\right)+ \frac{1}{2}\Omega^2\left|\eta\right|^2 
\\ \nonumber &\qquad \qquad - \Omega^2\eta\cdot\underline{\eta}-4\Omega^2\left(\Omega\underline{\omega}\right)\left(\Omega^{-1}{\rm tr}\chi\right) - 2\Omega^2\slashed{\rm div}\eta - 2\Omega^2\left|\eta\right|^2 + \Omega^2\left(\Omega^{-1}\hat{\chi}\right)\cdot\left(\Omega\hat{\underline{\chi}}\right) +\frac{1}{4}\Omega^2{\rm Ric}_{34}.
\end{align}
Due to self-similarity, we have that $\Omega^{-1}{\rm tr}\chi = u^{-1}H\left(\frac{v}{u},\theta^A\right)$ for a suitable function $H$. Thus
\[\Omega\nabla_3\left(\Omega^{-1}{\rm tr}\chi\right) =-\frac{v}{u}\mathcal{L}_{\partial_v}\left(\Omega^{-1}{\rm tr}\chi\right) + \mathcal{L}_b\left(\Omega^{-1}{\rm tr}\chi\right) -u^{-1}\Omega^{-1}{\rm tr}\chi.\]
Plugging this into~\eqref{2om3om2o} yields~\eqref{kfejifei}.

\end{proof}

This next equation will be used when we compute the expansion for $\eta$.
\begin{lemma}\label{2om3mo020202}On any self-similar double-null spacetime, the following hold:
\begin{align}\label{2lmo204202940204i}
&\mathcal{L}_{\partial_v}\eta_A = 
\\ \nonumber &\qquad \Omega^2\left[\left(\Omega^{-1}\chi\right)_{AB}\underline{\eta}^B+ \slashed{\nabla}^B\left(\Omega^{-1}\hat{\chi}\right)_{AB}-\frac{1}{2}\slashed{\nabla}_A\left(\Omega^{-1}{\rm tr}\chi\right)-\frac{1}{2}\left(\Omega^{-1}{\rm tr}\chi\right) \eta_A + \eta^B\left(\Omega^{-1}\hat{\chi}\right)_{AB} - \Omega{\rm Ric}_{A4}\right].
\end{align}
\end{lemma}
\begin{proof}From~\eqref{4eta} and~\eqref{tcod1}, we have
\begin{align*}
\nabla_4\eta_A = -\chi_{AB}\cdot\left(\eta^B-\underline{\eta}^B\right)+ \slashed{\nabla}^B\hat{\chi}_{AB}-\frac{1}{2}\slashed{\nabla}_A{\rm tr}\chi-\frac{1}{2}{\rm tr}\chi \zeta_A + \zeta^B\hat{\chi}_{AB} - {\rm Ric}_{A4}.
\end{align*}
To obtain~\eqref{2lmo204202940204i} we simply multiply though by the lapse and write $\Omega\nabla_4\eta = \mathcal{L}_{\partial_v}\eta - \left(\Omega\chi\right)\cdot\eta$.

\end{proof}

\subsection{Seed Data}

We will close the section with a definition of ``seed data'' for our expansions. 
\begin{definition}\label{thisistheseed}Let $N$ be a sufficiently large positive integer and $\epsilon > 0$ be sufficiently small. We say that a $7$-tuple 
\[\left(\slashed{g}^{({\rm seed})},b^{({\rm seed})},\left(\Omega\underline{\omega}\right)^{({\rm seed})},\left(\Omega^{-1} \hat{\chi}\right)^{(\rm seed,+)},\left(\Omega^{-1} \hat{\chi}\right)^{(\rm seed,-)},\log\Omega^{(\rm seed,+)},\log\Omega^{(\rm seed,-)}\right)\]
of a Riemannian metric $\slashed{g}^{({\rm seed})}$ on $\mathbb{S}^2$, a vectorfield $b^{({\rm seed})}$ on $\mathbb{S}^2$, a function $\left(\Omega\underline{\omega}\right)^{({\rm seed})}$ on $\mathbb{S}^2$, symmetric $\mathbb{S}^2_v$ valued $(0,2)$-tensors $\left(\Omega\hat{\chi}\right)^{(\rm seed,\pm)}$ defined for $v \in (0,1]$ and $v \in [-1,0)$ respectively, and functions  $\log\Omega^{(\rm seed,\pm)}\left(v,\theta^A\right) : I \times \mathbb{S}^2 \to \mathbb{R}$ for $I = (0,1]$ and $[-1,0)$ respectively, form ``$\left(\epsilon,N\right)$-seed data'' if the following equations hold (where we have dropped the ``seed'' superscript from the various quantities and we may extend $\slashed{g}$, $b$, and $\Omega\underline{\omega}$ to $v \in (-1,1)\setminus \{0\}$ by declaring them to be independent of $v$):
\begin{equation}\label{m2om3o2}
 \mathcal{L}_b\slashed{\rm div}b  - \slashed{\rm div}b  +\frac{1}{2}\left(\slashed{\rm div}b\right)^2 -8\left(\Omega\underline{\omega}\right) + 4\left(\Omega\underline{\omega}\right)\slashed{\rm div}b +\frac{1}{4}\left|\slashed{\nabla}\hat{\otimes}b\right|^2= 0,
 \end{equation}
\begin{align}\label{3oi2oij34iuorhiuehiu}
&-\frac{v}{u}\mathcal{L}_{\partial_v}\left(\Omega^{-1} \hat{\chi}\right)^{\pm}_{AB} +\mathcal{L}_b\left(\Omega^{-1} \hat{\chi}\right)^{\pm}_{AB}- \left(\slashed{\nabla}\hat{\otimes}b\right)^C_{\ \ (A}\left(\Omega^{-1} \hat{\chi}\right)^{\pm}_{B)C} 
\\ \nonumber &\qquad -\frac{1}{2}\slashed{\rm div}b\left(\Omega^{-1} \hat{\chi}\right)^{\pm}_{AB} -4\left(\Omega\underline{\omega}\right)\left(\Omega^{-1}\hat{\chi}\right)_{AB} 
 = \left(\left(\slashed{\nabla}\hat\otimes \eta^{(0)}\right)_{AB} + \left(\eta^{(0)}\hat\otimes \eta^{(0)}\right)_{AB}\right) - \frac{1}{4}\left(\Omega^{-1}{\rm tr}\chi\right)\left(\slashed{\nabla}\hat{\otimes}b\right)_{AB}, 
\end{align}
\[\slashed{g}^{AB}\left(\Omega^{-1}\hat{\chi}\right)^{\pm}_{AB}|_{v=\pm 1} = 0,\]
\begin{equation}\label{3jiio2}
\left(v\mathcal{L}_{\partial_v}+\mathcal{L}_b\right)\log\Omega^{\pm} = \Omega\underline{\omega},
\end{equation}
and if we have that
 \[ \left\vert\left\vert \left(\slashed{g}^{({\rm seed})}-\mathring{\slashed{g}},b^{({\rm seed})},\left(\Omega\omega\right)^{({\rm seed})},\left(\Omega^{-1}\hat{\chi}\right)^{(\rm seed,\pm)}|_{v=\pm 1},\log\Omega^{(\rm seed,\pm)}|_{v=\pm 1}\right)\right\vert\right\vert_{C^N\left(\mathbb{S}^2\right)} \lesssim \epsilon.\]
Here $\eta^{(0)}$ is defined along $\mathbb{S}^2$ by applying Proposition~\ref{somestuimdie} to solve the equation 
\begin{equation}\label{3ioj2oij4oi2}
 - \mathcal{L}_b\eta_A^{(0)} -\eta_A^{(0)}\left(-2+\slashed{\rm div}b\right)- 4\slashed{\nabla}_A\left(\Omega\underline{\omega}\right) =  \frac{1}{2}\slashed{\nabla}^B\left(\slashed{\nabla}\hat{\otimes}b\right)_{AB} - \frac{1}{2}\slashed{\nabla}_A\slashed{\rm div}b,
 \end{equation}
and is then extended to $v \in (-1,1)$ by declaring $\eta$ to be independent of $v$.

 If we drop $\left(\Omega^{-1}\hat{\chi}\right)^{(\rm seed,-)}$ and $\log\Omega^{(\rm seed,-)}$ then we call the corresponding $5$-tuple ``$1$-sided $\left(\epsilon,N\right)$-seed data.''
\end{definition}

\begin{remark}The equation~\eqref{m2om3o2} may be formally derived from~\eqref{eqnyaydivb} by taking $(u,v) = (-1,0)$, taking $\Omega^2{\rm Ric}_{33}|_{v=0} = 0$, and assuming that all terms multiplied by $\frac{v}{u}$ vanish when $v = 0$. 
\end{remark}

The following proposition provides one way to parametrize choices of solutions to~\eqref{m2om3o2}.
\begin{proposition}\label{2km3o2o23}Let $M$ be a sufficiently large positive integer, $\slashed{g}$ be any Riemannian metric on $\mathbb{S}^2$ satisfying $\left\vert\left\vert \slashed{g}-\mathring{\slashed{g}}\right\vert\right\vert_{\mathring{H}^M} \lesssim \epsilon$ and let $\left(\mathcal{O},W\right)$ be functions on $\mathbb{S}^2$ which satisfy $\int_{\mathbb{S}^2}\mathcal{O}\slashed{\rm dVol} =\int_{\mathbb{S}^2}W\slashed{\rm dVol}= 0$ and $\left\vert\left\vert \left(\mathcal{O},\mathring{\nabla}W\right)\right\vert\right\vert_{\mathring{H}^{M-1}} \lesssim \epsilon$. 

Then there exists $\left(\slashed{g}^{({\rm seed})},b^{({\rm seed})},(\Omega\underline{\omega})^{({\rm seed})}\right)$ solving~\eqref{m2om3o2} with $\slashed{g}^{({\rm seed})} = \slashed{g}$, $\slashed{\rm curl}b^{({\rm seed})} = \mathcal{O}$, and $\left(\Omega\underline{\omega}\right)^{(\rm seed)} = W + a$, where $a$ is a constant satisfying $|a| \lesssim \epsilon$. 
\end{proposition}
\begin{proof}The proof is essentially a minor adapter of the proof of Proposition A.1 of~\cite{nakedone}, so we will just provide a sketch: Let $\mathcal{P} \doteq \slashed{\Delta}^{-1}\mathcal{O}$. Then, if we set 
\[b_A \doteq \slashed{\nabla}_Af - \slashed{\epsilon}_A^{\ \ B}\slashed{\nabla}_B\mathcal{P},\qquad \left(\Omega\underline{\omega}\right) \doteq W + a,\]
for a function $f$ and a constant $a$ to be determined, we find that the equation~\eqref{m2om3o2} becomes
\begin{equation*}
 \mathcal{L}_b\slashed{\Delta}f  - \slashed{\Delta}f  +\frac{1}{2}\left(\slashed{\Delta}f\right)^2 -8\left(W+a\right) + 4\left(W+a\right)\slashed{\Delta}f +\frac{1}{4}\left|\slashed{\nabla}\hat{\otimes}\left(\slashed{\nabla}f - {}^*\slashed{\nabla}\mathcal{P}\right)\right|^2= 0,
 \end{equation*}
 We then solve this by  running an iteration scheme  indexed by non-negative integers $i$ where each function $b(i)$, $f(i)$, and constant $a(i)$ are required to satisfy, when $i > 0$,
 \begin{align}\label{ij32iojio}
& \mathcal{L}_{b(i-1)}\slashed{\Delta}f(i)  - \slashed{\Delta}f(i)  +\frac{1}{2}\left(\slashed{\Delta}f(i-1)\right)^2 -8\left(W+a(i)\right)
 \\ \nonumber &\qquad  + 4\left(W+a(i)\right)\slashed{\Delta}f(i-1) +\frac{1}{4}\left|\slashed{\nabla}\hat{\otimes}\left(\slashed{\nabla}f(i-1) - {}^*\slashed{\nabla}\mathcal{P}\right)\right|^2= 0,
 \end{align}
 where $b_A(i) = \slashed{\nabla}_Af(i) - \slashed{\epsilon}_A^{\ \ B}\slashed{\nabla}_B\mathcal{P}$. For $i = 0$, we simply set $f(0) = 0$ and $a(0) = 0$. Via Proposition~\ref{somestuimdie}, we see that~\eqref{ij32iojio} can be expected to determine $\slashed{\Delta}f(i)$ from the previous terms in the iteration and a choice of the constant $a(i)$. We fix the choice of $a(i)$ by the requirement that
 \[\int_{\mathbb{S}^2}\slashed{\Delta}f(i) \slashed{\rm dVol} = 0.\]
We omit the rest of the argument as the iterates may now be shown to converge exactly as in the proof of Proposition A.1 of~\cite{nakedone}.
 \end{proof}

\begin{remark}\label{rmkonfg}As discussed in the introduction, Fefferman--Graham type geometries correspond to self-similar metrics where the homothetic vector field $K$ is null along $\mathcal{H}_0$. For seed data associated to these solutions, we must thus take $b^{(\rm seed)} = 0$ (as $b^{(\rm seed)}$ will end up corresponding to the limit of $b$ as $v\to 0$). In turn~\eqref{m2om3o2} then implies that $\left(\Omega\underline{\omega}\right)^{(\rm seed)} = 0$. Furthermore, since $b^{(\rm seed)} = 0$, choices of $\log\Omega^{(\rm seed,\pm)}$ and $\left(\Omega^{-1}\hat{\chi}\right)^{(\rm seed,\pm)}$ are then determined just by a choice of a function(s), or trace-free symmetric $(0,2)$-tensor(s) respectively, along the sphere $\mathbb{S}^2$ which is then extended to have a trivial dependence in $v$. Thus, if one is only interested in Fefferman--Graham type geometries, then it is natural to define the seed data to just consist of a choice of a Riemannnian metric $\slashed{g}$ on $\mathbb{S}^2$, positive function(s) $\Omega$ on $\mathbb{S}^2$, and symmetric trace-free $(0,2)$-tensor(s) $\Omega^{-1}\hat{\chi}$ along $\mathbb{S}^2$. It then in fact turns out that, without loss of generality, one can take $\slashed{g}$ to be the round metric and $\Omega$ to be identically $1$, and thus that, up to self-similar change of double-null coordinates, the formal expansions corresponding to Fefferman--Graham geometries are parametrized by choices of symmetric trace-free $(0,2)$-tensors along $\mathbb{S}^2$.  See Appendix~\ref{feffgrahamreggauge} for the detailed statement. 
\end{remark}

\section{Statement of Main Results}\label{sectionwithmaintheorems}
In this section we state explicitly our main results. Let us introduce the convention that, in this section, when we compute the norm of an $\mathbb{S}^2_{u,v}$ tensor along a sphere $\mathbb{S}^2_{u,v}$, we use the metric $\left(v-u\right)^2\mathring{\slashed{g}}_{AB}$, where $\mathring{\slashed{g}}_{AB}$ is a choice of $(u,v)$ independent  round metric at each $\mathbb{S}^2_{u,v}$.

Our first theorem concerns the existence and uniqueness of formal expansions for twisted self-similar metrics.
\begin{theorem}\label{theoremexpand}Let $\left(\slashed{g}^{({\rm seed})},b^{({\rm seed})},\left(\Omega\underline{\omega}\right)^{({\rm seed})},\left(\Omega^{-1} \hat{\chi}\right)^{(\rm seed,\pm)},\log\Omega^{(\rm seed,\pm)}\right)$ be $\left(\epsilon,N\right)$-seed data for $N \in \mathbb{Z}_{> 0}$ sufficiently large and $\epsilon > 0$ sufficiently small (where the smallness of $\epsilon$ may depend on $N$). Then there exists an $\tilde{N} \gtrsim N$ (where the implied constant is independent of $N$ and $\epsilon$), a constant $\hat{C} > 0$ (independent of $\epsilon$), and a self-similar double null metric $g$ such that the following hold with all implicit constants possibly depending on $N$:
\begin{enumerate}
	\item The metric $g$ is $\left(\epsilon,\tilde{N}\right)$-regular up to $\{v = 0\}$. In particular, $g$ is a $C^{\tilde{N}}$ metric in the double-null coordinates (in both of the regions $\{v > 0\}$ and $\{v < 0\}$) and may moreover, by an application of Proposition~\ref{nonononononoo}, be put into the homothetic gauge where it will be a $\left(\epsilon,\tilde{N}-2\right)$-regular twisted self-similar spacetime. We emphasize that we do \underline{not} need any compatibility of $\Omega^{-1}\hat{\chi}^{(\rm seed,+)}$ and  $\Omega^{-1}\hat{\chi}^{(\rm seed,-)}$ or $\log\Omega^{(\rm seed,+)}$ and $\log\Omega^{(\rm seed,-)}$ in order to be an $\left(\epsilon,\tilde{N}-2\right)$-regular twisted self-similar spacetime.
	\item The seed data determines the behavior of $g$ in the self-similar double-null gauge near $\{v = 0\}$ in the following sense:
	\begin{equation}\label{3poj2poj4po2223}
	\sum_{j=0}^{\tilde{N}}\left\vert\left\vert \left(v\mathcal{L}_{\partial_v}\right)^j\left(u\Omega\underline{\omega}-u\Omega\underline{\omega}^{(\rm seed)},u\eta - u\eta^{(0)},\slashed{g} - \slashed{g}^{(\rm seed)},b - b^{(\rm seed)}\right)\right\vert\right\vert_{C^{\tilde{N}-j}\left(\mathbb{S}^2_{u,v}\right)} \lesssim \epsilon \left|\frac{v}{u}\right|^{1-\hat{C}\epsilon},
	\end{equation}
	\begin{equation}\label{3poj2poj4po2}
	\sum_{j=0}^{\tilde{N}}\left\vert\left\vert \left(v\mathcal{L}_{\partial_v}\right)^j\left(u\Omega^{-1}\hat{\chi}-u\left(\Omega^{-1}\hat{\chi}\right)^{(\rm seed,\pm)},\log\Omega - \log\Omega^{(\rm seed,\pm)}\right)\right\vert\right\vert_{C^{\tilde{N}-j}\left(\mathbb{S}^2_{u,v}\right)} \lesssim \epsilon \left|\frac{v}{u}\right|^{1-\hat{C}\epsilon},
	\end{equation}
	where~\eqref{3poj2poj4po2223} holds for all $(u,v)$,~\eqref{3poj2poj4po2} with the ``$+$'' holds for $v > 0$,~\eqref{3poj2poj4po2} with the ``$-$'' folds for $v < 0$, and $\eta^{(0)}$ is the quantity defined in Definition~\ref{thisistheseed}. 
	\item The Ricci tensor ${\rm Ric}\left(g\right)$ satisfies the following bounds in the double null coordinates:
	\begin{equation}\label{3oij2oi829}
	\sum_{j=0}^{\tilde{N}}\left\vert\left\vert \left(v\mathcal{L}_{\partial_v}\right)^j{\rm Ric}\left(X,Y\right)\right\vert\right\vert_{C^{\tilde{N}-j}\left(\mathbb{S}^2_{u,v}\right)} \lesssim \epsilon u^{-2}\left|\frac{v}{u}\right|^{\tilde{N}-1-\hat{C}\epsilon},
	\end{equation}
	where $X,Y \in \{\mathcal{L}_{\partial_v},\mathcal{L}_{\partial_u},u^{-1}\partial_{\theta^A}\}$.	
	\item The Ricci tensor ${\rm Ric}\left(g\right)$ satisfies the following bounds in the homothetic gauge when $\rho \neq 0$:
	\begin{equation}\label{9198j32}
	\sum_{j=0}^{\tilde{N}}\left\vert\left\vert \left(\rho\mathcal{L}_{\partial_{\rho}}\right)^j{\rm Ric}\left(X,Y\right)\right\vert\right\vert_{C^{\tilde{N}-j}\left(\mathbb{S}^2_{t,\rho}\right)} \lesssim \epsilon t^{-2}\left|\rho\right|^{\tilde{N}-1-\hat{C}\epsilon},
	\end{equation}
	where $X,Y \in \{\partial_t,t^{-1}\partial_{\rho},t^{-1}\partial_{\theta^A}\}$. The metric $g$ will, in general, only be H\"{o}lder continuous $\rho = 0$. However, the Ricci tensor (and up to $\tilde{N}$ applications of $\rho\mathcal{L}_{\rho}$ or angular derivatives) will be defined as an $L^1$ function in a weak sense and~\eqref{9198j32} will then continue to hold. 
	\item The metric components $g_{\alpha\beta}$ in the double-null coordinate system may be written as a sum
	\[g_{\alpha\beta} = \sum_{j=0}^{\tilde{N}}g_{\alpha\beta}^{(j)},\]
	where each $g_{\alpha\beta}^{(j)}$ satisfies
	\begin{equation}\label{32ojoj42}
	\sum_{i=0}^{\tilde{N}}\left\vert\left\vert \left(v\mathcal{L}_{\partial_v}\right)^ig^{(j)}_{\alpha\beta}\right\vert\right\vert_{C^{\tilde{N}-i}\left(\mathbb{S}^2_{-1,v}\right)} \lesssim \epsilon \left|v\right|^{j-\hat{C}\epsilon},
	\end{equation}
	and is defined inductively by either solving explicit transport equations of the type discussed in Lemma~\ref{laptrans}, or solving equations of the type discussed in Proposition~\ref{somestuimdie}, or by integration in $v$. In all three cases the right hand sides depend on previously computed terms in the expansions, and the base case of the inductive procedure is determined by the choice of seed data. (See the proof of Theorem~\ref{2omo2o492} for the explicit procedure.) We note that even though the bound~\eqref{32ojoj42} is only stated along $\{u = -1\}$, for any given metric component one may use the underlying self-similarity to obtain an estimate in the entire spacetime. In view of (the proof of) Proposition~\ref{nonononononoo} one may obtain corresponding expansions for the metric components in the homothetic gauge in an algorithmic fashion.
	\item\label{9090909090232} Let $g$ be a metric arising from an $\left(\epsilon,N\right)$-regular twisted self-similar spacetime and which weakly solves ${\rm Ric}(g) = 0$. We may then put $g$ into a double-null gauge via Proposition~\ref{makeitadouble}. There then exists a unique choice of seed data and a corresponding metric $\tilde{g}$ produced by the formal expansion procedure described above so that in the double-null coordinate system
	\begin{equation}\label{3981jno3}
	\sum_{j=0}^{\tilde{N}}\left\vert\left\vert \left(v\mathcal{L}_{\partial_v}\right)^j\left(g_{\alpha\beta}-\tilde{g}_{\alpha\beta}\right)\right\vert\right\vert_{C^{\tilde{N}-j}\left(\mathbb{S}^2_{-1,v}\right)} \lesssim \epsilon \left|v\right|^{\tilde{N}-\hat{C}\epsilon},
	\end{equation}
	for some $\tilde{N} \gtrsim N$.	As with the bound~\eqref{32ojoj42}, even though~\eqref{3981jno3} is only stated along $\{u = -1\}$, for any given metric component one may use the underlying self-similarity to obtain an estimate in the entire spacetime. (Analogous statements hold for $1$-sided $\left(\epsilon,N\right)$-regular spacetimes which solve the Einstein vacuum equations.)
		\end{enumerate}

Finally, we observe that if we are instead given $1$-sided $\left(\epsilon,N\right)$-regular seed data, then the analogues of the above results hold for $1$-sided self-similar solutions. (By a change of variables $v\mapsto -v$ we may consider $\left(\Omega^{-1}\hat{\chi}\right)^{\rm seed,+}$ and $\log\Omega^{{\rm seed,+}}$ as being defined for either $v\in (0,1]$ or $v \in [-1,0)$ depending on which ``side'' of $\{v = 0\}$ we want to define our expansions.)
\end{theorem}
The proof of Theorem~\ref{theoremexpand} will be given in Section~\ref{3ioj2oij42}.

\begin{remark}In the case when $b^{(\rm seed)} = 0$ then the resulting formal expansions correspond to a Fefferman--Graham type geometry where the homothetic vector field $K$ is null along $\mathcal{H}_0$. In this case, one may show that the metric in the corresponding self-similar double-null foliation exhibits greatly improved regularity as $v\to 0$ compared to the case of generic seed data. See Appendix~\ref{feffgrahamreggauge} for the detailed statements.
\end{remark}
\begin{remark}\label{u0canwedoityest}In the case of a metric $g$ which is $\left(\epsilon,N\right)$-regular up to $\{u=0\}$ (see the second part of Definition~\ref{thisdefissomething}), we may swap the roles of $u$ and $v$ in the sense of Remark~\ref{shiftshift}, and then apply part~\ref{9090909090232} of Theorem~\ref{theoremexpand} to the resulting metric. Finally, we can then un-swap $u$ and $v$ to deduce a result for the original metric. Thus, Theorem~\ref{theoremexpand} may be used to understand the asymptotic behavior as $u\to 0$ for metrics $g$ which are $\left(\epsilon,N\right)$-regular up to $\{u=0\}$.
\end{remark}

The next theorem concerns the existence of $1$-sided $\left(\epsilon,N\right)$-regular self-similar spacetimes corresponding to our expansions. In order to keep the exposition in the paper as streamlined as possible, we will not provide the proof of this result. However, given the expansions established by Theorem~\ref{theoremexpand}, one may prove Theorem~\ref{notprovedhere} below by an amalgamation of the techniques from~\cite{scaleinvariant,nakedone}.
\begin{theorem}\label{notprovedhere}Let $\left(\slashed{g}^{({\rm seed})},b^{({\rm seed})},\left(\Omega\underline{\omega}\right)^{({\rm seed})},\left(\Omega^{-1} \hat{\chi}\right)^{(\rm seed,+)},\log\Omega^{(\rm seed,+)}\right)$ be $1$-sided $\left(\epsilon,N\right)$-seed data for $N \in \mathbb{Z}_{> 0}$ sufficiently large and $\epsilon > 0$ sufficiently small. Let $\tilde{g}$ be the $1$-sided self-similar metric produced by the expansions of Theorem~\ref{theoremexpand} which is defined for $\{v > 0\}$ and is associated to the $1$-sided $\left(\epsilon,N\right)$-seed data. Then, for some $\tilde{N} \gtrsim N$ (for a constant which is independent of $N$ and $\epsilon$) there exists a $1$-sided $\left(\epsilon,\tilde{N}\right)$-regular self-similar spacetime which futhermore solves the Einstein vacuum equations weakly so that when the metric $g$ is put into the self-similar gauge, we have that~\eqref{3981jno3} holds for $\{v > 0\}$.  
\end{theorem}

Our final theorem concerns spacetimes in the $\left(\epsilon,N\right)$-small spacelike self-similar regime.
\begin{theorem}\label{spacespaceself} Let $\left(\mathcal{M},g\right)$ be a spacetime in the $\left(\epsilon,N\right)$-small spacelike self-similar regime. Then we may apply Lemma~\ref{oioioioi12} to equip $\left(\mathcal{M},g\right)$ with a double-null coordinate system valid in the region~\eqref{3ij2oijoi2}. Then, the maximal development of $\left(\mathcal{M},g\right)$ to the past of the hypersurface $\{\frac{v}{-u} = \tilde{b}\}$ includes the region
\[\left\{\left(u,v,\theta^A\right) \in (-\infty,0) \times \mathbb{R}\times\mathbb{S}^2 : 0 < \frac{v}{-u} \leq \tilde{b}\right\},\]
and is moreoever $\left(\epsilon,\tilde{N}\right)$-regular up to $\{v = 0\}$ for some $\tilde{N} \gtrsim N$ (where the implied constant is independent of $N$ and $\epsilon$). We may thus use Theorem~\ref{theoremexpand} to describe the asymptotic behavior of the metric as $v \to 0$.

Similarly, we may apply Lemma~\ref{oioioioi12} to equip $\left(\mathcal{M},g\right)$ with the double-null coordinate system~\eqref{2l3kj2lijo42} valid in the region~\eqref{3ij2oijoi2}. We may also consider the maximal development of $\left(\mathcal{M},g\right)$ to the future of the hypersurface $\{\frac{v}{-u} = \tilde{b}\}$. This development includes the region 
\[\left\{\left(u,v,\theta^A\right) \in (-\infty,0) \times \mathbb{R}\times\mathbb{S}^2 : \tilde{b} \leq \frac{v}{-u} < \infty\right\},\]
and is $\left(\epsilon,\tilde{N}\right)$-regular up to $\{u = 0\}$ for some $\tilde{N} \gtrsim N$ (where the implied constant is independent of $N$ and $\epsilon$). We may thus use Theorem~\ref{theoremexpand} (see Remark~\ref{u0canwedoityest}) to describe the asymptotic behavior of the metric as $u \to 0$.
\end{theorem}
\begin{proof}This is an immediate consequence of applying the a priori estimates from Sections 7, 8, and 9 of~\cite{nakedone} \emph{mutatis mutandis} (possibly after swapping the roles of $u$ and $v$ as in Remark~\ref{u0canwedoityest}), the equations~\eqref{2lmo204202940204i} and~\eqref{kfejifei} (for the estimates of $\mathcal{L}_{\partial_v}$ derivatives of $\eta$ and $\underline{\omega}$ as $v \to 0$ and, after swapping $u$ and $v$ as in Remark~\ref{u0canwedoityest}, for the estimates for $\mathcal{L}_{\partial_v}$ derivatives of $\underline{\eta}$ and $\omega$ as $u \to 0$),  the local existence theory for self-similar metrics from~\cite{scaleinvariant} \emph{mutatis mutandis}, and the fact that our assumptions on $\left(\mathcal{M},g\right)$ will exactly imply that all of the relevant norms from Section 7 of~\cite{nakedone} are finite and sufficiently small along $\frac{v}{-u} = \tilde{b}$ when applied to $\left(\mathcal{M},g\right)$. 
\end{proof}

We now state two corollaries of Theorem~\ref{spacespaceself}. These corollaries will refer to and assume familiarity with our previous works~\cite{scaleinvariant} and~\cite{nakedone}. The first corollary concerns the global behavior of Fefferman--Graham spacetimes (in the spacetime region $\{v > 0\}$).

\begin{corollary}\label{afirstcororor}Let $N \in \mathbb{Z}_{>0}$ be sufficiently large, and let $\left(\slashed{g}_0\right)_{AB}$ be a Riemannian metric on $\mathbb{S}^2$ and $h_{AB}$ be a $\slashed{g}_0$-trace free symmetric $(0,2)$-tensor on $\mathbb{S}^2$ so that $\left\vert\left\vert \slashed{g}_0 - \mathring{\slashed{g}}\right\vert\right\vert_{\mathring{H}^N} + \left\vert\left\vert h\right\vert\right\vert_{\mathring{H}^N}$ is sufficiently small, where, as usual, $\mathring{\slashed{g}}$ denotes a choice of a round metric on $\mathbb{S}^2$. Then, applying Theorem 1.3 of~\cite{scaleinvariant}, let $\left(\mathcal{M},g\right)$ denote the corresponding Fefferman--Graham spacetime given in double-null coordinates 
\[\left\{\left(u,v,\theta^A\right) \in (-\infty,0)\times [0,\infty)\times \mathbb{S}^2 : 0 \leq \frac{v}{-u} < \delta\right\},\]
for some $\delta \gtrsim 1$, where
\[\lim_{(u,v) \to (-1,0)}\slashed{g} = \slashed{g}_0,\qquad \lim_{(u,v)\to (-1,0)}\hat{\chi}_{AB} = h_{AB}.\]

Then, in a dilation invariant neighborhood of $\frac{v}{-u} = \delta/2$, $\left(\mathcal{M},g\right)$ is in the $\left(\epsilon,\tilde{N}\right)$-small spacelike self-similar regime for suitable $\tilde{N}$, and Theorem~\ref{spacespaceself} may be applied to describe the asymptotic behavior of the maximal development of $\left(\mathcal{M},g\right)$. 
\end{corollary}
\begin{proof}In view of Lemma~\ref{0990239023}, it suffices to check that the smallness condition of~\eqref{32oijoi991} holds along the hypersurface $\frac{v}{-u} = \frac{\delta}{2}$. This smallness is however an immediate consequence of the proof of Theorem 1.3 in~\cite{scaleinvariant}.
\end{proof}

Our second corollary concerns the asymptotic behavior of the naked singularity exterior constructed in~\cite{nakedone} and also the notion of self-similar extraction from~\cite{scaleinvariant}.
\begin{corollary}\label{nakedcorcor}Let $\left(\mathcal{M},g\right)$ denote the naked singularity exterior solution constructed in Theorem 1 of~\cite{nakedone}. For every $\lambda > 0$, we let $\Phi_{\lambda}$ be the map $\left(u,v,\theta^A\right) \mapsto \left(\lambda u,\lambda v,\theta^A\right)$. Then,
by applying the proof of Theorem 1.2 from~\cite{scaleinvariant} \emph{mutatis mutandis} to the rescaled metrics $g_{\lambda} \doteq \lambda^{-2}\Phi_{\lambda}^*g$, we will have that $g_{\lambda}$ converges as $\lambda \to \infty$ to a self-similar double-null metric $\left(\mathcal{M}_{\infty},g_{\infty}\right)$. Furthermore, $\left(\mathcal{M}_{\infty},g_{\infty}\right)$ will be in the $\left(\epsilon,N\right)$-small spacelike self-similar regime for suitable $\epsilon$ and $N$, and thus Theorem~\ref{spacespaceself} may be applied to describe the asymptotic behavior of the maximal development of $\left(\mathcal{M}_{\infty},g_{\infty}\right)$. 
\end{corollary}
\begin{proof}The needed estimates for $\left(\mathcal{M},g\right)$ follow from Theorem 6.1 from~\cite{nakedone}.
\end{proof}

\section{The Expansion: Proof of Theorem~\ref{theoremexpand} }\label{3ioj2oij42}
In this section we now explain how to generate our formal expansions and prove Theorem~\ref{theoremexpand}. Our starting point is a choice of, possibly $1$-sided, $\left(\epsilon,N\right)$-regular seed data in the sense of Definition~\ref{thisistheseed}. 

\subsection{Conventions and Notation}
We will work with double-null coordinate systems defined in regions $\mathcal{W} \times \mathbb{S}^2$, $\mathcal{W}^+ \times \mathbb{S}^2$, and $\mathcal{W}^- \times \mathbb{S}^2$, where
\[\mathcal{W}^+ \doteq \left\{ (u,v) \in (-\infty,0) \times (0,\infty) : 0 < \frac{v}{-u} < 1\right\}, \qquad \mathcal{W}^- \doteq \left\{ (u,v) \in (-\infty,0) \times (-\infty,0) : -1 < \frac{v}{-u} < 0\right\},\]
\[\mathcal{W} = \mathcal{W}^- \cup \mathcal{W}^+.\]
We will also denote by $\overline{\mathcal{W}}$, $\overline{\mathcal{W}^-}$, and $\overline{\mathcal{W}^+}$ the union of the various regions with the ray $\{u \in (-\infty,0)\text{ and }v = 0\}$.

The quantities $\left(\slashed{g}^{(\rm seed)},b^{(\rm seed)},\left(\Omega\underline{\omega}\right)^{(\rm seed)}\right)$ are initially defined along a sphere $\mathbb{S}^2$. We will now define these as suitable $\mathbb{S}^2_{u,v}$ tensors in $\overline{\mathcal{W}}\times \mathbb{S}^2$ as follows. We first identify the original sphere where they are defined with $\mathbb{S}^2_{-1,0}$ in  $\overline{\mathcal{W}}\times \mathbb{S}^2$. Then we extend $\left(\slashed{g}^{(\rm seed)},b^{(\rm seed)},\left(\Omega\underline{\omega}\right)^{(\rm seed)}\right)$ to the region $\{u = -1\} \cap \left(\overline{\mathcal{W}} \times \mathbb{S}^2\right)$ by declaring them to be independent of $v$. Finally, they may be extended to the entire region $\overline{\mathcal{W}} \times \mathbb{S}^2$ in the unique fashion dictated by the self-similar formulas~\eqref{scaleinvrelations2}. We next extend $\left(\left(\Omega^{-1} \hat{\chi}\right)^{(\rm seed,\pm)},\log\Omega^{(\rm seed,\pm)}\right)$  to $\mathcal{W}^{\pm}\times \mathbb{S}^2$ by identifying their domain with $\{u = -1\}$ and then using the self-similar formulas~\eqref{scaleinvrelations2}. More generally, throughout this section we will often define metric coefficients or Ricci coefficients only along the hypersurface $\{u=-1\}$ and then consider them automatically extended in the unique way determined by self-similarity. \underline{We will often not say so explicitly.} 

We next note that for each new term that we define in our expansions, we will generally be able to estimate strictly less angular derivatives of the term than we did for the previous terms. Furthermore, the further out in the expansion one desires to compute, the more angular derivatives one must control of the seed data and the smaller one must potentially take $\epsilon$. We will not attempt to track precisely the derivative loss; it will instead be clear that at each step of the expansion, we can, in principle, control as many angular derivatives as we would like if we suitable increase the regularity of the seed data and take $\epsilon$ smaller. Of course, any finite number of terms in the expansion will only require finitely many angular derivatives of the seed data to be regular. Given these considerations, throughout this section we will write bounds for $\left\vert\left\vert \cdot \right\vert\right\vert_{C^j}$ norms with the understanding that $j$ may be changing line to line, but that, conditional on increasing the regularity of the initial data and possibly taking $\epsilon$ smaller, we could take $j$ as high as we would like for any given term in the expansion. We will also have constants $D(j)$ that appear in many of our estimates. Our convention will be that each $D(j)$ may be distinct from the other $D(j)$'s that appear; instead it merely denotes some non-negative $j$ dependent constant relevant only for the specific estimate that is being written. For any fixed value of $N$, there will only be finitely many appearances of these $D(j)$ constants.

The general strategy will be as follows: For each metric component and a subset of the Ricci coefficients, we will multiply by a suitable power of the lapse to define a quantity $\phi$. Then we will define a sequence of terms $\{\phi^{(j)}\}_{j=0}^{J_0}$ for some $1 \ll J_0 \ll N$. Restricted to $\{u = -1\}$, each $\phi^{(j)}$  will decay uniformly as $v\to 0$ at least as fast as $v^{j-O_j\left(\epsilon\right)}$.  We emphasize that the constant in the $O_j$ may grow as $j \to \infty$. (This is why if we want to compute $P$ terms in our expansion, for some $P \gg 1$, then we need for the smallness of $\epsilon$ to depend on $P$.) We then will set $\phi[j] \doteq \phi^{(0)} + \cdots +\phi^{(j)}$. We will refer to $\phi[j]$ as the \emph{$j$th expansion for $\phi$}. If $\phi[j-1]$ has already been defined, then a definition of either $\phi^{(j)}$ or $\phi[j]$ will determine the other.  Lastly, we note that our estimates on the $\phi[j]$ will be so that given an expansion for $\phi$, we may immediately obtain a corresponding expansion for any (finite number of) angular derivative applied to $\phi$ by simply applying the (finite number of) angular derivative to each term in the expansion. 

Once we have computed $\phi[j]$ for each metric component $\phi$, we may then associate a corresponding spacetime metric $g[j]$ defined via~\eqref{iooioi284} by replacing the metric components with the corresponding $\phi[j]$. Associated to each $g[j]$ are the various set of Ricci coefficients $\psi$. When we want to refer to the $\psi$ associated to $g[j]$, we will write $\psi\left(g[j]\right)$. We note that it will generally \emph{not} be the case that $\psi[j] = \psi\left(g[j]\right)$. However, we will have $\left(\psi[j] -\psi\left(g[j]\right)\right)|_{u=-1} = O\left(v^{j-O\left(\epsilon\right))}\right)$.

In order to track the underlying smallness of our expansions, it will be convenient to work with the quantities:
\[\widetilde{\slashed{g}} \doteq \slashed{g} -\left(v+1\right)^2\mathring{\slashed{g}},\qquad \widetilde{\Omega^{-1}{\rm tr}\chi} \doteq \Omega^{-1}{\rm tr}\chi - 2\left(v+1\right)^{-1},\]
instead of $\slashed{g}$ and $\Omega^{-1}{\rm tr}\chi$. Our convention will be that given a known $j$th expansion for $\widetilde{\slashed{g}}$ or $\widetilde{\Omega^{-1}{\rm tr}\chi}$, we then define
\[\slashed{g}[j] \doteq \widetilde{\slashed{g}}[j] + \left(v+1\right)^2\mathring{\slashed{g}},\qquad \left(\Omega^{-1}{\rm tr}\chi\right)[j] \doteq \widetilde{\Omega^{-1}{\rm tr}\chi}[j] + 2(v+1)^{-1}.\]

In view of the relations~\eqref{20asdww},~\eqref{okodwok22}, and~\eqref{blahblah2} we will be able to remove the Ricci coefficients $\underline{\eta}$, $\zeta$, $\Omega{\rm tr}\underline{\chi}$, $\Omega\hat{\underline{\chi}}$, and $\omega$ from the null-structure equations and instead we will compute expansions only for $\left(\Omega,b,\slashed{g},\eta,\Omega^{-1}{\rm tr}\chi,\Omega^{-1}\hat{\chi}\right)$ (and any finite number of angular derivatives thereof). After all of our expansions are defined, one may, of course, revisit~\eqref{20asdww},~\eqref{okodwok22}, and~\eqref{blahblah2} to define expansions for the omitted quantities.

Finally, it is convenient to introduce the notation 
\[\left\vert\left\vert f\right\vert\right\vert_{\tilde{C}^j\left(\mathbb{S}^2\right)} \doteq \sum_{j_1+j_2 \leq j}\left\vert\left\vert \left(v\mathcal{L}_{\partial_v}\right)^{j_1}f\right\vert\right\vert_{C^{j_2}\left(\mathbb{S}^2\right)}.\]

\subsection{Computing the Expansion}
We now provide the details for our expansions.
\begin{theorem}\label{2omo2o492}Let $k \in \mathbb{Z}_{\geq 1}$, $J \in \mathbb{Z}_{ >0 }$ be any sufficiently large integer depending on $k$, and assume that $N$ is sufficiently large and $\epsilon$ is sufficiently small depending on both $k$ and $J$. Let 
\[\left(\slashed{g}^{({\rm seed})},b^{({\rm seed})},\left(\Omega\underline{\omega}\right)^{({\rm seed})},\left(\Omega^{-1} \hat{\chi}\right)^{(\rm seed,+)},\log\Omega^{(\rm seed,+)}\right)\]
be $1$-sided $\left(\epsilon,N\right)$-seed data. Then, for every integer $i \in [0,k-1]$ there exists Ricci coefficient expansions $\eta[i]$, $\widetilde{\Omega^{-1}{\rm tr}\chi}[i]$, $\left(\Omega^{-1}\hat{\chi}\right)[i]$ and $\left(\Omega\underline{\omega}\right)[i]$ defined as suitable $\mathbb{S}^2_{u,v}$ tensors on $\mathcal{W}^+\times \mathbb{S}^2$,  for every integer $i \in [0,k]$ there exists metric component expansions $\log\Omega[i]$, $b[i]$, and $\widetilde{\slashed{g}}[i]$ defined as suitable $\mathbb{S}^2_{u,v}$ tensors on $\mathcal{W}^+\times \mathbb{S}^2$, and for every integer $i \in [0,k]$ there is a metric $g[i]$ on $\mathcal{W}^+\times \mathbb{S}^2$ defined via~\eqref{iooioi284} by replacing the metric components with the corresponding metric component expansions so that $g[i]$ is a self-similar double-null metric which is $\left(\epsilon,N\right)$-regular up to $\{v = 0\}$ and so that the following hold \underline{along $\{u=-1\}$}:
\begin{enumerate}
	\item For every $i$ where the corresponding Ricci or metric coefficient is defined:
	\begin{align}\label{3kj2jkhkj23}
	&\left\vert\left\vert\left(\Omega\underline{\omega}[i]-\Omega\underline{\omega}^{(\rm seed)},\eta[i] - \eta^{(0)},\slashed{g}[i] - \slashed{g}^{(\rm seed)},b[i] - b^{(\rm seed)}\right)\right\vert\right\vert_{\tilde{C}^J\left(\mathbb{S}^2_{-1,v}\right)}
	 \lesssim \epsilon \left|v\right|^{1-\hat{C}\epsilon},
	\end{align}
	\begin{equation}\label{2klj3lkjl2}
	\left\vert\left\vert \left(\Omega^{-1}\hat{\chi}[i]-\left(\Omega^{-1}\hat{\chi}\right)^{(\rm seed,+)},\log\Omega[i] - \log\Omega^{(\rm seed,+)}\right)\right\vert\right\vert_{\tilde{C}^J\left(\mathbb{S}^2_{-1,v}\right)} \lesssim \epsilon \left|v\right|^{1-\hat{C}\epsilon},
	\end{equation}
	where $\eta^{(0)}$ is defined as in Definition~\ref{thisistheseed}.
	\item For every $0 \leq i \leq k-1$:
	\[\left\vert\left\vert \left(\eta^{(i)},\widetilde{\Omega^{-1}{\rm tr}\chi}^{(i)},\left(\Omega^{-1}\hat{\chi}\right)^{(i)}\right)\right\vert\right\vert_{\tilde{C}^J\left(\mathbb{S}^2_{-1,v}\right)} \lesssim \epsilon v^{i-D(i)\epsilon}.\]
	\item For every $0 \leq i \leq k$:
	\[\left\vert\left\vert \left(\left(\Omega\underline{\omega}\right)^{(i)},\left(\log\Omega\right)^{(i)},b^{(i)},\widetilde{\slashed{g}}^{(i)}\right)\right\vert\right\vert_{\tilde{C}^J\left(\mathbb{S}^2_{-1,v}\right)} \lesssim \epsilon v^{i-D(i)\epsilon}.\]
	\item  We have
	 \begin{align*}
	&\Big\vert\Big\vert\Big(\eta[k-1] - \eta\left(g[k]\right),\widetilde{\Omega^{-1}{\rm tr}\chi}[k-1]- \widetilde{\Omega^{-1}{\rm tr}\chi}\left(g[k]\right),
	\\ \nonumber &\qquad \qquad \qquad \qquad \left(\Omega^{-1}\hat{\chi}\right)[k-1]-\left(\Omega^{-1}\hat{\chi}\right)\left(g[k]\right)\Big)\Big\vert\Big\vert_{\tilde{C}^J\left(\mathbb{S}^2_{-1,v}\right)} \lesssim \epsilon v^{k-1-D(k-1)\epsilon},
	 \end{align*}
	 \[\left\vert\left\vert\left( \left(\Omega\underline{\omega}\right)[k] - \left(\Omega\underline{\omega}\right)\left(g[k]\right)\right)\right\vert\right\vert_{\tilde{C}^J\left(\mathbb{S}^2_{-1,v}\right)} \lesssim \epsilon v^{k - D(k)\epsilon}.\]
	 \item We have
	 \[\left\vert\left\vert \left({\rm Ric}_{33}\left(g[k]\right),{\rm Ric}_{3A}\left(g[k]\right),\left(R+{\rm Ric}_{34}\right)\left(g[k]\right)\right)\right\vert\right\vert_{\tilde{C}^J\left(\mathbb{S}^2_{-1,v}\right)} \lesssim \epsilon v^{1-D(1)\epsilon},\]
\[\left\vert\left\vert \left({\rm Ric}_{44}\left(g[k]\right),  {\rm Ric}_{4A}\left(g[k]\right) \right)\right\vert\right\vert_{\tilde{C}^J\left(\mathbb{S}^2_{-1,v}\right)} \lesssim \epsilon v^{(k-1) - D(k-1)\epsilon},\]
\[\left\vert\left\vert \left({\rm Ric}_{34}\left(g[k]\right),  \widehat{\rm Ric}_{AB}\left(g[k]\right) \right)\right\vert\right\vert_{\tilde{C}^J\left(\mathbb{S}^2_{-1,v}\right)}  \lesssim \epsilon v^{k - D(k-1)\epsilon}.\]
\end{enumerate}

Furthermore, if $g$ is any $1$-sided self-similar double-null metric which is $\left(\epsilon,N\right)$-regular up to $\{v = 0\}$ and which solves the Einstein vacuum equations for $\{v = 0\}$, then there exists $1$-sided seed data and corresponding metric $\hat{g}[\hat{N}]$ produced by the formal expansion procedure described above so that in the self-similar double-null coordinate system
	\begin{equation}\label{2klj3lkj2o9}
	\left\vert\left\vert \left(g_{\alpha\beta}-\hat{g}_{\alpha\beta}\right)\right\vert\right\vert_{\tilde{C}^J\left(\mathbb{S}^2_{-1,v}\right)} \lesssim_{\hat{N}} \epsilon \left|v\right|^{\hat{N}-\hat{C}\epsilon},
	\end{equation}
	for some $\hat{N} \gtrsim N$, where the implied constant is independent of $N$ and $\epsilon$. 	
	
	Finally, we note that the modifications needed are straightforward to prove a version of this theorem for (non-$1$ sided) $\left(\epsilon,N\right)$-regular seed data or for $1$-sided solutions defined on $\mathcal{W}^-\times \mathbb{S}^2$ instead of $\mathcal{W}^+\times \mathbb{S}^2$.
	\end{theorem}
\begin{proof}

We start by defining $\log\Omega^{(0)} \doteq \log\Omega^{(\rm seed,+)}$. In view of Lemma~\ref{laptrans}, this will satisfy the following estimate:
\begin{equation}\label{2om4o2029}
\left\vert\left\vert \log\Omega^{(0)}\right\vert\right\vert_{\tilde{C}^0\left(\mathbb{S}^2_{-1,v}\right)} \lesssim \epsilon \left(1+\left|\log(v)\right|\right), \qquad \left\vert\left\vert \log\Omega^{(0)}\right\vert\right\vert_{\tilde{C}^j\left(\mathbb{S}^2_{-1,v}\right)} \lesssim_j \epsilon v^{-D(j)\epsilon}
\end{equation}
We next set $b^{(0)} \doteq b^{(\rm seed)}$ and $\slashed{g}^{(0)} \doteq \slashed{g}^{(\rm seed)}$.

We may now define a metric $g[0]$ on $\mathcal{W}^+\times \mathbb{S}^2$ in the double-null form~\eqref{iooioi284} by substituting $\left(\Omega^{(0)},b^{(0)},\slashed{g}^{(0)}\right)$ for $\left(\Omega,b,\slashed{g}\right)$. We note that we do \emph{not} expect that the metric $g[0]$  extends continuously to $v = 0$ due to the degeneration of the estimate~\eqref{2om4o2029} as $v\to 0$. Finally, we observe that since $\chi\left(g[0]\right)$ clearly vanishes, it follows from~\eqref{eqnyaydivb} of Proposition~\ref{somanyformulassolittletime} and~\eqref{m2om3o2} that
\[{\rm Ric}_{33}\left(g[0]\right) = 0.\]

We now turn to the next order in the expansion. We define $\eta^{(0)}$ as in Definition~\ref{thisistheseed}. In view of Proposition~\ref{somestuimdie}  and Sobolev inequalities we will have
\[\left\vert\left\vert \eta^{(0)}\right\vert\right\vert_{\tilde{C}^j\left(\mathbb{S}^2\right)} \lesssim \epsilon.\]

We next define a function $\widetilde{\Omega^{-1}{\rm tr}\chi}^{(0)}$ along $\mathbb{S}^2$ by uniquely solving the equation
\begin{align}\label{2om3om2}
&\mathcal{L}_b\widetilde{\Omega^{-1}{\rm tr}\chi}^{(0)}+\widetilde{\Omega^{-1}{\rm tr}\chi}^{(0)}\left(-1+\slashed{\rm div}b\right) = -2\left(K-1\right) - 2\slashed{\rm div}b+ 2\slashed{\rm div}\eta + 2\left|\eta\right|^2,
\end{align}
where we substitute $\slashed{g} = \slashed{g}^{(\rm seed)}$, $\Omega\underline{\omega} = \left(\Omega\underline{\omega}\right)^{(\rm seed)}$, $b = b^{(\rm seed)}$, and $\eta = \eta^{(0)}$. We have obtained this equation by formally taking $(u,v) = (-1,0)$ in~\eqref{2o4oijoiouoiu2} from Proposition~\ref{somanyformulassolittletime}, using~\eqref{20asdww} from Proposition~\ref{somanyformulassolittletime}, and then dropping the term $R+{\rm Ric}_{34}$. We can use Proposition~\ref{somestuimdie} to solve for $\widetilde{\Omega^{-1}{\rm tr}\chi}^{(0)}$ which will then satisfy
\[\left\vert\left\vert\widetilde{\Omega^{-1}{\rm tr}\chi}^{(0)}\right\vert\right\vert_{\tilde{C}^j\left(\mathbb{S}^2\right)} \lesssim \epsilon.\]
We then extend $\widetilde{\Omega^{-1}{\rm tr}\chi}$ to $\{u = -1\}$ by declaring it to be independent of $v$.

We set $\left(\Omega^{-1}\hat{\chi}\right)^{(0)} \doteq \left(\Omega^{-1}\hat{\chi}\right)^{(\rm seed,+)}$, and have then the estimate
\begin{equation}\label{2ok4om2}
\left\vert\left\vert \left(\Omega^{-1}\hat{\chi}\right)^{(0)}\right\vert\right\vert_{\tilde{C}^j\left(\mathbb{S}^2_{-1,v}\right)} \lesssim \epsilon v^{-\epsilon D(j)},
\end{equation}
for suitable non-negative constants $D(j)$. (Note that $D(0)$ may be positive here.)

Now we turn to next term in the expansion of $\left(\Omega\underline{\omega}\right)$. We define, along $\{u = -1\}$ a function $\left(\Omega\underline{\omega}\right)^{(1)}$ by
\begin{align}\label{3omoj2ijo3}
&\left(\Omega\underline{\omega}\right)^{(1)} = \int_0^v\Omega^2\Bigg[\mathcal{L}_b\left(\Omega^{-1}{\rm tr}\chi\right) + \frac{1}{2}\left(\Omega^{-1}{\rm tr}\chi\right)\slashed{\rm div}b+ \frac{1}{2}\left|\eta\right|^2 
\\ \nonumber &\qquad \qquad - \eta\cdot\left(-\eta +2\slashed{\nabla}\log\Omega\right)-4\left(\Omega\underline{\omega}\right)\left(\Omega^{-1}{\rm tr}\chi\right) - 2\slashed{\rm div}\eta - 2\left|\eta\right|^2 + \left(\Omega^{-1}\hat{\chi}\right)\cdot\left(v\Omega^2\left(\Omega^{-1}\hat{\chi}\right)+\frac{1}{2}\slashed{\nabla}\hat{\otimes}b\right)\Bigg]\, dv,
\end{align}
where we substitute $\Omega = \Omega^{(0)}$, $\slashed{g} = \slashed{g}^{(0)}$, $\Omega\underline{\omega} = \left(\Omega\underline{\omega}\right)^{(0)}$, $b = b^{(0)}$, $\eta = \eta^{(0)}$, $\left(\Omega^{-1}{\rm tr}\chi\right) = \left(\Omega^{-1}{\rm tr}\chi\right)^{(0)}$, and $\Omega^{-1}\hat{\chi} = \left(\Omega^{-1}\hat{\chi}\right)^{(0)}$. We have obtained this equation by formally integrating along $\{u=-1\}$ the equation~\eqref{kfejifei} from Lemma~\ref{2kn3kn2i39}, using~\eqref{20asdww} from Proposition~\ref{somanyformulassolittletime}, and then dropping the term $\Omega^2{\rm Ric}_{34}$. It is immediate that we have the estimate
\[\left\vert\left\vert \left(\Omega\underline{\omega}\right)^{(1)}\right\vert\right\vert_{\tilde{C}^j\left(\mathbb{S}^2_{-1,v}\right)} \lesssim \epsilon v^{1-D(j)\epsilon}.\]

We are now ready to update the components of our metric with their next leading order term. We define, along $\{u=-1\}$, a $1$-parameter family of vectorfields $b^{(1)}$ on $\mathbb{S}^2$  in the coordinate frame by
\begin{equation}\label{om4om2o}
\left(b^{(1)}\right)^A = -4\int_0^v\Omega^2\left(\eta^A-\slashed{\nabla}^A\log\Omega\right)\, dv,
\end{equation}
where we substitute $\Omega = \Omega^{(0)}$, $\eta = \eta^{(0)}$, and $\slashed{g} = \slashed{g}^{(0)}$. We then have the estimate
\[\left\vert\left\vert b^{(1)}\right\vert\right\vert_{\tilde{C}^j\left(\mathbb{S}^2_{-1,v}\right)} \lesssim \epsilon v^{1-D(j)\epsilon}.\]

We next define, along $\{u=-1\}$, a $1$-parameter family of symmetric $(0,2)$-tensors $\widetilde{\slashed{g}}^{(1)}$ on $\mathbb{S}^2$ in the coordinate frame by
\begin{align}\label{2oom2o4o3}
\left(\widetilde{\slashed{g}}^{(1)}\right)_{AB} &= \int_0^v\Omega^2\left(\left(\widetilde{\Omega^{-1}{\rm tr}\chi}\right)\slashed{g}_{AB}+\left(\Omega^{-1}{\rm tr}\chi \right)\widetilde{\slashed{g}}_{AB} + \left(\Omega^{-1}\hat{\chi}\right)_{AB}\right)\, dv
+ 2\int_0^v\left(\Omega^2-1\right)\left(v+1\right)\mathring{\slashed{g}}_{AB}\, dv,
\end{align}
where we substitute all quantities on the right hand side with their $0$th expansion. We then have the estimate
\[\left\vert\left\vert \widetilde{\slashed{g}}^{(1)}\right\vert\right\vert_{\tilde{C}^j\left(\mathbb{S}^2_{-1,v}\right)} \lesssim \epsilon v^{1-D(j)\epsilon}.\]
Finally we use Lemma~\ref{laptrans} to define $\log\Omega^{(1)}$ to be the determined solution to
\begin{equation}\label{2mmo2}
\left(v\mathcal{L}_{\partial_v}+\mathcal{L}_{b^{(0)}}\right)\log\Omega^{(1)} = -\mathcal{L}_{b^{(1)}}\log\Omega^{(0)}+\left(\Omega\underline{\omega}\right)^{(1)},
\end{equation}
which decays to $0$ as $v\to 0$. We will then have the estimate
\[\left\vert\left\vert \log\Omega^{(1)}\right\vert\right\vert_{\tilde{C}^j\left(\mathbb{S}^2_{-1,v}\right)} \lesssim \epsilon v^{1-D(j)\epsilon}.\]

We then define $\log\Omega[1] \doteq \log\Omega^{(0)} + \log\Omega^{(1)}$, $b[1] \doteq b^{(0)} + b^{(1)}$, $\slashed{g}[1] \doteq \slashed{g}^{(0)} + \slashed{g}^{(1)}$,  and then obtain a corresponding double-null metric $g[1]$. One may easily then check that the following estimates are satisfied along $\{u=-1\}$:
\[\left\vert\left\vert \left[\left(\eta\left(g[1]\right) - \eta^{(0)}\right),\left(\Omega^{-1}\hat{\chi}\left(g[1]\right) - \left(\Omega^{-1}\hat{\chi}\right)^{(0)}\right),\left(\widetilde{\Omega^{-1}{\rm tr}\chi}\left(g[1]\right) - \widetilde{\Omega^{-1}{\rm tr}\chi}^{(0)}\right)\right]\right\vert\right\vert_{\tilde{C}^j\left(\mathbb{S}^2_{-1,v}\right)} \lesssim \epsilon v^{1-D(j)\epsilon}.\]
With these estimates in mind, we then also easily see from the way the expansions have been computed that
\begin{equation}\label{3om4om2o}
\left\vert\left\vert \left({\rm Ric}_{33},{\rm Ric}_{3A},\left(R+{\rm Ric}_{34}\right),\widehat{\rm Ric}_{AB},{\rm Ric}_{34}\right)\left(g[1]\right)\right\vert\right\vert_{\tilde{C}^j\left(\mathbb{S}^2_{-1,v}\right)} \lesssim \epsilon v^{1-D(j)\epsilon}.
\end{equation}

Now we come to the computation of the metric to second order. This will be the final case which requires a specific analysis, as the computation of the third order and higher expansions will be the same \emph{mutatis mutandis}. We start by defining $\eta_A^{(1)}$ in the coordinate frame along $\{u = -1\}$ via the formula
\begin{align}
\eta^{(1)}_A = \int_0^v\Omega^2&\Big[-\left(\Omega^{-1}\hat{\chi}\right)_{AB}\cdot\left(\eta^B-2\slashed{\nabla}^B\log\Omega\right)-\frac{1}{2}\left(\Omega^{-1}{\rm tr}\chi\right)\cdot\left(\eta_A-2\slashed{\nabla}_A\log\Omega\right)+\slashed{\nabla}^B\left(\Omega^{-1}\hat{\chi}\right)_{AB}
\\ \nonumber &\qquad \qquad \qquad -\frac{1}{2}\slashed{\nabla}_A\left(\Omega^{-1}{\rm tr}\chi\right)- \frac{1}{2}\left(\Omega^{-1}{\rm tr}\chi\right) \eta_A + \eta^B\left(\Omega^{-1}\hat{\chi}\right)_{AB}\Big]\, dv,
\end{align}
where on the right hand side we use the $0$th expansion of each quantity on the right hand side. We have obtained this equation by formally integrating along $\{u=-1\}$ the equation~\eqref{2lmo204202940204i} from Lemma~\ref{2om3mo020202}  and dropping the term $\Omega{\rm Ric}_{4A}$.  We immediately obtain the estimate
\[\left\vert\left\vert \eta^{(1)} \right\vert\right\vert_{\tilde{C}^j\left(\mathbb{S}^2_{-1,v}\right)} \lesssim \epsilon v^{1-D(j)\epsilon}.\]

Next we turn to $\widetilde{\Omega^{-1}{\rm tr}\chi}^{(1)}$. We define the function $\left(\Omega^{-1}{\rm tr}\chi\right)^{(1)}$ along $\{u = -1\}$ by the formula
\begin{align}\label{3ok4km2}
\widetilde{\Omega^{-1}{\rm tr}\chi}^{(1)} &\doteq -\int_0^v\Omega^2\left[2(v+1)^{-1}\widetilde{\Omega^{-1}{\rm tr}\chi} + \frac{1}{2}\left(\widetilde{\Omega^{-1}{\rm tr}\chi}\right)^2 + \left|\Omega^{-1}\hat{\chi}\right|^2\right]\, dv
\\ \nonumber &\qquad \qquad  - 2\int_0^v\left(\Omega^2-1\right)\left(v+1\right)^{-2}\, dv,
\end{align}
where on the right hand side we use the $0$th order expansions for the various quantities. We have obtained this equation by formally integrating along $\{u=-1\}$ the equation~\eqref{4trchi} (after conjugation of the equation with $\Omega^{-1}$),  dropping the term ${\rm Ric}_{44}$, and then writing the resulting equation in terms of $\widetilde{\Omega^{-1}{\rm tr}\chi}$.  This will satisfy the estimate
\[\left\vert\left\vert \widetilde{\Omega^{-1}{\rm tr}\chi}^{(1)}\right\vert\right\vert_{\tilde{C}^j\left(\mathbb{S}^2_{-1,v}\right)} \lesssim \epsilon v^{1-D(j)\epsilon}.\]

Now we come to $\left(\Omega^{-1}\hat{\chi}\right)^{(1)}$. We define $\mathscr{L}^{(0)}$ to be the operator from~\eqref{3o1984982} computed with respect to $b^{(0)}$ and $\slashed{g}^{(0)}$. We will define $\left(\Omega^{-1}\hat{\chi}\right)^{(1)}$ along $\{u = -1\}$ so that the following equation is satisfied
\begin{align}\label{kn2k22o330}
& v\mathcal{L}_{\partial_v}\left(\Omega^{-1}\hat{\chi}\right)^{(1)}_{AB} +\mathscr{L}^{(0)}\left(\Omega^{-1}\hat{\chi}\right)^{(1)}_{AB} -4\left(\Omega\underline{\omega}\right)^{(0)}\left(\Omega^{-1}\hat{\chi}\right)^{(1)}_{AB} = 
\\ \nonumber &\qquad \qquad \qquad \qquad \qquad \qquad -v\mathcal{L}_{\partial_v}\left(\Omega^{-1}\hat{\chi}\right)^{(0)}_{AB} -\mathscr{L}\left(\Omega^{-1}\hat{\chi}\right)^{(0)}_{AB} +4\left(\Omega\underline{\omega}\right)\left(\Omega^{-1}\hat{\chi}\right)^{(0)}_{AB}
\\ \nonumber &\qquad \qquad \qquad \qquad \qquad \qquad+ \left(\left(\slashed{\nabla}\hat\otimes \eta\right)_{AB} + \left(\eta\hat\otimes \eta\right)_{AB}\right) - \frac{1}{4}\left(\Omega^{-1}{\rm tr}\chi\right)\left(\slashed{\nabla}\hat{\otimes}b\right)_{AB},
 \end{align}
 where, unless otherwise indicated, the quantities on the right hand side are all computed with respect to their first order expansion, that is, we set $\slashed{g} = \slashed{g}[1]$, $b = b[1]$, etc. It follows from the definition of $\left(\Omega^{-1}\hat{\chi}\right)^{(0)}$ that the right hand side of~\eqref{kn2k22o330} is then $O\left(v^{1-D(j)\epsilon}\right)$ in the  $\tilde{C}^j\left(\mathbb{S}^2_{-1,v}\right)$ norm. In particular, we may use Lemma~\ref{laptrans} to define $\left(\Omega^{-1}\hat{\chi}\right)^{(1)}$ to be the solution to~\eqref{kn2k22o330} obtained by integrating from $\{v = 0\}$ and which satisfies
 \[\left\vert\left\vert \left(\Omega^{-1}\hat{\chi}\right)^{(1)}\right\vert\right\vert_{\tilde{C}^j\left(\mathbb{S}^2_{-1,v}\right)} \lesssim \epsilon v^{1-D(j)\epsilon}.\]

 We are now ready for $\left(\Omega\underline{\omega}\right)^{(2)}$. It is, of course, equivalent to define $\left(\Omega\underline{\omega}\right)[2] = \left(\Omega\underline{\omega}\right)^{(2)} + \left(\Omega\underline{\omega}\right)^{(1)} + \left(\Omega\underline{\omega}\right)^{{\rm seed}}$. In analogy with~\eqref{3omoj2ijo3}, we define $\left(\Omega\underline{\omega}\right)[2]$ along $\{u=-1\}$ by
 \begin{align}\label{knkrk2k}
&\left(\Omega\underline{\omega}\right)[2] -\left(\Omega\underline{\omega}\right)^{(\rm seed)}= \int_0^v\Omega^2\Bigg[\mathcal{L}_b\left(\Omega^{-1}{\rm tr}\chi\right) + \frac{1}{2}\left(\Omega^{-1}{\rm tr}\chi\right)\slashed{\rm div}b+ \frac{1}{2}\left|\eta\right|^2 
\\ \nonumber &\qquad \qquad - \eta\cdot\left(-\eta+2\slashed{\nabla}\log\Omega\right)-4\left(\Omega\underline{\omega}\right)\left(\Omega^{-1}{\rm tr}\chi\right) - 2\slashed{\rm div}\eta - 2\left|\eta\right|^2 + \left(\Omega^{-1}\hat{\chi}\right)\cdot\left(v\Omega^2\left(\Omega^{-1}\hat{\chi}\right)+\frac{1}{2}\slashed{\nabla}\hat{\otimes}b\right)\Bigg]\, dv,
\end{align}
where on the right hand side we use the $1$st expansion of each quantity. It is then immediate from the definition of $\left(\Omega\underline{\omega}\right)^{(1)}$ that
\[\left\vert\left\vert \left(\Omega\underline{\omega}\right)^{(2)}\right\vert\right\vert_{\tilde{C}^j\left(\mathbb{S}^2_{-1,v}\right)} \lesssim \epsilon v^{2-D(j)\epsilon}.\]

We are now ready to update the metric components. We start by defining a $1$-parameter family of vectorfields $\left(b[2]\right)^A$ on $\mathbb{S}^2$  in the coordinate frame by, in analogy with~\eqref{om4om2o},
\begin{equation}\label{3km2om2o2o}
\left(b[2]\right)^A - \left(b[0]\right)^A= -4\int_0^v\Omega^2\left(\eta^A-\slashed{\nabla}^A\log\Omega\right)\, dv,
\end{equation}
where we substitute $\Omega = \Omega[1]$, $\eta = \eta[1]$, and $\slashed{g} = \slashed{g}[1]$. We then have the estimate
\[\left\vert\left\vert b^{(2)}\right\vert\right\vert_{\tilde{C}^j\left(\mathbb{S}^2_{-1,v}\right)} \lesssim \epsilon v^{2-D(j)\epsilon}.\]
We next define, along $\{u=-1\}$, a $1$-parameter family of symmetric $(0,2)$-tensors $\left(\slashed{g}[2]\right)_{AB}$ on $\mathbb{S}^2$ in the coordinate frame by, in analogy with~\eqref{2oom2o4o3},
\begin{align}\label{3kn2kn2kn}
&\left(\widetilde{\slashed{g}}[2]\right)_{AB}-\left(\widetilde{\slashed{g}}[0]\right)_{AB} = \\ \nonumber &  \int_0^v\Omega^2\left(\left(\widetilde{\Omega^{-1}{\rm tr}\chi}\right)\slashed{g}_{AB}+\left(\Omega^{-1}{\rm tr}\chi \right)\widetilde{\slashed{g}}_{AB} + \left(\Omega^{-1}\hat{\chi}\right)_{AB}\right)\, dv
+ 2\int_0^v\left(\Omega^2-1\right)\left(v+1\right)\mathring{\slashed{g}}_{AB}\, dv,
\end{align}
where we now use the $1$st expansion for each of the quantities on the right hand side. It is then immediate that we have 
\[\left\vert\left\vert \widetilde{\slashed{g}}^{(2)}\right\vert\right\vert_{\tilde{C}^j\left(\mathbb{S}^2_{-1.v}\right)} \lesssim \epsilon v^{2-D(j)\epsilon}.\]
Finally we come to $\log\Omega^{(2)}$. In analogy to~\eqref{2mmo2} we use Lemma~\ref{laptrans} to define $\log\Omega^{(2)}$ to be the unique solution to
\begin{equation}\label{k3mo2mo2mo}
\left(v\mathcal{L}_{\partial_v}+\mathcal{L}_{b^{(0)}}\right)\log\Omega^{(2)} = -\left(v\mathcal{L}_{\partial_v} + \mathcal{L}_{b[2]}\right)\log\Omega[1]+\left(\Omega\underline{\omega}\right)^{(2)},
\end{equation}
which decays to $0$ as $v\to 0$. We will then have the estimate
\[\left\vert\left\vert \log\Omega^{(2)}\right\vert\right\vert_{\tilde{C}^j\left(\mathbb{S}^2_{-1,v}\right)} \lesssim \epsilon v^{2-D(j)\epsilon}.\]
We then obtain from $\log\Omega[2]$, $b[2]$, and $\slashed{g}[2]$ a corresponding double-null metric $g[2]$. One may easily then check that the following estimates are satisfied along $\{u=-1\}$:
\[\left\vert\left\vert \left[\left(\eta\left(g[2]\right) - \eta[1]\right),\left(\Omega^{-1}\hat{\chi}\left(g[2]\right) - \left(\Omega^{-1}\hat{\chi}\right)[1]\right),\left(\widetilde{\Omega^{-1}{\rm tr}\chi}\left(g[2]\right) - \left(\Omega^{-1}{\rm tr}\chi\right)[1]\right)\right]\right\vert\right\vert_{\tilde{C}^j\left(\mathbb{S}^2_{-1,v}\right)} \lesssim \epsilon v^{2-D(j)\epsilon}.\]
With these in mind, we then also obtain that 
\begin{equation}\label{m2lm2lml2}
\left\vert\left\vert \left({\rm Ric}_{33},{\rm Ric}_{44},{\rm Ric}_{4A}, {\rm Ric}_{3A},\left(R+{\rm Ric}_{34}\right)\right)\left(g[2]\right) \right\vert\right\vert_{\tilde{C}^j\left(\mathbb{S}^2_{-1,v}\right)} \lesssim \epsilon v^{1-D(j)\epsilon},
\end{equation}
\begin{equation}\label{2km2m4o}
\left\vert\left\vert \left(\widehat{\rm Ric}_{AB},{\rm Ric}_{34}\right)\left(g[2]\right)\right\vert\right\vert_{\tilde{C}^j\left(\mathbb{S}^2_{-1,v}\right)} \lesssim \epsilon v^{2-D(j)\epsilon},
\end{equation}

We are now ready to define the rest of the expansion by an inductive procedure. Let $k \geq 3$ and assume that the $k-2$ expansions of the Ricci coefficients $\eta$, $\widetilde{\Omega^{-1}{\rm tr}\chi}$, $\Omega^{-1}\hat{\chi}$ have been defined and that the the $k-1$ expansions of $\Omega\underline{\omega}$, $\log\Omega$, $b$, and $\widetilde{\slashed{g}}$ have been defined. 

We then define $\eta[k-1]$ in the coordinate frame along $\{u = -1\}$ via the formula
\begin{align}
\eta[k-1] - \eta^{(0)} = \int_0^v\Omega^2&\Big[-\left(\Omega^{-1}\hat{\chi}\right)\cdot\left(\eta-2\slashed{\nabla}\log\Omega\right)-\frac{1}{2}\left(\Omega^{-1}{\rm tr}\chi\right)\cdot\left(\eta-2\slashed{\nabla}\log\Omega\right)+\slashed{\rm div}\left(\Omega^{-1}\hat{\chi}\right)
\\ \nonumber &\qquad \qquad \qquad -\frac{1}{2}\slashed{\nabla}\left(\Omega^{-1}{\rm tr}\chi\right)- \frac{1}{2}\left(\Omega^{-1}{\rm tr}\chi\right) \eta + \eta\cdot\left(\Omega^{-1}\hat{\chi}\right)\Big]\, dv,
\end{align}
where we use the $k-2$ expansions of the various quantities on the right hand side. We define the function $\widetilde{\Omega^{-1}{\rm tr}\chi}[k-1]$ along $\{u = -1\}$ by the formula
\begin{align}\label{3ok4km2}
\widetilde{\Omega^{-1}{\rm tr}\chi}[k-1] -\widetilde{\Omega^{-1}{\rm tr}\chi}^{(0)}&\doteq -\int_0^v\Omega^2\left[2(v+1)^{-1}\widetilde{\Omega^{-1}{\rm tr}\chi} + \frac{1}{2}\left(\widetilde{\Omega^{-1}{\rm tr}\chi}\right)^2 + \left|\Omega^{-1}\hat{\chi}\right|^2\right]\, dv
\\ \nonumber &\qquad \qquad  - 2\int_0^v\left(\Omega^2-1\right)\left(v+1\right)^{-2}\, dv,
\end{align}
where we use the $k-2$ expansions of the various quantities on the right hand side. We then define $\left(\Omega^{-1}\hat{\chi}\right)^{(k-1)}$ along $\{u = -1\}$ so that the following equation is satisfied
\begin{align}\label{k2mm2m2o}
& v\mathcal{L}_{\partial_v}\left(\Omega^{-1}\hat{\chi}\right)^{(k-1)}_{AB} +\mathscr{L}^{(0)}\left(\Omega^{-1}\hat{\chi}\right)^{(k-1)}_{AB} -4\left(\Omega\underline{\omega}\right)^{(0)}\left(\Omega^{-1}\hat{\chi}\right)^{(k-1)}_{AB} = 
\\ \nonumber &\qquad \qquad \qquad \qquad -v\mathcal{L}_{\partial_v}\left(\Omega^{-1}\hat{\chi}\right)[k-2]_{AB} -\mathscr{L}\left(\Omega^{-1}\hat{\chi}\right)[k-2]_{AB} +4\left(\Omega\underline{\omega}\right)\left(\Omega^{-1}\hat{\chi}\right)[k-2]_{AB}
\\ \nonumber &\qquad \qquad \qquad \qquad \qquad \qquad+ \left(\left(\slashed{\nabla}\hat\otimes \eta\right)_{AB} + \left(\eta\hat\otimes \eta\right)_{AB}\right) - \frac{1}{4}\left(\Omega^{-1}{\rm tr}\chi\right)\left(\slashed{\nabla}\hat{\otimes}b\right)_{AB},
 \end{align}
 where, unless otherwise indicated, the quantities on the right hand side are all computed with respect to their $k-1$st expansion. We then define $\left(\Omega\underline{\omega}\right)[k]$ along $\{u=-1\}$ by
 \begin{align}\label{kn3no2o}
&\left(\Omega\underline{\omega}\right)[k] -\left(\Omega\underline{\omega}\right)^{(0)}= \int_0^v\Omega^2\Bigg[\mathcal{L}_b\left(\Omega^{-1}{\rm tr}\chi\right) + \frac{1}{2}\left(\Omega^{-1}{\rm tr}\chi\right)\slashed{\rm div}b+ \frac{1}{2}\left|\eta\right|^2 
\\ \nonumber &\qquad \qquad - \eta\cdot\left(-\eta+2\slashed{\nabla}\log\Omega\right)-4\left(\Omega\underline{\omega}\right)\left(\Omega^{-1}{\rm tr}\chi\right) - 2\slashed{\rm div}\eta - 2\left|\eta\right|^2 + \left(\Omega^{-1}\hat{\chi}\right)\cdot\left(v\Omega^2\left(\Omega^{-1}\hat{\chi}\right)+\frac{1}{2}\slashed{\nabla}\hat{\otimes}b\right)\Bigg]\, dv,
\end{align}
where on the right hand side we use the $k-1$st expansion of each quantity. Next we define $\left(b[k]\right)^A$ on $\mathbb{S}^2$  in the coordinate frame by
\begin{equation}\label{k3m2om2o02003}
\left(b[k]\right)^A - \left(b[0]\right)^A= -4\int_0^v\Omega^2\left(\eta^A-\slashed{\nabla}^A\log\Omega\right)\, dv,
\end{equation}
where we use the $k-1$st expansion on the right hand side. We next define, along $\{u=-1\}$, $\left(\slashed{g}[k]\right)_{AB}$ on $\mathbb{S}^2$ in the coordinate frame by
\begin{align}\label{km2om2oop1}
&\left(\widetilde{\slashed{g}}[k]\right)_{AB}-\left(\widetilde{\slashed{g}}[0]\right)_{AB} = \\ \nonumber &  \int_0^v\Omega^2\left(\left(\widetilde{\Omega^{-1}{\rm tr}\chi}\right)\slashed{g}_{AB}+\left(\Omega^{-1}{\rm tr}\chi \right)\widetilde{\slashed{g}}_{AB} + \left(\Omega^{-1}\hat{\chi}\right)_{AB}\right)\, dv
+ 2\int_0^v\left(\Omega^2-1\right)\left(v+1\right)\mathring{\slashed{g}}\, dv,
\end{align}
where we now use the $k-1$st expansion for each of the quantities on the right hand side. Finally we use Lemma~\ref{laptrans} to define $\log\Omega^{(k)}$ to be the unique solution to
\begin{equation}\label{3om2om2o4i}
\left(v\mathcal{L}_{\partial_v}+\mathcal{L}_{b^{(0)}}\right)\log\Omega^{(k)} = -\left(v\mathcal{L}_{\partial_v} + \mathcal{L}_{b[k]}\right)\log\Omega[k-1]+\left(\Omega\underline{\omega}\right)[k].
\end{equation}

It is now straightforward to see that this definition of the expansions finishes the proof of the proposition with the exception of the final estimate~\eqref{2klj3lkj2o9}. To prove this final statement, we observe the following facts about the metric $g$:
\begin{enumerate}
	\item The quantities $\slashed{g}$, $b$, $\Omega\underline{\omega}$, and $\eta$ must all have limits as $v\to 0$, and, in view of~\eqref{eqnyaydivb} and~\eqref{3pk2o294},  these limits at $(u,v) = (-1,0)$ must satisfy~\eqref{m2om3o2} and~\eqref{3ioj2oij4oi2}. These limiting values then define $\left(\slashed{g}^{(\rm seed)},b^{(\rm seed)},\Omega\underline{\omega}^{(\rm seed)}\right)$.
	\item By integrating \eqref{4trchi}, we see that $\Omega^{-1}{\rm tr}\chi$ has a limiting value as $v\to 0$. Moreover, in view of~\eqref{2o4oijoiouoiu2}, this limiting value of $\Omega^{-1}{\rm tr}\chi$ must satisfy the equation~\eqref{2om3om2}. Then,  in view of~\eqref{blahblah1} and~\eqref{kwdkodwok23dg}, there must exist $\log\Omega^{(\rm seed,+)}$ and $\left(\Omega^{-1}\hat{\chi}\right)^{(\rm seed,+)}$ so that $\log\Omega - \log\Omega^{(\rm seed,+)}$ and $\Omega^{-1}\hat{\chi} - \left(\Omega^{-1}\hat{\chi}\right)^{(\rm seed,+)}$  are the solutions given by Lemma~\ref{laptrans} to their respective equations which have the maximal possible decay to $\{v = 0\}$.
	\item Having determined the seed data which corresponds to $g$, we now observe that $b$, $\slashed{g}$, $\Omega^{-1}{\rm tr}\chi$, $\eta$, and $\Omega\underline{\omega}$ all satisfy equations which relate a $\mathcal{L}_{\partial_v}$ derivative to (angular derivatives of) the other Ricci coefficients. 
	\item Now we observe that our formal expansions are computed above exactly by inductively using these $\mathcal{L}_{\partial_v}$ equations to compute further terms in the expansion for $b$, $\slashed{g}$, $\Omega^{-1}{\rm tr}\chi$, $\eta$, and $\Omega\underline{\omega}$ and then updating $\log\Omega- \log\Omega^{(\rm seed,+)}$ and $\Omega^{-1}\hat{\chi} - \left(\Omega^{-1}\hat{\chi}\right)^{(\rm seed,+)}$ by integrating their equations from $\{v = 0\}$ with an application of Lemma~\ref{laptrans}. 
\end{enumerate}
It thus follows that~\eqref{2klj3lkj2o9} must hold.

\end{proof}

\subsection{Constraint Propagation}\label{woahconstraintsarecool}

Since the desired regularity statements for the expansions $g[k]$ in the homothetic gauge follow from Proposition~\ref{nonononononoo}, in order to finish the proof of Theorem~\ref{theoremexpand}, it only remains to show that the full Ricci tensor of the metrics $g[k]$ produced by Theorem~\ref{2omo2o492} vanish to a sufficiently higher order so as to satisfy~\eqref{3oij2oi829} and~\eqref{9198j32}. \emph{We note that, in view of the underlying self-similarity, it suffices to establish that~\eqref{3oij2oi829} and~\eqref{9198j32} hold along $\{u = -1\}$.}

In order to do this we will exploit the well known fact that the Einstein tensor is always divergence free. This will impose certain differential relations between the Ricci curvature components, and allow us to improve on the vanishing established in Theorem~\ref{2omo2o492}. In the following lemma (which is Lemma 16.3 from~\cite{nakedinterior}) we write out in a double-null formalism the results of the fact that the Einstein tensor is divergence free.
\begin{lemma}Every $3+1$ dimensional Lorentzian spacetime satisfies the following:
\begin{align}\label{3kcom1}
&-\frac{1}{2}\nabla_3{\rm Ric}_{44} +2\underline{\omega}{\rm Ric}_{44} + 2\eta^A{\rm Ric}_{4A} -\frac{1}{2}\nabla_4\left({\rm Ric}_{34}+ R\right) + \underline{\eta}^A{\rm Ric}_{A4}
\\ \nonumber &\qquad +\slashed{\nabla}^A{\rm Ric}_{A4} - \frac{1}{2}{\rm tr}\underline{\chi}{\rm Ric}_{44} - \frac{1}{2}{\rm tr}\chi\left({\rm Ric}_{34} + R\right) + \zeta^A{\rm Ric}_{A4}-\frac{1}{2}{\rm tr}\chi{\rm Ric}_{34} - \hat{\chi}^{AB}\widehat{{\rm Ric}}_{AB} = 0,
\end{align}

\begin{align}\label{3kcom2}
&-\frac{1}{2}\nabla_4{\rm Ric}_{33} +2\omega{\rm Ric}_{33} + 2\underline{\eta}^A{\rm Ric}_{3A} -\frac{1}{2}\nabla_3\left({\rm Ric}_{34}+ R\right) + \eta^A{\rm Ric}_{A3}
\\ \nonumber &\qquad +\slashed{\nabla}^A{\rm Ric}_{A3} - \frac{1}{2}{\rm tr}\chi{\rm Ric}_{33} - \frac{1}{2}{\rm tr}\underline{\chi}\left({\rm Ric}_{34} + R\right) - \zeta^A{\rm Ric}_{A3}-\frac{1}{2}{\rm tr}\underline{\chi}{\rm Ric}_{34} - \hat{\underline{\chi}}^{AB}\widehat{{\rm Ric}}_{AB} = 0,
\end{align}

\begin{align}\label{3kcom3}
&-\frac{1}{2}\nabla_3{\rm Ric}_{4A} + \underline{\omega}{\rm Ric}_{4A} + \eta^B{\rm Ric}_{BA} + \frac{1}{2}\eta_A{\rm Ric}_{34} -\frac{1}{2}\nabla_4{\rm Ric}_{3A} + \omega{\rm Ric}_{3A} + \underline{\eta}^B{\rm Ric}_{BA} + \frac{1}{2}\underline{\eta}_A{\rm Ric}_{34}
\\ \nonumber &\qquad \slashed{\nabla}^B\widehat{{\rm Ric}}_{BA} + \frac{1}{2}\slashed{\nabla}_A{\rm Ric}_{34}-\frac{1}{2}{\rm tr}\underline{\chi}{\rm Ric}_{4A} - \frac{1}{2}{\rm tr}\chi{\rm Ric}_{3A} - \frac{1}{2}\underline{\chi}_A^{\ \ B}{\rm Ric}_{B4}-\frac{1}{2}\chi_A^{\ \ B}{\rm Ric}_{B3} = 0.
\end{align}

\end{lemma}

We can now prove that the Ricci tensors vanishes quickly as $v\to 0$ in the double-null coordinate system.
\begin{proposition}Let $g[k]$ be a metric produced by Theorem~\ref{2omo2o492}. Then, along $\{u = -1\}$, we have
\begin{equation}\label{3oiw2j23ijo32ijo23}
\left\vert\left\vert {\rm Ric}\left(g[k]\right)\right\vert\right\vert_{\tilde{C}^{J-2}\left(\mathbb{S}^2_{-1,v}\right)} \lesssim v^{k-1 - \epsilon D(k-1)}.
\end{equation}
\end{proposition}
\begin{proof}From~\eqref{3kcom1} and the already established estimates for the Ricci tensor from Theorem~\ref{2omo2o492}, we obtain
\begin{equation}\label{3lm3om23o}
\mathcal{L}_{\partial_v}\left(R+{\rm Ric}_{34}\right) + 2\Omega^2\left(\Omega^{-1}{\rm tr}\chi\right)\left(R+{\rm Ric}_{34}\right)= H_1,
\end{equation}
where $H_1$ satisfies
\[\left\vert\left\vert H_1\right\vert\right\vert_{\tilde{C}^{J-1}\left(\mathbb{S}^2_{-1,v}\right)} \lesssim \epsilon v^{k-1 - \epsilon D\left(k-1\right)}.\]
Since $R + {\rm Ric}_{34}$ vanishes when $v = 0$. We immediately obtain from~\eqref{3lm3om23o} that
\begin{equation}\label{2pom3om2}
\left\vert\left\vert \left(R+{\rm Ric}_{34}\right) \right\vert\right\vert_{\tilde{C}^{J-1}\left(\mathbb{S}^2_{-1,v}\right)} \lesssim \epsilon v^{k - \epsilon D\left(k-1\right)}.
\end{equation}

Next, from~\eqref{3kcom1}, the already established estimates for the Ricci tensor from Theorem~\ref{2omo2o492}, the estimate~\eqref{2pom3om2}, and the fact that ${\rm Ric}_{AB} = \frac{1}{2}\left(R+{\rm Ric}_{34}\right)\slashed{g}_{AB} + \widehat{\rm Ric}_{AB}$, we obtain that
\begin{equation}\label{3lmom2oio4}
\nabla_{\partial_v}\left(\Omega{\rm Ric}_{3A}\right) + \Omega^2\left(\frac{3}{2}\Omega^{-1}{\rm tr}\chi + \Omega^{-1}\hat{\chi}   \right)\left(\Omega{\rm Ric}_{3A}\right) = H_2,
\end{equation}
where 
\[\left\vert\left\vert H_2\right\vert\right\vert_{\tilde{C}^{J-1}\left(\mathbb{S}^2_{-1,v}\right)} \lesssim \epsilon v^{k-1 - \epsilon D\left(k-1\right)}.\]
Since $\Omega{\rm Ric}_{3A}$ vanishes when $v = 0$, we may integrate this from $v=0$ to obtain
\begin{equation}\label{i3jri3jri3jri3}
\left\vert\left\vert {\rm Ric}_{3A}\right\vert\right\vert_{C^{J-1}\left(\mathbb{S}^2_{-1,v}\right)} \lesssim \epsilon v^{k- \epsilon D\left(k-1\right)}.
\end{equation}

Finally, from~\eqref{3kcom2}, the already established estimates for the Ricci tensor from Theorem~\ref{2omo2o492}, the estimate~\eqref{2pom3om2}, and the estimate~\eqref{i3jri3jri3jri3}, we obtain that
\begin{equation}\label{3o4ok395rij39}
\mathcal{L}_{\partial_v}\left(\Omega^2{\rm Ric}_{33}\right) + \Omega^2\left(\Omega^{-1}{\rm tr}\chi\right)\left(\Omega^2{\rm Ric}_{33}\right) = H_3,
\end{equation}
where 
\[\left\vert\left\vert H_3 \right\vert\right\vert_{\tilde{C}^{J-2}\left(\mathbb{S}^2_{-1,v}\right)} \lesssim \epsilon v^{k-1 - \epsilon D\left(k-1\right)}.\]
Since $\Omega^2{\rm Ric}_{33}$ vanishes when $v= 0$, we may integrate~\eqref{3o4ok395rij39} to obtain that 
\begin{equation}\label{2kmek2nke2}
\left\vert\left\vert {\rm Ric}_{33}\right\vert\right\vert_{\tilde{C}^{J-2}\left(\mathbb{S}^2_{-1,v}\right)} \lesssim \epsilon v^{k- \epsilon D\left(k-1\right)}.
\end{equation}
This completes the proof since all components of the Ricci curvature tensor have now been estimated.
\end{proof}

It remains to check that in the homothetic gauge,~\eqref{9198j32} holds (in a weak sense).
\begin{proposition}\label{ioioioioio39289389}Let $g[k]$ be a metric produced by Theorem~\ref{2omo2o492}. Then, we may apply Proposition~\ref{nonononononoo} to put $g[k]$ into the homothetic gauge. In the homothetic gague we then have that~\eqref{9198j32} holds in a weak sense.
\end{proposition} 
\begin{proof}This proof is very similar to the proof of Theorem 18.1 in~\cite{nakedinterior}, and we will thus be brief in our arguments. Let $\left(\hat{v},u,\theta^A\right)$ denote the coordinates defined in Proposition~\ref{nonononononoo}. A computation shows that it suffices to check that~\eqref{3oiw2j23ijo32ijo23} holds weakly in $\left(\hat{v},u,\theta^A\right)$ coordinates (with $v$ replaced by $\hat{v}$).

Let us refer to $\partial_u$ and $\partial_{\theta^A}$ in $\left(\hat{v},u,\theta^A\right)$ coordinates as $\partial_{\hat{u}}$ and $\partial_{\hat{\theta}^A}$. The symbols $\partial_u$, $\partial_{\theta^A}$, and $e_A$ will always be defined with respect to the $\left(v,u,\theta^A\right)$ coordinates. Under the change of coordinates $\left(v,u,\theta^A\right) \mapsto \left(\hat{v},u,\theta^A\right)$ we have
\begin{equation}\label{3oij23jioioj}
\partial_v \mapsto \Omega^2\partial_{\hat{v}},\qquad \partial_u \mapsto \left(\int_0^v \partial_u\left(\Omega^2\right)\, dv\right)\partial_{\hat{v}} + \partial_{\hat{u}},\qquad  \partial_{\theta^A} \mapsto \left(\int_0^v\partial_{\theta^A}\left(\Omega^2\right)\, dv\right)\partial_{\hat{v}}  +\partial_{\hat{\theta}^A}.
\end{equation}
In particular, it is clear that~\eqref{3oiw2j23ijo32ijo23} continues to hold in the $\left(\hat{v},u,\theta^A\right)$ coordinates for $\hat{v} \neq 0$. 

We now study the behavior of the Ricci tensor at $\{\hat{v} = 0\}$. We let $X,Y$ denote vector fields and introduce the convention that all Ricci tensors refer to the Ricci tensor of $g[k]$. We will say that~\eqref{3oiw2j23ijo32ijo23} holds (weakly) for ${\rm Ric}\left(X,Y\right)$  if~\eqref{3oiw2j23ijo32ijo23} holds (weakly) with ${\rm Ric}\left(g[k]\right)$ replaced by ${\rm Ric}\left(X,Y\right)$. Similarly, for any function $h(v)$, we will say that~\eqref{3oiw2j23ijo32ijo23} holds (weakly) for ${\rm Ric}\left(X,Y\right)$ after multiplying both sides by $h(v)$  if~\eqref{3oiw2j23ijo32ijo23} holds (weakly) with ${\rm Ric}\left(g[k]\right)$ replaced by $h(v){\rm Ric}\left(X,Y\right)$ and the right hand side of~\eqref{3oiw2j23ijo32ijo23} multiplied by $h(v)$.  The following two principles will be useful:
\begin{itemize}
	\item If every term in the distribution ${\rm Ric}(X,Y)$ is in fact a locally integrable function, then~\eqref{3oiw2j23ijo32ijo23} holds for ${\rm Ric}\left(X,Y\right)$ weakly everywhere in $\left(\hat{v},u,\theta^A\right)$ coordinates. 
	\item If ${\rm Ric}(X,Y) = P + \partial_{\hat{v}}f$ where $P$ is locally integrable and $f$ is continuous at $\hat{v} = 0$, then the weak derivative $\partial_{\hat{v}}f$ is an $L^1_{\rm loc}$ function and~\eqref{3oiw2j23ijo32ijo23} also continues to hold weakly everywhere in $\left(\hat{v},u,\theta^A\right)$ coordinates. 
\end{itemize}
In view of these observations and Proposition~\ref{thenullstructeqns}, one may then check that for $X,Y \in \{\Omega^{-1}e_4,\Omega e_3,e_A\}$, we will have that~\eqref{3oiw2j23ijo32ijo23} holds weakly for ${\rm Ric}\left(X,Y\right)$ and continues to hold weakly for ${\rm Ric}\left(X,Y\right)$ after multiplication by any function $h(v)$ which satisfies, say, $|h(v)| \lesssim |v|^{1/2}$ (see the formulas for the Ricci tensor in the proof of Theorem 18.1 of~\cite{nakedinterior}; note that the powers of the lapse $\Omega$ multiplying the frame vector fields result in the cancellation of all terms in ${\rm Ric}\left(X,Y\right)$ which would potentially contain $\omega$). For the reader's convenience we reproduce the formulas for the Ricci coefficients in the frame $\left\{\Omega^{-1}e_4,\Omega e_3,e_A\right\}$ from~\cite{nakedinterior}:
\begin{align*}
{\rm Ric}\left(\Omega^{-1}e_4,\Omega^{-1}e_4\right)& = -\Omega^{-1}e_4\left(\Omega^{-1}{\rm tr}\chi\right) - \frac{1}{2}\left({\rm tr}\chi\right)^2 -\left|\hat{\chi}\right|^2,
\\ \nonumber {\rm Ric}\left(\Omega^{-1}e_4,e_A\right) &= -\Omega^{-1}e_4 \eta -\Omega^{-1}\hat{\chi}\cdot\left(\eta-\underline{\eta}\right) +\Omega^{-1}\slashed{\nabla}^B\hat{\chi}_{AB} -\frac{1}{2}\Omega^{-1}\slashed{\nabla}_A{\rm tr}\chi 
\\ \nonumber &\qquad -\frac{1}{2}\Omega^{-1}{\rm tr}\chi \zeta_A + \Omega^{-1}\zeta^B\hat{\chi}_{AB},
\\ \nonumber {\rm Ric}\left(\Omega^{-1}e_4,\Omega e_3\right) &= 4\Omega^{-1}e_4\left(\Omega\underline{\omega}\right) -2\left|\eta\right|^2 +4\eta\cdot\underline{\eta} + \left(\Omega e_3\right)\left(\Omega^{-1}{\rm tr}\chi\right) + \frac{1}{2}\Omega^{-1}{\rm tr}\chi \left(\Omega{\rm tr}\underline{\chi}\right) -2\left(\Omega\underline{\omega}\right)\Omega^{-1}{\rm tr}\chi
\\ \nonumber &\qquad -2\slashed{\rm div}\eta +2\left|\eta\right|^2 + \left(\Omega^{-1}\hat{\chi}\right)\cdot \Omega\hat{\chi},
\\ \nonumber {\rm R} + {\rm Ric}\left(\Omega^{-1}e_4,\Omega e_3\right) &= \Omega e_3\left(\Omega^{-1}{\rm tr}\chi\right) + \frac{1}{4}\Omega^{-1}{\rm tr}\chi\left(\Omega{\rm tr}\underline{\chi}\right) -4 \Omega\underline{\omega}\left(\Omega^{-1}{\rm tr}\chi\right) -2\slashed{\rm div}\eta - 2\left|\eta\right|^2 + K,  
\\ \nonumber \widehat{\rm Ric}\left(e_A,e_B\right) &= \Omega e_3\left(\Omega^{-1}\hat{\chi}\right)_{AB} +\frac{1}{2}\left(\Omega{\rm tr}\underline{\chi}\right)\Omega^{-1}\hat{\chi}_{AB} -2 \left(\Omega\underline{\omega}\right)\Omega^{-1}\hat{\chi}_{AB} - \left(\slashed{\nabla}\hat{\otimes}\eta\right)_{AB} 
\\ \nonumber &\qquad -\left(\eta\hat{\otimes}\eta\right)_{AB} + \frac{1}{2}\left(\Omega^{-1}{\rm tr}\chi\right)\Omega\hat{\underline{\chi}}_{AB},
\\ \nonumber {\rm Ric}\left(\Omega e_3,e_A\right) &= -\frac{v}{u}\mathcal{L}_{\partial_v}\eta_A  + \mathcal{L}_b\eta_A +\eta_A\left(\Omega{\rm tr}\underline{\chi}\right) + 4\slashed{\nabla}_A\left(\Omega\underline{\omega}\right) + \slashed{\nabla}^B\left(\Omega\hat{\underline{\chi}}\right)_{AB} -\frac{1}{2}\slashed{\nabla}_A\left(\Omega{\rm tr}\underline{\chi}\right),
\\ \nonumber {\rm Ric}\left(\Omega e_3,\Omega e_3\right) &= -\left(\Omega e_3\right)\left(\Omega{\rm tr}\underline{\chi}\right) -\frac{1}{2}\left(\Omega {\rm tr}\underline{\chi}\right)^2-4\left(\Omega\underline{\omega}\right)\Omega{\rm tr}\underline{\chi} - \left|\Omega \hat{\underline{\chi}}\right|^2.
\end{align*}

The proof is then concluded by using~\eqref{3oij23jioioj} to write ${\rm Ric}\left(\hat{X},\hat{Y}\right)$ for $\hat{X},\hat{Y} \in \left\{\partial_{\hat{v}},\partial_{\hat{u}},\partial_{\hat{\theta}^A}\right\}$ in terms of ${\rm Ric}\left(X,Y\right)$ for $X,Y \in \{\Omega^{-1}e_4,\Omega e_3,e_A\}$, keeping in mind that
\[\Omega^{-1}e_4 = \Omega^{-2}\partial_v,\qquad \Omega e_3 = \partial_u + \partial_{\theta^A},\]
and using the fact proved in Theorem~\ref{2omo2o492} that
\[\left\vert\left\vert \left(\int_0^v\partial_u\left(\Omega^2\right)\, dv,\int_0^v\partial_{\theta^A}\left(\Omega^2\right)\, dc\right) \right\vert\right\vert_{\tilde{C}^j\left(\mathbb{S}^2_{-1,v}\right)} \lesssim |v|^{1-d},\]
for some $|d| \ll 1$. 
\end{proof}
\appendix
\section{Fefferman--Graham Geometries: Regularity and Gauges}\label{feffgrahamreggauge}
In this appendix, we briefly discuss the specialization of our formal expansions to Fefferman--Graham type geometries.
\begin{theorem}Let $g[k]$ be a metric produced by Theorem~\ref{2omo2o492} and assume that $b^{(\rm seed)} = 0$. Then the following all hold:
\begin{enumerate}
	\item We have $\lim_{v\to 0}b = 0$, $\lim_{v\to 0}\left(\Omega\underline{\omega}\right) = 0$, and $\lim_{v\to 0 }\eta = 0$. The metric corresponds to a Fefferman--Graham geometry in that $K$ is a null vector field along $\mathcal{H}_0$.
	\item\label{123456123456} The double-null self-similar coordinate system defined for $\{v > 0\}$, respectively $\{v < 0\}$, in fact extends to $\{v \geq 0\}$, respectively $\{v \leq 0\}$, where the metric is $C^{\tilde{N}}$ for $\tilde{N} \gtrsim N$. 
	\item\label{alternativeyayayay123} The spacetime corresponding to the region $\{v \geq 0\}$ (or $\{v \leq 0\}$) may be covered by an alternative self-similar double-null system where $\Omega|_{v = 0} = 1$ and $\slashed{g}_{AB}|_{v=0} = u^2\mathring{\slashed{g}}_{AB}$, where $\mathring{\slashed{g}}_{AB}$ denotes a round metric on $\mathbb{S}^2$. 
	\item As noted in Remark~\ref{rmkonfg}, when $b^{(\rm seed)} = 0$, then the choice of $\log\Omega^{(\rm seed,\pm)}$ and $\left(\Omega^{-1}\hat{\chi}\right)^{(\rm seed,\pm)}$ are determined by the choice of functions and symmetric trace-free $(0,2)$-tensors corresponding to $\lim_{v\to 0^{\pm}}\log\Omega$ and $\lim_{v\to 0^{\pm}}\left(\Omega^{-1}\hat{\chi}\right)_{AB}$. The compatibility of these choices determine the regularity over $\{v = 0\}$ of the metric $g$ which is obtained by gluing together the two separate double-null metrics defined in $\{v \leq 0\}$ and $\{v \geq 0\}$. We have the following possibilities:
	\begin{enumerate}
		\item\label{couldbenice} If $\lim_{v\to 0^+}\log\Omega^{(\rm seed,+)} = \lim_{v\to 0^-}\log\Omega^{(\rm seed,-)}$ and $\lim_{v\to 0^+}\left(\Omega^{-1}\hat{\chi}\right)^{(\rm seed,+)} = \lim_{v\to 0^-}\left(\Omega^{-1}\hat{\chi}\right)^{(\rm seed,-)}$, then the metric will be $C^{\tilde{N}}$ across $\{v = 0\}$. The estimate~\eqref{3oij2oi829} for the Ricci tensor will holds everywhere in the spacetime.
		\item\label{couldbenice2} If $\lim_{v\to 0^+}\log\Omega^{(\rm seed,+)} = \lim_{v\to 0^-}\log\Omega^{(\rm seed,-)}$ and $\lim_{v\to 0^+}\left(\Omega^{-1}\hat{\chi}\right)^{(\rm seed,+)} \neq \lim_{v\to 0^-}\left(\Omega^{-1}\hat{\chi}\right)^{(\rm seed,-)}$, then the metric will be $C^{0,1}$ across $\{v = 0\}$. The Einstein tensor expressed in $\left(u,v,\theta^A\right)$ and understood weakly will correspond to a locally integrable function for which~\eqref{3oij2oi829} holds.
		\item If $\lim_{v\to 0^+}\log\Omega^{(\rm seed,+)} \neq \lim_{v\to 0^-}\log\Omega^{(\rm seed,-)}$, then it is not possible to glue the metrics in the $\{v \geq 0\}$ and $\{v \leq 0\}$ regions and produce a continuous metric across $\{v = 0\}$. However, in view of item~\ref{alternativeyayayay123} above, one can pick new self-similar double-null coordinate systems in the regions $\{v < 0\}$ and $\{v > 0\}$ so that  $\lim_{v\to 0^+}\log\Omega^{(\rm seed,+)} = \lim_{v\to 0^-}\log\Omega^{(\rm seed,-)} = 1$. Then, one will be in the situation of item~\ref{couldbenice} or~\ref{couldbenice2} above.
					\end{enumerate}
\end{enumerate}
\end{theorem}
\begin{proof}(sketch) We briefly indicate how to show that these various assertions are true:
\begin{enumerate}
	\item Since we have $b|_{v=0} = b^{(\rm seed)}$, we have $b|_{v= 0} = 0$. Then~\eqref{m2om3o2} implies that $\left(\Omega\underline{\omega}\right)^{(\rm seed)}$, and hence $\Omega\underline{\omega}|_{v=0}$ vanishes, and~\eqref{3ioj2oij4oi2} implies that $\eta|_{v=0} = 0$.  The fact that $K$ is null along $\mathcal{H}_0$ then follows from~\eqref{3ij290901}.
	\item In view of the fact that $\left(b,\Omega\underline{\omega},\eta\right)|_{v=0} = 0$, the equations which we use in the proof of Theorem~\ref{2omo2o492} to inductively solve for the expansions of $\log\Omega$ and $\Omega^{-1}\hat{\chi}$ now take the form
	\[v\mathcal{L}_{\partial_v}\left(\log\Omega,\Omega^{-1}\hat{\chi}\right) = \left(H_1,H_2\right),\]
	where $H_1$ and $H_2$ vanish at $\{v = 0\}$. Thus the equations for $\log\Omega$ and $\Omega^{-1}\hat{\chi}$ do not introduce any singular behavior and each term in the expansion will extend to $\{v = 0\}$ in a $C^{\tilde{N}}$ fashion.
	\item We find the new double-null coordinate system in a three-step process:
	\begin{enumerate}
	\item We first observe that in the homothetic gauge constructed by Proposition~\ref{nonononononoo} we will have that $h|_{v=0} = 0$ and $\left(1,\mathcal{L}_{\partial_{\hat{v}}}\right)\left(P-4\rho\right)|_{\hat{v} = 0} = 0$. 
	\item We next observe that the particular homothetic gauge produced by Proposition~\ref{homotheticgaugeexist} depends on the choice of the sphere $\mathcal{S}$ in item~\ref{9090901902190} of Proposition~\ref{homotheticgaugeexist} which will eventually correspond to $\mathcal{S}_{-1}$.  Given any choice of $\mathcal{S}_t$, as a consequence of self-similarity and the uniformization theorem, the metric $\slashed{g}_{AB}|_{v=0}$ takes the form $t^2e^{2\varphi}\mathring{\slashed{g}}_{AB}$ for some round metric $\mathring{\slashed{g}}_{AB}$ and $\varphi : \mathbb{S}^2 \to \mathbb{R}$. Now we  replace the original choice of the sphere $\mathcal{S}$ with the sphere corresponding $t = -e^{-\varphi}$; the induced metric on this new sphere will then exactly be the round metric. The corresponding new homothetic gauge will still satisfy that $h|_{v=0} = 0$ and that $\left(1,\mathcal{L}_{\partial_{\hat{v}}}\right)\left(P-4\rho\right)|_{\hat{v} = 0} = 0$.
	\item With this new homothetic gauge in hand, we then turn to Proposition~\ref{makeitadouble} to define a new double-null coordinate system. The key freedom in  the construction of the coordinate system from the proof of Proposition~\ref{makeitadouble} is the choice of the solution $v$ to the equation~\eqref{3ij2oij42}. In turn, this solution is constructed by an application of Proposition~\ref{solvedegennonlinear}. We first observe that if we add the hypothesis to Proposition~\ref{solvedegennonlinear} that $H$ is a $C^N$ function of $\left(\hat{v},\theta^A\right)$ and that $\mathcal{L}_{\partial_{\hat{v}}}H|_{\hat{v}=0} = 0$, then one may take the solution $v$ produced by Proposition~\ref{solvedegennonlinear} to satisfy the properties that $\hat{v}\left(v,\theta^A\right)$ is a $C^{N-2}$ function in the region $\{v \geq 0\}$ and that $\frac{\partial \hat{v}}{\partial v}|_{v=0} = 1$. In view of the facts that $h|_{v=0} = 0$ and that $\left(1,\mathcal{L}_{\partial_{\hat{v}}}\right)\left(P-4\rho\right)|_{\hat{v} = 0} = 0$, we will exactly have the necessary hypothesis to apply the new version of Proposition~\ref{solvedegennonlinear}. The resulting double-null coordinate system will be regular up to $\{v =0 \}$ and will satisfy that $\Omega|_{v = 0} = 1$ and $\slashed{g}|_{v = 0} = u^2\mathring{\slashed{g}}$. 
	\item The various possibilities for the metric's regularity across $\{v = 0\}$ are straightforward to establish, and the fact that when $\log\Omega^{(\rm seed,+)} = \log\Omega^{(\rm seed,-)}$ we have a suitable weak solution is established by reasoning as in the proof of Proposition~\ref{ioioioioio39289389} (this situation is in fact simpler than the setting of Proposition~\ref{ioioioioio39289389}). 
	\end{enumerate}
		
\end{enumerate} 
\end{proof}


\begin{thebibliography}{9999}

\bibitem[BG16]{critrotate}
T.~W. Baumgarte and C.~Gundlach.
\newblock Critical collapse of rotating radiation fluids.
\newblock {\em Phys. Rev. Lett.}, 116:221103, Jun 2016.

\bibitem[Cho93]{chop}
M.~Choptuik.
\newblock Universality and scaling in gravitational collapse of a massive
  scalar field.
\newblock {\em Phys. Rev. Lett.}, 70:9--12, 1993.

\bibitem[Chr84]{dustnaked}
D.~Christodoulou.
\newblock Violation of cosmic censorship in the gravitational collapse of a
  dust cloud.
\newblock {\em Comm. Math. Phys.}, 93(2):171--195, 1984.

\bibitem[Chr93]{ChristBV}
D.~Christodoulou.
\newblock Bounded variation solutions of the spherically symmetric
  {E}instein-scalar field equations.
\newblock {\em Comm. Pure Appl. Math.}, 46(8):1131--1220, 1993.

\bibitem[Chr94]{ChristNaked}
D.~Christodoulou.
\newblock Examples of naked singularity formation in the gravitational collapse
  of a scalar field.
\newblock {\em Ann. of Math. (2)}, 140(3):607--653, 1994.

\bibitem[Chr09]{Chr}
D.~Christodoulou.
\newblock {\em The formation of black holes in general relativity}.
\newblock EMS Monographs in Mathematics. European Mathematical Society (EMS),
  Z\"urich, 2009.

\bibitem[CK93]{CK}
D.~Christodoulou and S.~Klainerman.
\newblock {\em The global nonlinear stability of the {M}inkowski space},
  volume~41 of {\em Princeton Mathematical Series}.
\newblock Princeton University Press, Princeton, NJ, 1993.

\bibitem[CT71]{cahilltaub}
M.~E. Cahill and A.~H. Taub.
\newblock Spherically symmetric similarity solutions of the {E}instein field
  equations for a perfect fluid.
\newblock {\em Comm. Math. Phys.}, 21:1--40, 1971.

\bibitem[Ear74]{eardleyss}
D.~M. Eardley.
\newblock Self-similar spacetimes: geometry and dynamics.
\newblock {\em Comm. Math. Phys.}, 37:287--309, 1974.

\bibitem[FG85]{FG1}
C.~Fefferman and C.~R. Graham.
\newblock Conformal invariants.
\newblock {\em Ast\'{e}risque}, (Num\'{e}ro Hors S\'{e}rie):95--116, 1985.
\newblock The mathematical heritage of \'{E}lie Cartan (Lyon, 1984).

\bibitem[FG12]{FG2}
C.~Fefferman and C.~R. Graham.
\newblock {\em The ambient metric}, volume 178 of {\em Annals of Mathematics
  Studies}.
\newblock Princeton University Press, Princeton, NJ, 2012.

\bibitem[GHJ21]{nakedeulereinstin}
Y.~Guo, M.~Hadzic, and J.~Jang.
\newblock {N}aked singularities in the {E}instein--{Eu}ler system.
\newblock {\em arxiv:2112.10826, preprint}, 2021.

\bibitem[GMG07]{critsurv}
C.~Gundlach and J.~Martin-Garcia.
\newblock Critical phenomena in gravitational collapse.
\newblock {\em Living Rev. Relativ.}, 10, 2007.

\bibitem[JD92]{nakedfluidjoshidwivedi}
P.~S. Joshi and I.~H. Dwivedi.
\newblock The structure of naked singularity in self-similar gravitational
  collapse.
\newblock {\em Comm. Math. Phys.}, 146(2):333--342, 1992.

\bibitem[KN03]{KN}
S.~Klainerman and F.~Nicol{\`o}.
\newblock {\em The evolution problem in general relativity}, volume~25 of {\em
  Progress in Mathematical Physics}.
\newblock Birkh\"auser Boston Inc., Boston, MA, 2003.

\bibitem[OP90]{nakedfluidoripiran}
A.~Ori and T.~Piran.
\newblock Naked singularities and other features of self-similar
  general-relativistic gravitational collapse.
\newblock {\em Phys. Rev. D (3)}, 42(4):1068--1090, 1990.

\bibitem[RSR18]{scaleinvariant}
I.~Rodnianski and Y.~Shlapentokh-Rothman.
\newblock The asymptotically self-similar regime for the {E}instein vacuum
  equations.
\newblock {\em Geom. Funct. Anal.}, 28(3):755--878, 2018.

\bibitem[RSR19]{nakedone}
I.~Rodnianski and Y.~Shlapentokh-Rothman.
\newblock {N}aked {S}ingularities for the {E}instein {V}acuum {E}quations: The
  {E}xterior {S}olution.
\newblock {\em arXiv:1912.08478, preprint}, 2019.

\bibitem[SR22]{nakedinterior}
Y.~Shlapentokh-Rothman.
\newblock {N}aked {S}ingularities for the {E}instein {V}acuum {E}quations: The
  {I}nterior {S}olution.
\newblock {\em preprint}, 2022.

\end{thebibliography}
\end{document}